\documentclass[a4paper,10pt]{scrartcl}
\usepackage[latin1]{inputenc}
\usepackage{amsmath,amsthm}
\usepackage{amsfonts}
\usepackage{amssymb}
\usepackage{color}
\usepackage{a4wide}
\usepackage{graphicx}
\usepackage{mathabx}
\usepackage{mathrsfs}
\usepackage{subfigure}
\usepackage{pgfplots}
\usepackage{float}
\usepackage{tikz}
\usepackage[authoryear]{natbib}
\usepackage{bbm}
\usepackage{capt-of}
\usetikzlibrary{patterns}
\pgfplotsset{compat=newest}
\newlength\fheight \newlength\fwidth
\newcommand{\R}{\mathbb{R}}
\newcommand{\bk}{\mathbf{k}}
\newcommand{\bl}{\mathbf{l}}
\newcommand{\balpha}{\pmb{\alpha}}
\newcommand{\bxi}{\pmb{\xi}}
\newcommand{\bt}{\mathbf{t}}
\newcommand{\bz}{\mathbf{z}}
\newcommand{\bx}{\mathbf{x}}
\newcommand{\bj}{\mathbf{j}}
\newcommand{\by}{\mathbf{y}}
\newcommand{\bh}{\mathbf{h}}
\newcommand{\bs}{\mathbf{s}}

\newcommand{\bp}{\mathbf{p}}
\newcommand{\bq}{\mathbf{q}}
\newcommand{\be}{\mathbf{1}}
\newcommand{\tail}{\overline{\psi}}
\newcommand{\psf}{k}

\newcommand{\Leb}{\mathrm{Leb}}
\newcommand{\grid}{\log(n)^{\frac{1}{\gamma}}\log\log(n)}
\usepackage{booktabs}
\usepackage{multirow}

\renewcommand{\phi}{\varphi}

\begin{document}

\newtheorem{Theorem}{Theorem}
\newtheorem{Corollary}{Corollary}
\newtheorem{Lemma}{Lemma}
\newtheorem{Assumption}{Assumption}
\newtheorem{Remark}{Remark}
\newtheorem{Definition}{Definition}

\begin{center}
\begin{minipage}{.8\textwidth}
\centering 
\LARGE Multiscale scanning in inverse problems\\[0.5cm]

\normalsize
\textsc{Katharina Proksch}\footnotemark[1]\\[0.1cm]
\verb+kproksc@uni-goettingen.de+\\
Institute for Mathematical Stochastics, University of G\"ottingen\\[0.1cm]

\textsc{Frank Werner}\\[0.1cm]
\verb+frank.werner@mpibpc.mpg.de+\\
Max Planck Institute for Biophysical Chemistry, G\"ottingen, Germany\\[0.1cm]

\textsc{Axel Munk}\\[0.1cm]
\verb+munk@math.uni-goettingen.de+\\
Institute for Mathematical Stochastics, University of G\"ottingen\\
and\\
Felix Bernstein Institute for Mathematical Statistics in the Bioscience, University of G\"ottingen\\
and\\
Max Planck Institute for Biophysical Chemistry, G\"ottingen, Germany
\end{minipage}
\end{center}

\footnotetext[1]{Corresponding author}

\begin{abstract}
In this paper we propose a multiscale scanning method to determine active components of a quantity $f$ w.r.t. a dictionary $\mathcal{U}$ from observations $Y$ in an inverse regression model $Y=Tf+\xi$ with linear operator $T$ and general random error $\xi$. To this end, we provide uniform confidence statements for the coefficients $\langle \varphi, f\rangle$, $\varphi \in \mathcal U$, under the assumption that $(T^*)^{-1} \left(\mathcal U\right)$ is of wavelet-type. Based on this we obtain a multiple test that allows to identify the active components of $\mathcal{U}$, i.e. $\left\langle f, \varphi\right\rangle \neq 0$, $\varphi \in \mathcal U$, at controlled, family-wise error rate. Our results rely on a Gaussian approximation of the underlying multiscale statistic with a novel scale penalty adapted to the ill-posedness of the problem. The scale penalty furthermore ensures weak convergence of the statistic's distribution towards a Gumbel limit under reasonable assumptions. 
The important special cases of tomography and deconvolution are discussed in detail. Further, the regression case, when $T = \text{id}$ and the dictionary consists of moving windows of various sizes (scales), is included, generalizing previous results for this setting.
We show that our method obeys an oracle optimality, i.e. it attains the same asymptotic power as a single-scale testing procedure at the correct scale.
Simulations support our theory and we illustrate the potential of the method as an inferential tool for imaging. 
As a particular application we discuss super-resolution microscopy and analyze experimental STED data to locate single DNA origami.
\end{abstract}


\textit{AMS classification numbers:} Primary 62G10, Secondary 62G15, 62G20, 62G32. \\[0.3cm]
	\section{Introduction}\label{Sec:Intro}
Suppose we have access to observations $Y_{\bj}$ which are linked to an unknown quantity $f \in \mathbb{H}_1$ via the inverse regression model
	\begin{align}\label{model}
	Y_{\bj}=Tf(\bx_{\bj})+\xi_{\bj}, \quad \bj\in I_{n}^d:=\{1,\ldots,n\}^d,\;d\in\mathbb{N}.
	\end{align}
	Here, $T : \mathbb{H}_1 \to \mathbb{H}_2\subset C \left[0,1\right]^d$ is a bounded linear operator acting between proper Hilbert spaces $\mathbb{H}_1$ and $\mathbb{H}_2$. In model \eqref{model}, $n$ stands for the level of discretization such that, more rigorously, the model reads $Y_{\bj,n}=Tf(\bx_{\bj,n})+\xi_{\bj,n}$ with triangular schemes of sampling points $\bx_{\bj}=\bx_{\bj,n}$ in the $d$-cube $\left[0,1\right]^d$ and independent, centered but not necessarily identically distributed random variables $\xi_{\bj}=\xi_{\bj,n},\; \bj\in I_n^d$. For ease of notation, this dependence on $n$ is suppressed whenever it is not relevant. Here and throughout the paper, bold print letters and numbers denote vectors and multi-indices, whereas scalars are printed in regular type face.
	
	Models of the kind \eqref{model} underly a plenitude of applied problems varying from astrophysics and tomography to cell biology \citep[see e.\,g.][]{os86,bbdv09} and have received considerable interest in the statistical literature. Most of research targets (regularized) estimation of $f$ and associated theory.
	An early approach for estimation is based on a singular value decomposition (SVD) of the operator, where $f$ is expanded in a series of eigenfunctions  of $T^*T$ \citep[see e.\,g.][]{Mair,jkpr04,RiskHull,BissantzSiam,Kerkyacharian2010,jp14,Scherzer2016}. Given a proper choice of the regularization parameter, SVD-based estimators are well-known to be minimax optimal \citep{Johnstone}. Adaptive estimation in this context was studied, e.\,g. by \cite{Goldenshluger1999,Tsybakov,cglt03,cg14}.  Since in SVD-based estimation the basis for the expansion is entirely defined by the operator, as an alternative, wavelet-based methods which incorporate the properties of the function of interest have also been frequently employed. Examples are wavelet-vaguelette \citep{Donoho1995} and vaguelette-wavelet methods  \citep{Abramovich}, where $f$ and $Kf$ are expanded in a wavelet and vaguelette basis or vice versa, and the coefficients are estimated by proper thresholding. This allows for a natural adaptation to the local smoothness of the unknown function \citep[see e.\,g.][]{ct02}. Related to this, \cite{Reiss2004} proposed an adaptive estimator based on a combination of linear Galerkin projection methods and adaptive wavelet thresholding. Besides of these selective references a vast amount of work has been devoted to recovery of $f$ during the last decades and the common ground of all these works is that the ill-posedness of an inverse problem usually only gives poor (minimax) rates for estimation and makes full recovery of $f$ a very difficult problem in general (in the setup of \eqref{model} see, e.\,g. \citet{w09} or for deconvolution, see, e.\,g., the monograph by \citet{Meister} and the references given there).
	
	A possibility to deal with this intrinsic difficulty is to relax the ambitious goal of recovering the entire function $f$.
	Indeed, in many applications, only certain properties or aspects of $f$ are of primary interest and a full, precise reconstruction is not necessary any more. Examples of practical relevance are the detection and localization of \textquotedblleft hot spots\textquotedblright\, in astrophysical image analysis \citep{SourceDetection}, functional magnetic resonance imaging \citep{sdt09}, non-destructive testing \citep{Kazantsev2002791}, and image deformation in microscopy \citep{bchm09}, to mention a few. For a theoretical account in deconvolution see \citep{ButuceaCompte}. In a similar spirit, the detection of certain geometric shapes in image analysis has been studied by \citet{Filaments}, but the authors do not take into account the underlying inverse problem. All these issues  can be treated by means of statistical testing, presumably a simpler task than estimation.
	
	In contrast to estimation, hypothesis testing in inverse problems has been investigated much less, early references are \citet{b07,Holzmannetal07}. \citet{iss12} treat the problem of testing $f= 0$ against $f \in \Theta_q \left(r\right)$ where $\Theta_q\left(r\right)$ is a suitable smoothness class restricted to $\left\Vert f\right\Vert\geq r$ by means of the classical minimax testing approach (see e.g. the series of papers by \cite{i93}). Also \citet{llm11,llm12} follow this path and investigate the differences and commonalities of testing in the image space ($Tf = 0$) and the preimage space ($f = 0$). The authors prove that in several situations it does not matter if first $f$ is approximately reconstructed using an SVD-based regularization method and then tested to be $0$, or if $Tf$ is directly tested to be $0$, see also \citet{Holzmannetal07} for a similar observation. More precisely, minimax testing procedures for one of these problems are also minimax for the rephrased problem and the asymptotic detection boundary for both testing problems coincides. For related results in the multivariate setting or for more general regularization schemes see \citet{ilm14,mm14}. In contrast to the problem treated here, in all these studies only ``global'' features of the full signal are investigated, such as testing that the full signal is zero, and no simultaneous inference on sub-structures of the signal is targeted. In fact, this is a much more challenging task in an inverse problems setup and it turns out also to be substantially different to the corresponding direct testing problem of \textquotedblleft hot spot\textquotedblright\, detection. This will be the topic of this paper.
	
	\smallskip
	
	In direct problems ($T=\mathrm{id}$ in  \eqref{model}), finding relevant sub-structures, such as the detection of regions of activity, is of \textquotedblleft scanning-type\textquotedblright, which means that it can be reformulated as a (multiple) testing problem for structures on the grid $I_n^d$ in \eqref{model} and scanning-type procedures can be employed. These have received much attention in the literature over the past decades. \citet{Walther} considers the two dimensional problem of detecting spatial clusters in the Bernoulli model by scanning with rectangular windows of varying sizes, 
	see also \citet{Kabluchko2011} and \citet{ExactScan} for results in a Gaussian setting. 
	%\citep[see also][for results in a Gaussian setting]{Kabluchko2011,ExactScan}.
	In a similar spirit, scan statistics have been employed in the context of multiscale inference about higher order qualitative characteristics such as modes of a  density \citep[see][]{DuemWalt, RufibachWalther,lmsw16,ebdpe17}.
	
	However, in an inverse problem as in \eqref{model}, it is not obvious how to perform statistically efficient \textquotedblleft scanning\textquotedblright~ because local properties of $f$ may propagate in a non-local manner into $Tf$. If, e.\,g., $f$ is a function on $\left[0,1\right]^d$  and we want to infer on the support of $f$, we find that despite the fact that \textit{globally} testing $f \equiv 0$ is equivalent to testing $Tf \equiv 0$, this is not true for localized tests on regions $B \subset \left[0,1\right]^d$ we are interested in here.  This is due to the fact that $\left(Tf\right)_{|_B}$ is not necessarily related to $f_{|_B}$ only. Indeed, we will see that reducing this problem to the image domain $\mathbb{H}_2$, i.\,e., simultaneously testing $H_B : \left(Tf\right)_{|_B} \equiv 0$ against $K_B : \left(Tf\right)_{|_B} >0 $ cannot lead to a competitive procedure as it does not take into account the propagation of (multiscale) features of $f$ by $T$ (cf. Figure \ref{fig:appl_intro}(f)).
	Instead, it becomes necessary to employ probe functionals $\varphi_i=\varphi_{i,n}$ (again dependent on the discretization level $n$, but this dependence will be suppressed whenever not relevant below), which are compatible with the operator $T$ and hence allow for transportation of ``local'' information from $Tf$ back to $\left\langle f, \varphi_{i}\right\rangle$. If the probe functionals $\varphi_{i}$ are chosen properly, the values $\left\langle f, \varphi_{i}\right\rangle$ hold information about ``local'' features of $f$, e.\,g. in form of a wavelet-type analysis, see also \citet{SHMunkDuem2013,ebd16}, who infer on shape characteristics in i.i.d. density deconvolution. 
	\citet{CastroDonoho1} propose a scanning procedure based on a multiscale dictionary of beamlets that allows to detect line segments hidden in a noisy image, however, not in an inverse problems context.
	
	\smallskip
	
	The problem we consider in our paper is as follows: Given model \eqref{model} and an associated  sequence of dictionaries
	\begin{equation}\label{def:U}
	\mathcal{U}=\mathcal{U}_n=\left\{\varphi_{1,n},...,\varphi_{N\left(n\right),n}\right\} \subset R \left(T^*\right),
	\end{equation}
	of cardinality $N = N(n) \to\infty$ as $n\to \infty$, we provide a sequence of multiple tests (``scanning'') for the associated sequence of multiple testing problems
	\begin{equation}\label{hypothesis}
	\left\langle f,\varphi_{i,n}\right\rangle = 0 \qquad\text{for all}\qquad i \in J \tag{$H_{J,n}$}
	\end{equation}
	vs.
	\begin{equation}\label{alternative}
	\exists~i \in J \qquad\text{such that}\qquad \langle\varphi_{i,n},f\rangle>0,
	\tag{$K_{J,n}$}
	\end{equation}
	simultaneously over all subsets $J \subset I_{N(n)}=:\{1,\ldots,N(n)\}$.
	It is clear that the structure of the testing problem stays the same if $\cdot > 0$ in ($K_{J,n}$) is replaced by $ \cdot < 0$ or $\left|\cdot\right| \neq 0$, hence we restrict ourselves to ($K_{J,n})$ in the following. Moreover, it is also clear that as $n\to\infty$, there is a detection boundary, given by a sequence $(\mu_{i,n})_{i\in\mathbb{N}}$, dividing the space of all signals into the asymptotically detectable region and the non-detectable region  such that $\cdot > 0$ will be replaced by   $\cdot > \mu_{i,n}$ later on.
	
	With this choice of a sequence of multiple tests we will not simply control the error of a wrong rejection of $f \equiv 0$, rather we control the \textit{family wise error rate (FWER)} of making any wrong decision, cf. \citet[Def. 1.2]{d14}. Mathematically, our test is a  level-$\alpha$-test for the simultaneous testing problem $H_{J,n}$ against $K_{J,n}$, $J \subset I_N(n)$, i.\,e., it guarantees that
	\begin{equation}\label{eq:FWER}
	\sup\limits_{J: J \subset I_{N(n)}} \mathbb P_{H_{J,n}} \left[\text{''at least one (wrong) rejection in }J\text{``}\right] \leq \alpha +o(1),
	\end{equation}
	as $n$ and hence $N(n)\to\infty.$
	Consequently, all rejections (i.\,e. decisions for signal strength $>0$) will be made at a \textit{uniform} error control, no matter what the underlying configuration of $\langle f, \varphi_{i,n}\rangle$'s is.

	Fundamental to our simultaneous scanning procedure are uniform confidence statements for the coefficients $\left\langle f, \varphi_{i,n}\right\rangle$, $i \in I_{N(n)}$ in the inverse regression model \eqref{model}. Conceptually related, \citet{st12} and \citet{NicklReiss} provide uniform Donsker-type results in the context of i.i.d. deconvolution for single-scale contrasts $\left\langle f, \varphi\right\rangle$, however, not uniform in a sequence of multiscale dictionaries $\mathcal U_n$, a much more challenging task. 
	
	\subsection{\textbf{M}ultiscale \textbf{I}nverse \textbf{SCA}nning \textbf{T}est: MISCAT} 
	
	As we have assumed that $\phi_{i,n}\in\mathcal{R}(T^{\ast})$ for all $i\in I_{N(n)}$, there exists a sequence of  dictionaries $\mathcal{W}=\mathcal{W}_n=\{\Phi_{i,n}\,|\,i \in I_{N(n)}\}\subset\mathbb{H}_2$ such that $\varphi_{i,n} = T^* \Phi_{i,n}$.	
	In the following we will assume that $\mathcal W$ obeys a certain wavelet-type structure, i.e. for each $i \in I_{N(n)}$ there is an associated \textit{scale} $\bh_{i,n} = (h_{i,n,1},\ldots,h_{i,n,d})^T \in \left(0,1\right]^d$ and an associated \textit{translation} $\bt_{i,n}\in[\bh_{i,n},\be]$. The products $\bh_{i,n}^{\be} := h_{i,n,1}\cdot\ldots\cdot h_{i,n,d}$ will be referred to as \textit{sizes of scales}. In contrast to the direct problem ($T = \text{id}$), in an inverse problem the condition $\varphi_i = T^* \Phi_i$ implies a non-standard scaling of the $\Phi_i$'s which can be chosen to depend only on $\bh_i$ and not on $\bt_i$ in many cases. To highlight this scaling property, with a slight abuse of notation, we will also introduce a sequence of dictionary functions $\Phi_{\bh_{i,n}}$ and assume that 
	$\mathcal{W}_n$ is as follows:
	\begin{align}\label{Wavelet}
	\mathcal{W}_n=\biggl\{\Phi_{i,n}(\bz):=\Phi_{\bh_{i,n}}\Bigl(\frac{t_{i,n,1}-z_1}{h_{i,n,1}},\ldots,\frac{t_{i,n,d}-z_d}{h_{i,n,d}}\Bigr)\,\bigg|\,\mathrm{supp}(\Phi_{\bh_{i,n}})\subset[0,1]^d,\;i\in I_{N(n)}\biggr\}.
	\end{align}
	All quantities depend on $n$, and this dependence is suppressed in the following. Note that if $\Phi_{\bh_i} \equiv \Phi$ for all $i \in I_N$, then the dictionary \eqref{Wavelet} is a wavelet dictionary in the classical sense, which is appropriate for direct regression problems, i.e. $T = \text{id}$ in \eqref{model} \citep[see e.\,g.][]{CastroDonoho1}. For our asymptotic results we will further assume that the normed functions $\Phi_{\bh_i} / \left\Vert \Phi_{\bh_i}\right\Vert$ satisfy an average H\"older condition, see \eqref{S1} or \eqref{S2} below. 
	Such conditions are satisfied for many important operators $T$ such as the Radon transform (see Section \ref{Sec:RT}) and convolution operators (see Section \ref{Sec:Dec}). 
	
	\smallskip
	
	To construct a level-$\alpha$-test for simultaneously testing $H_{J,n}$ against $K_{J,n}$, $J \subset I_N$ we can now employ
	\begin{align}\label{PI}
	\langle f,\varphi_i\rangle_{\mathbb{H}_1}=\langle Tf,\Phi_i\rangle_{\mathbb{H}_2}
	\end{align}
	to estimate the local coefficients $\left\langle f,\varphi_i\right\rangle$ by their empirical counterparts
	\begin{align}\label{EmpCo}
	\langle Y,\Phi_i\rangle_n:=\frac{1}{n^d}\sum_{\bj\in I_n^d}Y_{\bj}\Phi_i(\bx_\bj).
	\end{align}
	
	MISCAT combines these local statistics by taking their maximum into a multiple \textquotedblleft dictionary scanning\textquotedblright~ test statistic of the form
	\begin{align}\label{scan2}
	\mathcal{S}(Y):=\max_{i\in I_N} S(Y,i), \quad\text{with}\quad
	S(Y,i):=\omega_i\biggl(\frac{\langle Y,\Phi_i\rangle_n}{\sigma_i}-\omega_{i}\biggr),
	\end{align}
	where $\sigma^2_i:=\mathrm{Var}[\langle Y,\Phi_i\rangle_n]$ depend on the variances $\sigma^2(\bj)$ of the errors $\xi_{\bj}$, which are unknown in general. For simplicity, all results will be stated with known $\sigma_i^2$, as
	all results remain valid
	if the unknown ones are
	replaced by estimates (see Remark \ref{rem:variance_estimation}). The weights
	\begin{align}\label{eq:calibration}
	\omega_i=\omega_{\bh_i}\left(K, C_d\right)=\sqrt{2\log(K/\bh_{i}^{\be})}+C_d\frac{\log(\sqrt{2\log(K/\bh_{i}^{\be})})}{\sqrt{2\log(K/\bh_{i}^{\be})}}
	\end{align}
	provide a proper scale calibration (see Section \ref{Sec:GT}) if $K/\bh_i\geq\sqrt{e}$ for all $i\in I_N$. Since for all results $\max_{i\in I_N}\bh_i\to\mathbf{0}$, this is satisfied for any fixed $K>0$ if $n$ is large enough and we may assume throughout this paper, without loss of generality, that $\min_{i\in I_N}K/\bh_i^{\be}\geq\sqrt{e}$. In this sense, our results hold for any constant $K>0$, however, in many situations $K$ can be chosen such that the weak limit of $\mathcal{S}(Y)$ in \eqref{scan2} is a standard Gumbel distribution (see Remark \ref{BT1}(c) and Theorems \ref{Thm:Radon} and \ref{Thm:Conv1}).
	$C_d$ is an explicit constant only depending on the dimension, the system of scales considered and the degree of $L^2$-smoothness of $\Phi_{\bh_i}$ (see Theorem \ref{Th:1} and Remark \ref{BT1}(b)). Our scale balancing \eqref{eq:calibration} is in line with \citet{DuemSpok2001} and others (but notably different as explained in detail below), who pointed out that, in a multiscale setting, some elements of the dictionary may dominate the behavior of the maximum of a scanning statistic and it is most important to balance all local tests on the different scales in order to obtain good overall power, i.e. a scale dependent correction is necessary. 	
	MISCAT now selects all probe functionals $\Phi_{i,n}$ as \textquotedblleft active\textquotedblright, where $\mathcal{S}(Y,i)$ is above a certain (universal) threshold, which guarantees \eqref{eq:FWER}, to be specified now.
	To this end, notice that in \eqref{eq:FWER} we have
	\begin{equation}
	\sup\limits_{J: J \subset I_N} \mathbb P_{H_{J,n}} \left[\text{''at least one rejection in }J\text{``}\right] \leq \mathbb{P}_0\left[\text{''at least one rejection in }I_{N(n)}\text{``}\right],
	\end{equation}
	where $\mathbb{P}_0=\mathbb{P}_{0,n}= \mathbb P_{H_{I_{N(n)},n}}$, corresponding to $f \perp \mathcal U_n$. The reason for this is that the chance of a false positive among a selection of possible false positives is highest if this selection is as large as possible and all positives are false. Therefore, in order to control the FWER, we only need a universal global threshold $q_{1-\alpha}$ such that $\mathbb P_{0} \left[\mathcal S \left(Y\right) > q_{1-\alpha}\right] \leq \alpha$. To obtain this universal threshold $q_{1-\alpha}$ we will determine the $\mathbb{P}_0$-limiting distribution of $\mathcal{S}(Y)$
	under a general moment condition including many practically relevant models. Theorem \ref{Th:1}(a) in Section \ref{Sec:GT} provides a distribution free (i.\,e. independent of any unknown quantities such as $f$) limit, which is obtained as an almost surely bounded Gaussian approximation for the scan statistic \eqref{scan2} by replacing the errors by a standard Brownian sheet $W,$ i.\,e.
	\begin{align}\label{scan3}
	\mathcal{S}(W):=\max_{i\in I_N}S(W,i),\quad\text{with}\quad 
	S(W,i):=\omega_{i}\biggl(\frac{|\int \Phi_i(\mathbf{z})\,\mathrm dW_{\mathbf{z}}|}{\|\Phi_i\|_2}-\omega_{i}\biggr).
	\end{align}
	Since $\mathcal{S}(W)$ does not depend on any unknown quantities, it can be used to simulate $q_{1-\alpha}$. Exploiting the specific and new choice of calibration in \eqref{eq:calibration} we will furthermore show in Theorem \ref{Th:1}(b) that $\mathcal{S}(Y)$ convergences in distribution towards a Gumbel limit for a wide-range of dictionary functions $\Phi_i$. As $\mathcal S(Y)$ can be seen as a maximum over extreme value statistics of different scales, it follows that the contributions of the different scales are balanced in an ideal way. This result is remarkable, as it provides a general recipe how to calibrate multiscale statistics depending on the degree of smoothness of the probe functionals $\Phi_i$ and the system of scales considered. To best of our knowledge, this is new even in $d = 1$, and in addition, it generalizes results by \citet{ExactScan} to other systems than rectangular scanning (see Remark \ref{BT1}), and to inverse problems and non-Gaussian errors. Note that the calibration proposed by \citet{DuemSpok2001} for direct regression problems (which is frequently employed in multiscale procedures, see e.g. \citet{Rohde2008,Walther,SHMunkDuem2013,ebd16}) is tailored to a continuous observation setting in which all scales within a range $(0,a],\,a\in\R^+$ are considered. If this calibration is used in a discrete setting like \eqref{model}, the overall test-statistic converges to a degenerate limit, since the largest scale $h_{\max}$ has to satisfy $h_{\max}\to0$ as $n\to\infty$, otherwise the finite sample approximations do not converge to their continuous counterparts. Therefore, we propose a different scale calibration which also takes into account the ill-posedness and yields a proper weak limit in many of such cases.
	
	The approximation in \eqref{scan3} requires a coupling technique to replace the observation errors by i.i.d. Gaussian random variables. To this end we do not make use of strong approximations by KMT-like constructions (see \citet{KMT} for the classical KMT results and, e.\,g. \citet{Rio}  or \citet{DMR:2014} for generalizations) as, for instance, \citet{SHMunkDuem2013} in the univariate case, $d=1$, but we take a different route and employ a coupling of the supremum based on recent results by \citet{CheCheKat2014}. Doing so, we can prove the approximation in \eqref{scan3} to hold for a much larger range of scales.
	
	A major benefit of MISCAT is its wide range of applicability and its multiscale detection power. Given the operator $T$, one chooses a dictionary $\mathcal U$ of probe functionals as in \eqref{def:U} such that $\mathcal W$ is of the form \eqref{Wavelet}. We will demonstrate this for the case of $T$ being the Radon transform in Section \ref{Sec:RT} and for $T$ being a convolution operator in Section \ref{Sec:Dec}. For the latter situation we will also discuss an optimal choice of the probe functionals $\varphi_i$. Once the dictionaries $\mathcal U$ and $\mathcal W$ have been obtained, the quantiles $q_{1-\alpha}$ from the Gaussian approximation \eqref{scan3} or its finite sample analogues can be simulated. As it is well-known that convergence towards the Gumbel limit is extremely slow, it is beneficial that for deconvolution we find that the limit only depends on the degree of smoothness (see Theorem \ref{Thm:Conv2}), and hence the finite sample distribution can be pre-simulated in a universal manner.  
	
	We will show in Section \ref{Sec:Opt} that the power of MISCAT asymptotically coincides with the power of a single-scale oracle test which knows the correct size of the unknown object beforehand. More generally, if prior scale information is available, our method can be adapted immediately to this situation by restricting \eqref{scan2} to this subset, which mat lead to different calibration constants in \eqref{eq:calibration} (see  Remark \ref{BT1}(b)). This will further increase detection power in finite sample situations. 
	
	\subsection{MISCAT in action: Locating fluorescent markers in STED super-resolution microscopy}\label{intro:appl}
	In Section \ref{Sec:Dec}, we specify and refine our results to deconvolution which is applied  to a data example from nanobiophotonics in Section \ref{Application} which we briefly review in the following.
	Suppose that the operator $T$ is a convolution operator  having a kernel $\psf$ such that
	\begin{equation}\label{eq:conv}
	\left(Tf\right) \left(\by\right) = \left(\psf \ast f\right) \left(\by\right) = \int\limits_{\mathbb R^d} \psf\left(\bx-\by\right) f\left(\by\right) \,\mathrm d \by.
	\end{equation}
	In our subsequent application the convolution kernel $\psf$ corresponds to the point spread function of a microscope and the object of interest, $f$, is an image such that $d=2$. We assume that $\psf$ is finitely smooth, which is equivalent to a polynomial decay of its Fourier coefficients. In this situation, we may choose the dictionaries $\mathcal U$ and $\mathcal{W}$ such that each $\varphi_i \geq 0$ has compact support $\text{supp} \left(\varphi_i\right) \subset \left[\bt_i-\bh_i,\bt_i\right]$. Consequently, if $f \geq 0$, we find
	\begin{equation}\label{eq:support_inference}
	\left\langle f, \varphi_i\right\rangle > 0 \qquad\Rightarrow\qquad \exists \bx\in \left[\bt_i-\bh_i,\bt_i\right] \quad\mathrm{s.t.}\quad f(\bx)>0,
	\end{equation}
	i.\,e., there must be a point $\bx \in \left[\bt-\bh,\bt\right]$ belonging to the support of $f$. Employing this, we can use MISCAT to segment $f$ into active and (most likely) inactive parts, which is of particular interest in many imaging modalities.
	
	With this setup, MISCAT will be used to infer on the location of fluorescent markers in DNA origami imaged by a super-resolution STED microscope \citep[cf.][]{H:07}. In STED microscopy, the specimen is illuminated by a laser beam along a grid with a diffraction-limited spot centered at the current grid point and the entire specimen is scanned this way, pixel by pixel, leading to observations as in \eqref{model} with a convolution $T$ as in \eqref{eq:conv}. The error distribution and the kernel $\psf$ in \eqref{eq:conv} are well-known experimentally, see Appendix \ref{appB} for a detailed description of the mathematical model.
	
	The investigated specimen consists of DNA origami, which have been designed in a way such that each of the clusters contains up to $24$ fluorescent markers, arrayed in two strands of up to $12$ having a distance of $71$ nanometers (nm) (cf. the sketch in the upper left of Figure \ref{fig:exp_data}). As the ground truth is basically known, this serves as a real world phantom. Data were provided by the lab of Stefan Hell of the Department of NanoBiophotonics of the Max Planck Institute for Biophysical Chemistry, cf. Figure \ref{fig:exp_data}.
	
	\begin{figure}[!htb]
		\centering
		\setlength\fheight{6cm} \setlength\fwidth{6cm}
		\input{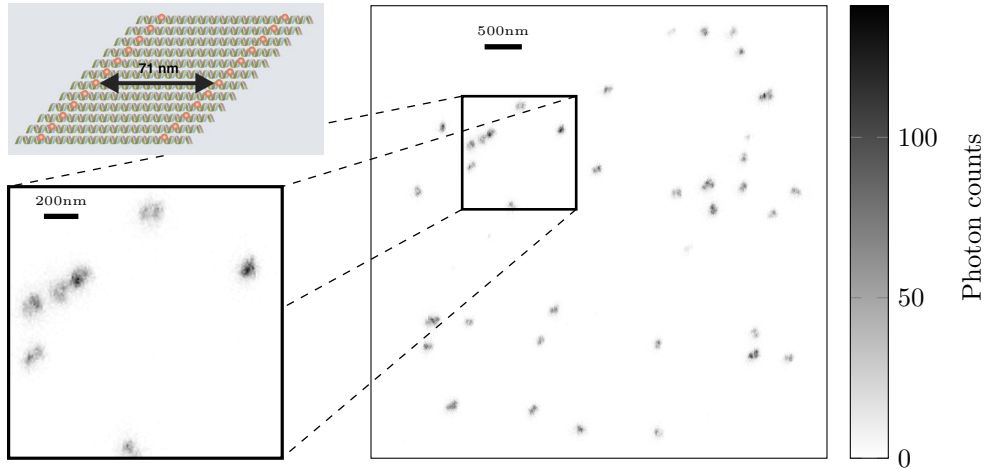} 
		\caption{Experimental data of the DNA origami sample and zoomed region (150$\times$150 pixels). The sketch in the upper left shows the structure of the investigated DNA origami sample (red dots represent possible positions for fluorophores) see \citep{t15}).}
		\label{fig:exp_data}
	\end{figure}
	
	To infer on the positions of the fluorescent markers, we apply MISCAT with a set of scales defined by boxes of size $k_x \times k_y$ pixels, $k_x, k_y = 4,6,...,20$. One pixel in the measurements in Figure \ref{fig:exp_data} is of size $10$ nm $\times$ $10$ nm. 
	To highlight our multiscale approach we also display results of a single scale version of MISCAT (see Remark \ref{BT1}(b) and Section \ref{Application}) using only boxes of size $4 \times 6$ pixels (these are the smallest boxes found by MISCAT), and to highlight the deconvolution effect, we apply a direct multiscale  scanning test not designed for deconvolution (i.e. $T = \text{id}$ in the model \eqref{model} and $\Phi_i = \varphi_i$ in \eqref{scan2}) based on indicator functions as probe functionals using the scale calibration suggested by \citet{DuemSpok2001}. For details see Section \ref{Application}.
	
	In Figure \ref{fig:appl_intro} the zoomed region of Figure \ref{fig:exp_data} is shown together with \textit{significance maps} for all three tests. The significance map color-codes for each pixel the smallest scale (volume of the box in $\mathrm{nm}^2$) on which it is significant. In case that a pixel belongs to significant boxes of different scales, only the smallest one is displayed for ease of visualization by the color coding.
	\begin{figure}[!htb]
		\setlength\fheight{4cm} \setlength\fwidth{4cm}
		\centering
		\begin{tabular}{cccl}
			\begin{tikzpicture}

\begin{axis}[
width=\fwidth,
height=\fheight,
axis on top,
scale only axis,
trim axis left,
trim axis right,
xmin=0.5,
xmax=151.5,
y dir=reverse,
ymin=0.5,
ymax=151.5,
xtick = \empty,
ytick = \empty
]
\addplot [forget plot] graphics [xmin=0.5,xmax=151.5,ymin=0.5,ymax=151.5] {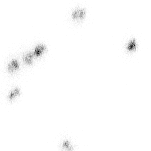};
\draw[solid, draw=black] (axis cs:65.5,0.5) rectangle (axis cs:95.5,30.5) node[below left] {(1)};
\draw[solid, draw=black] (axis cs:118.5,32.5) rectangle (axis cs:148.5,62.5) node[below left] {(2)};
\end{axis}
\end{tikzpicture}
			&
			% This file was created by matlab2tikz v0.4.7 running on MATLAB 7.11.
% Copyright (c) 2008--2014, Nico Schlömer <nico.schloemer@gmail.com>
% All rights reserved.
% Minimal pgfplots version: 1.3
% 
% The latest updates can be retrieved from
%   http://www.mathworks.com/matlabcentral/fileexchange/22022-matlab2tikz
% where you can also make suggestions and rate matlab2tikz.
% 
\begin{tikzpicture}

\begin{axis}[%
width=\fwidth,
height=\fheight,
axis on top,
scale only axis,
trim axis left,
trim axis right,
xmin=0.5,
xmax=31.5,
ymin=0.5,
ymax=31.5,
xtick = \empty,
ytick = \empty
]
\addplot [forget plot] graphics [xmin=0.5,xmax=31.5,ymin=0.5,ymax=31.5] {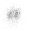};
\end{axis}
\end{tikzpicture}%
			&
			% This file was created by matlab2tikz v0.4.7 running on MATLAB 7.11.
% Copyright (c) 2008--2014, Nico Schlömer <nico.schloemer@gmail.com>
% All rights reserved.
% Minimal pgfplots version: 1.3
% 
% The latest updates can be retrieved from
%   http://www.mathworks.com/matlabcentral/fileexchange/22022-matlab2tikz
% where you can also make suggestions and rate matlab2tikz.
% 
\begin{tikzpicture}

\begin{axis}[%
width=\fwidth,
height=\fheight,
axis on top,
scale only axis,
xmin=0.5,
xmax=31.5,
trim axis left,
trim axis right,
ymin=0.5,
ymax=31.5,
xtick = \empty,
ytick = \empty
]
\addplot [forget plot] graphics [xmin=0.5,xmax=31.5,ymin=0.5,ymax=31.5] {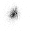};
\end{axis}
\end{tikzpicture}%
			&
			\begin{tikzpicture}
\begin{axis}[
hide axis,
scale only axis,
height=0pt,
width=0pt,
colormap={mymap}{[1pt] rgb(0pt)=(1,1,1); rgb(63pt)=(0,0,0)},
colorbar,
colorbar/width = .4cm,
colorbar style={
width = .4cm,
height = .9\fheight,
scaled ticks = false
},
point meta min=0,
point meta max=141
]
\addplot [draw = none] coordinates {(0,0)};
\end{axis}
\end{tikzpicture} \\
			
			(a) data
			&
			(b) zoomed data, region (1)
			&
			(c) zoomed data, region (2) 
			&
			\\
			
			\begin{tikzpicture}

\begin{axis}[
width=\fwidth,
height=\fheight,
axis on top,
scale only axis,
trim axis left,
trim axis right,
xmin=0.5,
xmax=151.5,
y dir=reverse,
ymin=0.5,
ymax=151.5,
xtick = \empty,
ytick = \empty
]
\addplot [forget plot] graphics [xmin=0.5,xmax=151.5,ymin=0.5,ymax=151.5] {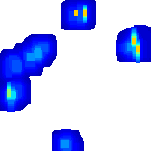};
\end{axis}

\end{tikzpicture}
			&
			% Minimal pgfplots version: 1.3
% 
% The latest updates can be retrieved from
%   http://www.mathworks.com/matlabcentral/fileexchange/22022-matlab2tikz
% where you can also make suggestions and rate matlab2tikz.
% 
\begin{tikzpicture}

\begin{axis}[%
width=\fwidth,
height=\fheight,
axis on top,
scale only axis,
xmin=0.5,
xmax=151.5,
trim axis left,
trim axis right,
ymin=0.5,
ymax=151.5,
xtick = \empty,
ytick = \empty
]
\addplot [forget plot] graphics [xmin=0.5,xmax=151.5,ymin=0.5,ymax=151.5] {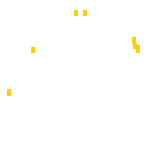};

\end{axis}
\end{tikzpicture}%
			&
			% Minimal pgfplots version: 1.3
% 
% The latest updates can be retrieved from
%   http://www.mathworks.com/matlabcentral/fileexchange/22022-matlab2tikz
% where you can also make suggestions and rate matlab2tikz.
% 
\begin{tikzpicture}

\begin{axis}[%
width=\fwidth,
height=\fheight,
axis on top,
scale only axis,
xmin=0.5,
xmax=151.5,
trim axis left,
trim axis right,
ymin=0.5,
ymax=151.5,
xtick = \empty,
ytick = \empty
]
\addplot [forget plot] graphics [xmin=0.5,xmax=151.5,ymin=0.5,ymax=151.5] {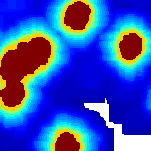};
\end{axis}
\end{tikzpicture}% 
			& 
			\begin{tikzpicture}
\begin{axis}[
hide axis,
scale only axis,
height=0pt,
width=0pt,
colormap={mymap}{[1pt] rgb(0pt)=(1,1,1); rgb(1pt)=(0,0,0.625); rgb(7pt)=(0,0,1); rgb(23pt)=(0,1,1); rgb(39pt)=(1,1,0); rgb(55pt)=(1,0,0); rgb(63pt)=(0.5,0,0)},
colorbar,
colorbar/width = .4cm,
colorbar style={
ytick={0,0.0125,0.025,0.0375,0.05,0.0625},
yticklabels={{},{8000},{4000},{2667},{2000},{1600}},
scaled ticks = false,
width = .4cm,
height = .95\fheight
},
point meta min=0,
point meta max=0.0625
]
\addplot [draw = none] coordinates {(0,0)};
\end{axis}
\end{tikzpicture} \\
			
			(d) MISCAT
			&
			(e) single scale deconvolution test
			& 
			(f) multiscale direct test
			&
		\end{tabular}
		\caption{Experimental data and corresponding $90\%$ significance maps computed by different tests. The color-coding of the significance maps always show the size of smallest significance in nm$^2$, cf. the main text. 
			(a)--(c) data and zoomed regions, (d) MISCAT, (e) a single scale test with deconvolution, (f) a multiscale scanning test without deconvolution. We emphasize that MISCAT performs $2.125.764$ tests on the data in (a), and out of those $94.824$ local hypotheses are rejected. The FWER control ensures that with (asymptotic) probability at least $90\%$ among the selected regions there is no wrong detection.}
		\label{fig:appl_intro}
	\end{figure}
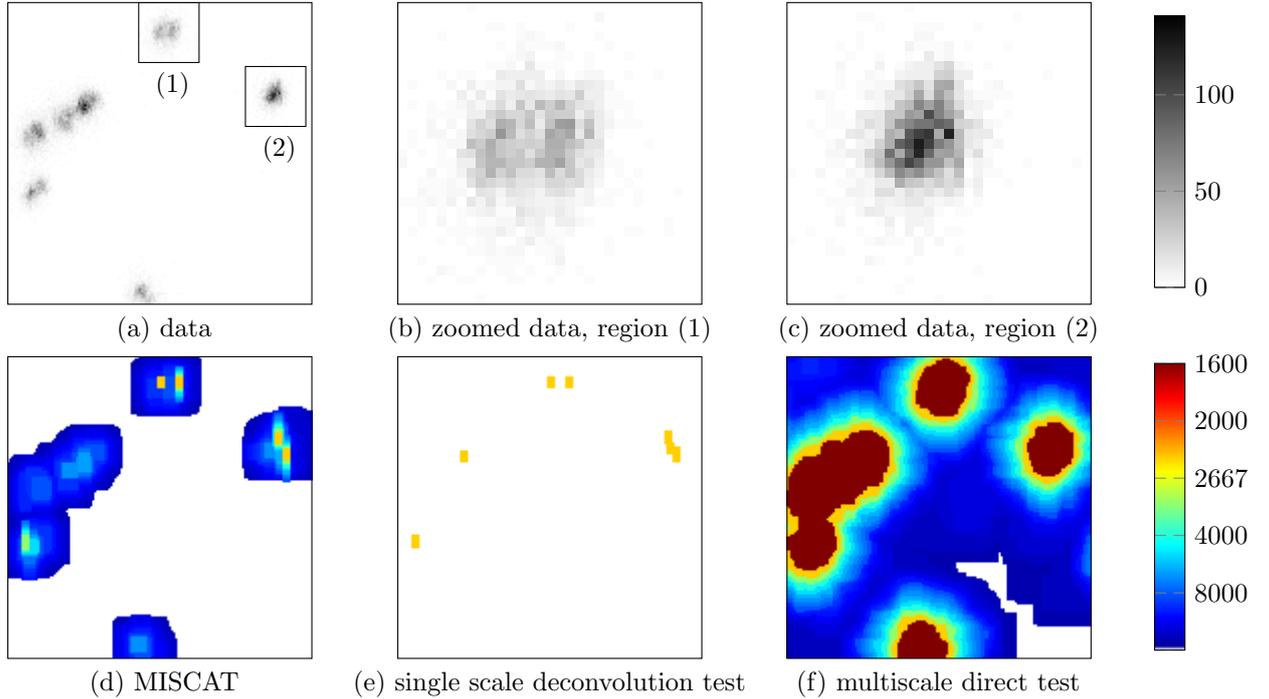
	For instance,  in Figure \ref{fig:appl_intro}(d) MISCAT marked several boxes as significant, and the smallest scale on which significant boxes were found is of size $2400\,\text{nm}^2$ (yellow).
	These results show that MISCAT is able (at least for some of the single DNA origamis) to distinguish both strands. In view of the zoomed data in Figure \ref{fig:appl_intro}(b) and Figure \ref{fig:appl_intro}(c) this is quite remarkable as not visible from the data. The latter is due to the fact that the distance between the two strands of $71\mathrm{nm}$ is slightly smaller than the full width at half maximum (FWHM, see Appendix \ref{appA} for details) of the convolution kernel $k$ ($\approx 76 \mathrm{nm}$), and there is a common understanding that objects which are closer to each other than a distance of approximately the FWHM cannot be identified as separate objects. Hence, MISCAT allows to discern objects below the resolution level of the STED microscope. The single scale variant of MISCAT (for explanation see Section 4) in Figure \ref{fig:appl_intro}(e) has clearly more power in detecting small features on this single scale. While the multiscale test detects 4 boxes of $4\times 6$ pixels, the single scale test detects several more, however, at the price of overseeing many DNA origamis at different scales. Note that the investigated specimen consists only of structures, which are present on a few (known) scales. For illustrative purposes, MISCAT, as employed here, does not use this information, as in general, these scales are not a priori known in living cell imaging. It is also clearly visible in Figure \ref{fig:appl_intro}(f) that ignoring the deconvolution does not lead to a competitive test: distinguishing between different DNA origamis fails completely, as the support of the DNA-origami has been severely blurred by the STED microscope. We emphasize that the FWER control in \eqref{eq:FWER} with $\alpha = 0.1$ implies that with (asymptotic) probability $\geq 90 \%$, each of the 94.824 detections out of 2.125.764 local tests in Figure \ref{fig:appl_intro}(d) is correct.
	
	\section{General Theory}\label{Sec:GT}
	\subsection{Framework and Notation}\label{Sec:GT1}
	Recall the general  framework introduced in Section \ref{Sec:Intro} and model \eqref{model} and that all quantities may depend on the sample size $n.$
	Throughout this paper,
	$\{Tf(\bx_{\bj,n})\,|\,\bj\in I_n^d\}$ is the discretization of the function $Tf$ on the grid $\{(j_1/n,\ldots,j_d/n)\,|\,1\leq j_k\leq n,\;1\leq k\leq d\}.$ This discretization model is a prototype for many inverse problems and in particular matches the application to imaging considered in Section \ref{Sec:Appl} below. For different applications alternative discretization schemes may be of interest as well but, for the sake of a clearer display, we consider uniform sampling on a complete grid since most of the results presented below do not crucially depend on the specific discretization. We make the following assumption on the dictionaries $\mathcal U$ and $\mathcal W$ in \eqref{def:U} and \eqref{Wavelet}.
	\begin{Assumption}\label{Ass:Dict} 
		Let $\mathcal U$ as in \eqref{def:U} and $\mathcal W$ as in \eqref{Wavelet}.
		\begin{description}
			\item[(a) Dictionary source condition] Let
			\begin{align}\label{A1}
			\varphi_i\in \mathcal{R}(T^*),\quad\text{i.\,e.,}\quad \varphi_i=T^*\Phi_i.\tag{DSC}
			\end{align}
			\item[(b) Growth of the dictionary] For some $\kappa>0$ 
			\begin{align}\label{growth}
			|\mathcal{U}|=|\mathcal{W}|=N=O(n^{\kappa}).\tag{G}
			\end{align}
			\item[(c) Scale restrictions] For $\bh_{\min}=(h_{\min},\ldots,h_{\min})^T$ and $\bh_{\max}=(h_{\max},\ldots,h_{\max})^T$ the smallest and the largest scale in \eqref{Wavelet}, respectively,
			\begin{align}\label{R1}
			h_{\min}\gtrsim n^{-1}\log(n)^{15/d\vee3}\log\log(n)^2\quad\text{and}\quad h_{\max}=o\bigl(\log(n)^{-2}\bigr).\tag{SR}
			\end{align}
			\item[(d) Average H\"older condition] Suppose that $\Phi_{\bh_i}$ in \eqref{Wavelet} is uniformly bounded, supported on $[0,1]^d$, vanishing at the boundary and 
			\begin{align}\label{S1}
			\int|\Phi_{\bh_i}(\bt-\bz)-\Phi_{\bh_i}(\bs-\bz)|^2\,\mathrm{d}\bz\leq L\|\bt-\bs\|_{2}^{2\gamma} \|\Phi_{\bh_i} \|_2^2,\tag{AHC}
			\end{align}
			for some $\gamma\in(0,1]$ and all $i \in I_N$ uniformly as $n$ and hence $N \to \infty$.
		\end{description}
	\end{Assumption}
	\begin{Remark}
		\begin{enumerate}
			\item Assumption \eqref{A1} is a smoothness condition on the functions of the dictionary $\mathcal{U}$ related to $T$. Instead of posing such an assumption on the dictionary, it is common to pose such an assumption on $f$, e.\,g. the so-called \textit{benchmark source condition} $f \in \mathcal{R} \left(T^*\right)$, which requires the unknown solution $f$ to be at least as smooth as any function in the range of $T^*$. For deconvolution problems with real-valued kernel this means that $f$ is at least as smooth as the kernel itself. 
			In this paper, as we want to reconstruct pairings $\left\langle f, \varphi_i\right\rangle$ instead of $f$, we may relax this and pose conditions on the functions $\varphi_i$ instead of $f$, see also \citep{bfh13}. Note, that if additionally $f$ admits a sparse representation w.r.t. the dictionary $\mathcal U$, then \eqref{A1} implies $f \in \mathcal{R}\left(T^*\right)$. We emphasize that our approach strongly relies on the condition \eqref{A1}, see also \citet{Donoho1995,a86}. For a strategy how to estimate a linear functional $\left\langle f, \varphi\right\rangle$ for $\varphi \notin R \left(T^*\right)$ we refer to \citet{mp02}.
			\item Assumption \eqref{growth} is rather mild. In particular it implies that positions and scales $(\bt_i,\bh_i)$ from any grid of polynomial size can be used. In the example of imaging this is naturally satisfied as the $\bt_i$ are grid points of the pixel grid and the sizes of the scales $\bh_i$ are given by rectangular groups of pixels and are hence also only of polynomial order in $n$. Furthermore, to serve as an approximation for a continuous version, the grid can be chosen sufficiently fine and still \eqref{growth} is satisfied. The constant $\kappa$ only enters into our results via some constants.
			\item As already discussed in the introduction, the scale restrictions \eqref{R1} are also rather mild. The lower bound on $h_{\min}$ is up to a poly-log factor of the same order as the sampling error, and the upper bound on $h_{\max}$ is required to ensure asymptotic unbiasedness of our local test statistics. For some of the results presented below, a slightly stricter bound on $h_{\max}$ will be necessary, and this is emphasized in the corresponding theorems.
			\item Assumption \eqref{S1} is a smoothness condition on the dictionary $\mathcal{W}$. It is satisfied, for instance, if all $\Phi_{\bh_i}$ are H\"older-continuous of order $\gamma$. In case $T = \text{id}$, the 'classical' scanning function $\Phi_{\bh_i} \equiv I_{(0,1)^d}$ satisfies condition \eqref{S1} with $\gamma=1/2$ and $L = d$. In Section \ref{Sec:Ex} we discuss this condition in more detail and show its validity if $T$ is the Radon transform and if $T$ is a convolution operator (see Section \ref{Sec:RT} and Section \ref{Sec:Dec}, respectively).
		\end{enumerate}
	\end{Remark}
		
	The following assumptions concern the noise $\xi_{\bj},\;\bj\in I_n^d$ in model \eqref{model}. 
	\begin{Assumption}\label{Ass:Noise} Let $\xi_{\bj},\;\bj\in I_n^d$ in \eqref{model} be independent and centered random variables.
		Assume  that there exists a function $\sigma\in C^1[0,1]^d$ such that $\mathrm{Var}[\xi_{\bj}]=\sigma^2(\bx_{\bj})$ 
		and 
		\begin{align}\label{Moments}
		\mathbb{E}|\xi_{\bj}|^{2J}\leq\frac{1}{2}J!\,\mathbb{E}\xi_{\bj}^4\qquad\text{for all}\qquad J\geq2.\tag{M1}
		\end{align}
		Assume further that  
		\begin{align}\label{Moments2}
		0<\liminf_{n\to\infty}\inf_{\bj\in I_n^d}\mathbb{E}[|\xi_{\bj}|^{2}]\quad\text{and}\quad \limsup_{n\to\infty}\sup_{\bj\in I_n^d}\mathbb{E}[|\xi_{\bj}|^{4}] <\infty.\tag{M2}
		\end{align}
	\end{Assumption}
	Note that \eqref{Moments} is in fact equivalent to the well-known Cram\'er condition that the moment generating function exists in a small neighborhood of $0$ \citep[cf.][Thm. 1]{l17} and is satisfied by many distributions, including Gaussian and Poisson. The latter is most relevant for our subsequent application. 
	
	\subsection{Asymptotic Theory}\label{Sec:Sub}
	We are now in the  position to provide some general asymptotic properties of MISCAT such as a uniform Gaussian approximation of the test statistic, a.s. boundedness of the simulated quantiles, and weak convergence under further specification of assumptions towards an explicit Gumbel-type distribution. The latter is for ease of presentation only shown when using the full set of possible scales. If MISCAT is restricted to smaller subsets of scales (e.g. resulting from prior information), this may change the limit distribution, see Remark \ref{BT1} below.
	\begin{Theorem}\label{Th:1}
		Suppose we are given observations from model \eqref{model} with random noise satisfying Assumption \ref{Ass:Noise} and dictionaries $\mathcal{U}$ and $\mathcal{W}$ as specified in Assumption \ref{Ass:Dict}.
		Let $h_{\max}\leq n^{-\delta}$ for some (small) $\delta>0$ in \eqref{R1} and suppose that the  approximation error of $\langle \mathbb{E}[Y],\Phi_i\rangle_n:=\frac{1}{n^d}\sum_{\bj\in I_n^d}Tf(\bx_{\bj})\Phi_i(\bx_{\bj})$ is asymptotically negligible, i.\,e.,
		\begin{align}\label{bias}
		n^{\frac{d}{2}}\max_{i \in I_N}\frac{ \langle\mathbb{E}[Y],\Phi_i\rangle_n -\langle Tf,\Phi_i\rangle}{\|\Phi_i\|_2}=o\biggl(\frac{1}{\log(n)^2\log\log(n)^2}\biggr).
		\end{align}
		For any constant $K>0$ and $C_d=2d+d/\gamma-1$ consider the calibration values $\omega_i = \omega_i \left(K, C_d\right)$ as in \eqref{eq:calibration}.
		\begin{enumerate}
			\item Then, for a standard Brownian sheet $W$ on $[0,1]^d$, it holds
			\begin{align*}
			\lim_{n\to\infty}\biggl|\mathbb{P}_{0}\Bigl(\mathcal{S}(Y)\leq q\Bigr)-\mathbb{P}_{0}\Bigl(\mathcal{S}(W)\leq q\Bigr)\biggr|=0,\qquad q \in \mathbb R
			\end{align*}
			where $\mathcal{S}(Y)$ and $\mathcal{S}(W)$ are defined in \eqref{scan2} and \eqref{scan3}, respectively. Consequently, under $H_0$, $\mathcal S \left(Y\right)$ and $\mathcal S \left(W\right)$ converge weakly towards the same limit. Furthermore, the approximating statistic $\mathcal{S}(W)$ is almost surely bounded and does not depend on any unknown quantity.
			
			\medskip
			
			\item Instead of \eqref{S1} assume the stronger condition that there exists a function $\Xi$ supported on $\left[0,1\right]^d$ with $\left\Vert \Xi \right\Vert_2 = 1$ such that
			\begin{equation}\label{eq:Phi_h_convergence}
			\max\limits_{i \in I_N} \left| \int \left(\frac{\Phi_{\bh_i} \left(\bt_i -\bz\right)}{\left\Vert \Phi_{\bh_i} \right\Vert_2}- \Xi \left(\bt_i-\bz\right)\right) \,\mathrm d W_{\bz} \right| = o_{\mathbb P} \left(\frac{1}{\sqrt{\log\left(n\right)}} \right)
			\end{equation}
			and
			\begin{align}\label{S2}
			\int\left|\Xi\left(D_{\Xi}(\bt-\bz)\right)-\Xi\left(D_{\Xi}(\bs-\bz)\right)\right|^2\,\mathrm d\bz= \sum_{j=1}^d|t_j-s_j|^{2\gamma}+o\biggl(\sum_{j=1}^d|t_j-s_j|^{2\gamma}\biggr)\tag{AHCb}
			\end{align}
			with $\gamma \in \left(0,1\right]$ and a symmetric, positive definite matrix $D_\Xi \in \mathbb R^{d \times d}$. Suppose that  the set of scales $\mathcal H := \left\{\bh_i ~\big|~ i \in I_N\right\}$ is complete, i.e. $\mathcal H = \left\{h_{\min}, ...,h_{\max}\right\}^d$, where 
			\begin{equation}\label{eq:scale_logs}
			-\log(h_{\max})=\delta\log(n)+o(\log(n))\qquad\text{and}\qquad-\log(h_{\min})=\Delta\log(n)+o(\log(n))
			\end{equation}
			with $0 < \delta < \Delta \leq 1$. If the grids of positions $\bt$ and scales $\bh$ are furthermore sufficiently fine, i.e.
			\begin{align} \label{Grid1}
			\max_{i \in I_N}\min_{j \in I_N:\bt_i\neq\bt_j}\|\bt_i-\bt_{j}\|_{\infty}=O\biggl(\frac{1}{n}\biggr)
			\end{align}
			and 
			\begin{align}\label{Grid2}
			\max\limits_{i \in I_N}\min\limits_{\substack{j \in I_N \\ h_{i,l} \neq h_{j,l}}}\left|\frac{h_{j,l} - h_{i,l}}{\sqrt{h_{i,l}h_{j,l}}}\right| \to 0\qquad\text{for all}\qquad 1\leq l \leq d
			\end{align}
			then it holds
			\begin{align}\label{eq:GumbelGW}
			\lim_{n\to\infty}\mathbb{P}_0\biggl(\mathcal{S}(Y)\leq\lambda\biggr)=\exp\left(-\exp\left(-\lambda\right)\cdot\frac{H_{2\gamma}\det\left(D_{\Xi}^{-1}\right) I_d(\delta,\Delta)}{\sqrt{2\pi}K}\right),
			\end{align}
			with
			\begin{align}\label{Def:Id}
			I_d(\delta,\Delta):=\frac{(-1)^{d-1}}{(d-1)!}\sum_{k=0}^d(-1)^k\binom{d}{k}\log\big(k\delta+(d-k)\Delta\big) >0
			\end{align}
			and Pickands' constant $H_{2\gamma}$ \citep[cf.][]{p69}.
		\end{enumerate}
	\end{Theorem}
	
	All proofs will be given in in Section \ref{Proofs}.
	
	\begin{Remark}\label{BT1}~\\
		\vspace{-0.5cm}
		\begin{enumerate}
			\item Assumption \eqref{bias} is a mild assumption on the integral approximation as the required rate is very slow. It is satisfied, in particular, if $Tf$ and $\Phi$ in \eqref{Wavelet} are H\"older-continuous of some order, or if $Tf$ is H\"older-continuous and $\Phi$ is an indicator function. Note that due to the ill-posedness of the problem, $Tf$ being H\"older-continuous does typically not require $f$ to be continuous. 
			\item Although it might seem marginal, a proper choice of the constant $C_d$ is crucial for the boundedness of $\mathcal S(W)$. The choice $C_d = 2d + d/\gamma-1$ used in the formulation of the theorem is adjusted to the case where a dense grid of scales in the sense of \eqref{Grid2} is considered. In particular, this includes the case where  \textbf{all} scales in Assumption \ref{Ass:Dict} \eqref{R1} ranging from $\bh_{\min}$ to $\bh_{\max}$ are used. If now, for instance, $T = \mathrm{id}$ and $\Phi$ in \eqref{Wavelet} is chosen to be the indicator function of $\left[0,1\right]^d$, we have $\gamma = 1/2$ and consequently $C_d = 4d-1$, which coincides with the constant of \citet{ExactScan} for the Gaussian case.
			
			However, in many situations a less dense grid of scales might be of interest, e.g. under prior scale information on the object of interest $f$. Then for the choice $C_d = 2d + d/\gamma-1$ the statistic $\mathcal S \left(W\right)$ is still a.s. bounded from above, but \eqref{eq:GumbelGW} might not be valid anymore. To avoid this, $C_d$ has to be adjusted. Suppose in what follows that the grid of positions still satisfies \eqref{Grid1}. In the least dense regime, when $\mathcal S(W)$ behaves as in a single scale scenario, the proper choice is $C_d = d/\gamma -1$. Another interesting special case is when only squares in a dense range are considered (this is $\bh_i = \left(h_i, ..., h_i\right)$ and \eqref{Grid2} is satisfied), where one should choose $C_d = 1+d/\gamma$.
			
			All these choices of $C_d$ are specified in more detail in Corollary \ref{Cor} in Section \ref{Sec:Aux} and follow from our general result in Theorem \ref{Thm:Gumbel}.
			\item As specified in the theorem, $\mathcal S (W)$ is bounded for any choice of the constant $K>0$. In fact, $K$ does not affect the asymptotic power of MISCAT as it only determines the location of the limiting distribution. For $\gamma\in\{1/2,1\},$ $H_{2\gamma}$ can be computed explicitly \citep[see][]{p69}, i.e. $H_1 = 1$ and $H_2= \pi^{-\frac{d}{2}}$. In this case the choices
			\begin{align*}
			K=
			\begin{cases}
			\frac{|\det D_{\Xi}^{-1}|I_d(\delta,\Delta)}{\sqrt{2\pi}},&\text{if }\gamma=\frac{1}{2}\\
			\frac{|\det D_{\Xi}^{-1}|I_d(\delta,\Delta)}{(2\pi)^{\frac{d+1}{2}}},&\text{if }\gamma=1,
			\end{cases}
			\end{align*}
			yield standard Gumbel limit distributions. If $\gamma =1$ and if the correlation function $r_\Xi$ of the Gaussian field $Z_{\bt}=\int\Xi\bigl(\bt-\bz\bigr)\,\mathrm dW_{\bz}$ is twice differentiable in $\mathbf{0}$, the matrix $D_{\Xi}$ can be computed via $D^*_{\Xi}D_{\Xi}=\mathrm{Hess}_{r_{\Xi}}(\mathbf{0})^{-1}.$ For $T$ being the Radon transform or a convolution operator, this allows us to give explicit constants $K$ in \eqref{eq:CalRadon} and \eqref{eq:CalDec}, respectively, ensuring standard Gumbel limit distributions.
			\item In the situation of Theorem \ref{Th:1} (b) under a  weaker assumption than  \eqref{eq:Phi_h_convergence} and \eqref{S2} it can be shown that the limiting distribution is stochastically bounded by Gumbel distributions and is hence non-degenerate in the limit. This will be done in Theorem \ref{Thm:Conv2} in the situation of deconvolution.  
		\end{enumerate}
	\end{Remark}
	
	\subsection{Statistical Inference}\label{Sec:GT3}
	
	In the following, let $q_{1-\alpha}$ denote the $1-\alpha$-quantile of the approximating process $\mathcal{S}(W)$. To compare the local test statistics $\mathcal S \left(Y,i\right)$ in \eqref{scan2} with $q_{1-\alpha}$, we have assumed so far to know the local variances $\sigma^2_i=\mathrm{Var}[\langle Y,\Phi_i\rangle_n]$. The next Remark shows that they can easily be estimated without changing the limiting distribution of $\mathcal S \left(W\right)$.

	\begin{Remark}\label{rem:variance_estimation}
		As mentioned before, the local variances $\sigma_i^2$, $i \in I_N$, depend on $\text{Var}\left[\xi_{\bj}\right] = \sigma^2(\bx_{\bj})$ (cf. Assumption \ref{Ass:Noise}), $\bj \in I_n^d$, which are typically unknown in applications. Nevertheless, all results remain valid if the $C^1$-function $\sigma^2$ (see Assumption \ref{Ass:Noise}) can be estimated from the data by $\hat\sigma^2$ such that
		\begin{align}\label{eq:var}
		\max_{i \in I_N}\bigl|\hat\sigma^2(\bt_i)-\sigma^2(\bt_i)\bigr|=o_{\mathbb{P}}\left(\log(n)^{-\frac{1}{2}}\right).\tag{V}
		\end{align}
		The local variances $\sigma_i^2$ can then be estimated by $\hat \sigma_i^2 := \left\langle \hat\sigma^2, \Phi_i^2 \right\rangle_n$. Condition \eqref{eq:var} is e.g. satisfied for (suitable) kernel-type estimators or point-wise maximum likelihood estimators as used in Section \ref{Application}.
	\end{Remark}
	
	We conclude by Theorem \ref{Th:1} that
	\begin{align*}
	\lim_{n\rightarrow\infty}\mathbb{P}_{0}\bigl(\mathcal S \left(Y,i\right)\leq q_{1-\alpha}\quad\text{for all}\quad i\in I_N\bigr)\geq1-\alpha,
	\end{align*}
	and hence \eqref{eq:FWER} is valid, i.e. all rejections are significant findings. Conversely, it can be shown that, with overall confidence of approximately $(1-\alpha)\cdot100\%,$ all relevant components are found, provided that the signal is sufficiently strong. 
	\begin{Lemma}\label{Lemma:Power1} 
		Suppose we are given observations from model \eqref{model} with random noise satisfying Assumption \ref{Ass:Noise} and dictionaries $\mathcal{U}$ and $\mathcal{W}$ as specified in Assumption \ref{Ass:Dict}. Let $\mathcal{I}_{\alpha}$ denote the set of all large components, i.\,e.
		\begin{align*}
		\mathcal{I}_{\alpha}:= \left\{i~\big|~\langle\phi_i\,,\,f\rangle>2\biggl(\frac{q_{1-\alpha}}{\omega_i}+\omega_i\biggr)\sigma_i\right\}.
		\end{align*}
		Then, under the assumptions of Theorem \ref{Th:1}
		\begin{align*}
		\lim_{n\rightarrow\infty}\mathbb{P}\bigl(\mathcal S \left(Y,i\right)>q_{1-\alpha}\quad\text{for all}\quad i\in\mathcal{I}_{\alpha}\bigr)\geq1-\alpha
		\end{align*}
	\end{Lemma}
	
	For general $T$ it is not clear if the detection guarantee in Lemma \ref{Lemma:Power1} is optimal in the sense that weaker signals cannot be detected by any procedure. However, in the next subsection we will show that in special situations MISCAT obeys an oracle optimality property.
	
	\subsection{Asymptotic Optimality}\label{Sec:Opt}
	
	For signals built from block signals, the asymptotic power of MISCAT can be computed explicitly  which reveals an oracle optimality property of MISCAT in the following sense. Suppose that $f=\mu_{n,\bh_{\star}}I_{\left[\bt_\star-\bh_\star,\bt_\star\right]}$.
	If one knew the correct scale $\bh_\star$, one would perform a single-scale test in order to find the location $\bt_\star$. Hence, in this idealized situation, the \textquotedblleft oracle scan statistic\textquotedblright\, $\mathcal{S}^\star(Y)$ given by
	\begin{align*}
	\mathcal{S}^\star(Y)=\sup_{i\in I_N}\omega_{\bh_{\star}}\big(K, \tfrac{d}{\gamma}-1\big)\bigg(\sigma_i^{-1}\left\langle Y,\Phi_{\bh_{\star}}\bigg(\frac{\bt_i-\cdot}{\bh_{\star}}\bigg)\right\rangle_n-\omega_{\bh_{\star}}\big(K, \tfrac{d}{\gamma}-1\big)\bigg)
	\end{align*} 
	would be used. Note the different adjustment of weights due to Remark \ref{BT1}(b). It turns out that MISCAT performs as well in terms of its asymptotic power as the oracle test corresponding to $\mathcal S^\star \left(Y\right)$. Moreover, the following theorem guarantees that signals will be detected asymptotically with probability 1, if 
	\begin{align*}
	\mu_{n,\bh}\geq\max_{\bt}\sigma(\bt)(\sqrt{2\log(1/\bh_{\star})}+\beta_n) n^{-\frac{d}{2}} \|\Phi_{i_\star}\|_2,
	\end{align*}
	where $i_\star$ is such that $(\bt_i,\bh_i)=(\bt_\star,\bh_\star)$ and $\beta_n\to\infty$. In this setting, if the errors are i.i.d. standard normal and $T=\mathrm{id}$, the single scale test is minimax optimal if $\|\Phi_i\|_2=\sqrt{\bh_i^{\mathbf 1}}$ \citep[see][]{DuemSpok2001,ChanWalther,k17}. Thus, also the multiscale procedure MISCAT is minimax optimal in this case. If $T\neq\mathrm{id}$, optimality depends on both dictionaries $\mathcal{W}$ and $\mathcal{U}$ and special care has to be put into the choice of dictionary functions. This is discussed in more detail in Section \ref{Sec:OptDec} below.
	
	Theorem \ref{Thm:Power} provides an expansion of the asymptotic power of MISCAT under general noise assumptions. This is a generalization of Theorems 4 and 6 in \citet{ExactScan}. 
	\begin{Theorem}[Asymptotic Power of MISCAT]\label{Thm:Power}
		Suppose we are given observations from model \eqref{model} with random noise satisfying Assumption \ref{Ass:Noise} and dictionaries $\mathcal{U}=\{\phi_i\,|\,\phi_i(\bz)=\phi((\bt_i-\bz)/\bh_i),\,\phi(\bz)>0,\bz\in(0,1)^d \}$ and $\mathcal{W}$ as specified in Assumption \ref{Ass:Dict}.  Suppose \eqref{eq:scale_logs} with $0 < \delta < \Delta \leq 1$ and fix a scale $\mathbf{h}_\star = \bh_{\star}(n) \in \left[\bh_{\min},\bh_{\max}\right]$ and a subset $\mathcal T_\star \subset I_N$ such that $\bh_i = \bh_{\star}$ for all $i \in \mathcal T_{\star}$. Now consider the set of functions $f$ with support given by the union of all corresponding boxes which are sufficiently strong, i.e. 
		\[
		\mathscr{S}_{\mathcal{T}_\star}(\bh_\star, \mu_n):=\Big\{f\,\big|\, \eqref{bias} \text{ holds, }\mathrm{supp}(f)=\bigcup_{i\in\mathcal{T}_\star}[\bt_i-\bh_\star,\bt_i], \quad \left\langle \varphi_i, f\right\rangle \geq \mu_n  \frac{\left\Vert \Phi_{i} \right\Vert_{2}}{n^{d/2}}, i\in\mathcal{T}_\star\Big\}.
		\]
		Assume that $\sigma\in C^{1}([0,1]^d)$ and $\bt_\star\in(0,1)^d$ where $\bt_\star\in\mathrm{argmax}\{\sigma(\bt)\,|\,\bt\in[0,1]^d\}$ and let $K>0$. 
		\begin{enumerate}
			\item If $\{\bh_i\,|\;i \in I_N\}=\{\bh_\star\}$, i.\,e. for each $\bt$ we consider scanning windows of (correct) size $\bh_\star$, then MISCAT with the single-scale-calibration $\omega_i \left(K, d/\gamma-1 \right)$ as in \eqref{eq:calibration} (cf. Remark \ref{BT1}(b)) attains power
			\begin{align*}
			\inf_{f\in \mathscr{S}_{\mathcal{T}_\star}(\bh_\star, \mu_n)}
			\mathbb{P}_{f}\big(\mathcal{S}^\star(Y)>q_{1-\alpha}\big)
			&=\inf_{f\in \mathscr{S}_{\{\bt_\star\}}(\bh_\star, \mu_n)}\mathbb{P}_{f}\big(\mathcal{S}^\star(Y)>q_{1-\alpha}\big)\\
			&=\alpha+(1-\alpha)\cdot\tail\bigg(\sqrt{2\log\left(\tfrac{1}{\bh_\star^{\be}}\right)}-\frac{\mu_n}{\sigma(\bt_{\star})}\bigg)+o(1).
			\end{align*}
			Here and in the following, $\tail \left(x\right) := \int_x^\infty \left(2\pi\right)^{-1/2} \exp\left(-y^2/2\right) \,\mathrm d y$ is the tail function of the standard normal distribution.
			\item In general, MISCAT with the multiscale-calibration $\omega_i \left(K, 2d+d/\gamma-1 \right)$ as in \eqref{eq:calibration} satisfies
			\begin{align}\label{eq:ExpPower}
			\inf_{f\in \mathscr{S}_{\mathcal{T}_\star}(\bh_\star, \mu_n)}
			\mathbb{P}_{f}\big(\mathcal{S}(Y)>q_{1-\alpha}\big)+o(1)
			\geq\inf_{f\in \mathscr{S}_{\mathcal{T}_\star}(\bh_\star, \mu_n)}
			\mathbb{P}_{f}\big(\mathcal{S}^\star(Y)>q_{1-\alpha}\big),
			\end{align}
			i.\,e. the multiscale procedure performs at least as well as the oracle procedure.
		\end{enumerate}
	\end{Theorem}
	
	\section{Examples}\label{Sec:Ex}
	
	\subsection{The $d$-dimensional Radon transform} \label{Sec:RT}
	
	Assume one observes a discretized and noisy sample of the Radon transform of $f$,
	\begin{align}\label{modelRadon}
	Y_{\mathbf{k},l}=Tf(\pmb{\vartheta}_\mathbf{k},u_l)+\xi_{\mathbf{k},l};\quad  u_l=\frac{l-1/2}{n}, \quad l=1,\ldots,n
	\end{align}
	and $\pmb{\vartheta}_\mathbf{k}\in\mathbb{S}^{d-1}, \mathbf k \in I_n^{d-1}$ are design points which are uniformly distributed w.r.t. the angles in a parametrization using polar coordinates,
	where
	\begin{align*}
	Tf(u,\pmb{\vartheta})=\int_{\langle\mathbf{ v},\pmb{\vartheta}\rangle=u}f(\mathbf{v})\,\mathrm d\mu_{d-1}(\mathbf{v})
	\end{align*}
	denotes Radon transformation \citep[cf.][]{n86}, $\,\mathrm d\mu_{d-1}$ denotes the $(d-1)$-dimensional Lebesgue measure on the hyperplane $\{\mathbf{v}\,|\,\langle\mathbf{ v},\pmb{\vartheta}\rangle=u\}$  and $\xi_{\mathbf{k},l}$ are i.i.d., Var$[\xi_{(\be,1)}]=\sigma^2$.  In this case fix $\widetilde\varphi:\R^+\rightarrow\R$, set $\varphi \left(\bx\right) := \widetilde \varphi \left(\|x\|_2\right),$ $\mathrm{supp}(\widetilde \varphi)\subset[0,1]$ and define
	\begin{align}\label{eq:URadon} 
	\mathcal{U}=\left\{\varphi_i=h_i^{-d/2}\varphi\left(\frac{\cdot-\bt_i}{h_i}\right)\,\bigg|\,i \in I_N\right\},
	\end{align} 
	i.e. we consider a dictionary $\mathcal{U}$ of rotationally invariant functions. We now construct the corresponding dictionary $\mathcal{W}$. To this end we need to fix some more notation.
	Let $\mathrm d\pmb{\vartheta}$ denote the common surface measure on $\mathbf{S}^{d-1}$ such that for measurable $S\subset\mathbb{S}^{d-1}$ we have
	\begin{align*}
	|S|=\int_{S}\,\mathrm d\pmb{\vartheta}.
	\end{align*}
	Let further $\mathcal{F}_df$ denote the $d$-dimensional Fourier transform of $f$, defined by
	\begin{align*}
	\mathcal{F}_df(\pmb{\xi})=\int f(\bx)\,\exp\left(\textup{i}\langle\bx,\bxi\rangle\right)\,\mathrm d\bx\quad\text{such that}\quad f(\bx)=\frac{1}{(2\pi)^d}\int \mathcal{F}_df(\bxi)\,\exp\left(-\textup{i}\langle\bxi,\bx\rangle\right)\,\mathrm d\bxi.
	\end{align*}
	\begin{Lemma}\label{Le:WRadon}
		Let $\mathcal{U}$ be as in \eqref{eq:URadon}, $\varphi\in\mathcal{R}(T^*)$. Then
		\begin{align*}
		\mathcal{W}=\biggl\{\Phi_i\,\bigg|\,\Phi_{i}(u,\pmb{\vartheta}) = h_i^{-\frac{d}{2}} \Phi \left(\frac{u-\langle\pmb{\vartheta},\bt_i\rangle}{h_i} \right)\biggr\},
		\end{align*}
		where, due to the rotational invariance of $\varphi$, the function $\Phi$, defined by 
		\begin{align}\label{DefPhiRadon}
		\Phi \left(x\right) := \frac{1}{2(2\pi)^{d}}\mathcal{F}_1\left(\left(\mathcal{F}_d\varphi\right)\left(\cdot\pmb{\vartheta}\right)|\cdot|^{d-1}\right)\left(x\right), \qquad x\in \mathbb R,
		\end{align}
		is independent of $\pmb{\vartheta}$.
	\end{Lemma}
	
	Consequently, the functions $\Phi_{h_i}$ as in \eqref{Wavelet} satisfy
	\[
	\Phi_{h_i}(u,\pmb{\vartheta}) = h_i^{-\frac{d}{2}} \Phi \left(\frac{u-\langle\pmb{\vartheta},\bt_i\rangle}{h_i} \right),
	\] 
	i.e. we have the special structure $\Phi_{h_i}=C_{h_i}\Phi$ and hence we can set $\Phi_{h_i}/\|\Phi_{h_i}\|_{L^2(\R\times\mathbb{S}^{d-1})}=:\Xi$.
	It turns out that \eqref{S1} and \eqref{S2} are satisfied if $\varphi$ is sufficiently smooth. This is made more precise in the following Lemma.
	\begin{Lemma}
		\label{Le:AHCbRadon}
		Let $4\pi\|\mathcal{F}_1\big((\mathcal{F}_d\varphi)(\cdot\pmb{\vartheta})|\cdot|^{d-1}\big)(u-\langle \bt_i,\pmb{\vartheta}\rangle)\|_{L^2(\R\times\mathbb{S}^{d-1})}^{-1}=:C_{\varphi,d}$. If $\varphi\in H^{\frac{d+1}{2}}(\R^d)$, \eqref{S2} holds with
		\begin{align}\label{eq:DXiRadon}
		D_{\Xi}^{-2}:= \mathrm{diag}\biggl(C_{\varphi,d}\int_{\mathbb{R}^d}\omega_1^2\|\pmb{\omega}\|^{d-1}\big|(\mathcal{F}_d\varphi)(\pmb{\omega})\bigr|^2\,\mathrm d\pmb{\omega}\biggr).
		\end{align}
	\end{Lemma}
	
	In general, the dictionary functions $\Phi$ may be of unbounded support.
	In this case the results from Theorem \ref{Th:1} b) remain valid if we exclude a small boundary region from our analysis. Here, we only consider positions $\bt_i\in[\mathbf{0},\be-\pmb{\rho}]$, where $\pmb{\rho}=(\rho,\ldots,\rho)^T,\,\rho>0$ and we obtain the following extreme value theorem for MISCAT in the case of the Radon transform.
	\begin{Theorem}[MISCAT for the Radon Transform]\label{Thm:Radon}
		Suppose that we have access to observations following model \eqref{modelRadon}. Let $\bt_i\in[\mathbf{0},\be-\pmb{\rho}]$, where $\pmb{\rho}=(\rho,\ldots,\rho)^T,\,\rho>0.$
		Assume further that
		the  approximation error of $\langle \mathbb{E}[Y],\Phi_{h_i}\rangle_n$ is asymptotically negligible, i.\,e., \eqref{bias} holds and  $\varphi\in H^{\frac{d+1}{2}}(\R^d)$, such that the integral in \eqref{eq:DXiRadon} is finite.
		If furthermore \eqref{eq:scale_logs} holds true with $0 < \delta < \Delta \leq 1$ and the grids of positions $\bt$ and scales $h$ are sufficiently fine, i.e. satisfy \eqref{Grid1} and \eqref{Grid2} and if the calibration 
		\begin{align}\label{eq:CalRadon}
		\omega(K,1+d)\quad\text{with}\quad K=(1-\rho)^d(2\pi)^{-\frac{d+1}{2}}\mathrm{det}(D_{\Xi}^{-2})^{\frac{1}{2}}\log(\Delta/\delta)
		\end{align}
		is used (see \eqref{eq:calibration} and Remark \ref{BT1}(b)), where $D_{\Xi}^{-2}$ is defined in \eqref{eq:DXiRadon}, the following holds
		\begin{align*}
		\lim_{n\to\infty}\mathbb{P}_0\bigl[\mathcal{S}(Y)\leq\lambda\bigr]=e^{-e^{-\lambda}}.
		\end{align*}
		Furthermore the statements of Lemma \ref{Lemma:Power1} and Theorem \ref{Thm:Power} also hold.
	\end{Theorem}
	
	\subsection{Deconvolution}\label{Sec:Dec}
	We discuss now in detail the case of deconvolution, i.\,e. \eqref{model} specializes to
	\begin{align}\label{model2}
	Y_{\bj}=(\psf \ast f)(\bx_{\bj})+\xi_{\bj}, \quad \bj\in\{1,\ldots, n\}^d,
	\end{align}
	where the function $\psf$ is a convolution kernel and the operation "$\ast$" denotes convolution as defined in \eqref{eq:conv}. In our subsequent data example $\psf$ corresponds to the point-spread function (PSF) of a microscope \citep[see e.\,g.][]{bbdv09,aem15,HW:16} .
	
	Assume that there exist positive constants $\underline{c},\overline{C}$ and $a$ such that
	\begin{align}\label{D1}
	\underline{c}\bigl(1+\|\bxi\|_2^2\bigr)^{-a}\leq|\mathcal{F}_d\psf(\bxi)|\leq\overline{C}\bigl(1+\|\bxi\|_2^2\bigr)^{-a}.\tag{D1}
	\end{align} 
	Assumption \eqref{D1} is a standard assumption characterizing mildly ill-posed deconvolution problems \citep[see e.\,g.][]{f91,Meister}.
	For any fixed function $\varphi$, $\|\varphi\|_2>0,$ generating a dictionary
	\begin{equation}\label{eq:dec_dict}
	\mathcal{U}=\Bigl\{\varphi_{i}\,\Big|\,\varphi_i(\bz)=\varphi\Big(\frac{\bt_i-\bz}{\bh_i}\Big),\;i\in I_N\Bigr\},
	\end{equation}
	the corresponding dictionary $\mathcal{W}$ inherits the required wavelet-type structure:
	\begin{align}\label{DictConvolution}
	\mathcal{W}=\Bigl\{\Phi_{i}\,\Big|\,\Phi_i(\bz)=\Phi_{\bh_i}\Big(\frac{\bt_i-\bz}{\bh_i}\Big),\;\Phi_{\bh_i}:=\mathcal{F}_d^{-1}\biggl(\frac{\mathcal{F}_d\phi}{\overline{\mathcal{F}_d\psf(\cdot/\bh_i)}}\biggr),\;i\in I_N\Bigr\},
	\end{align}
	and the results from the previous section transfer to deconvolution as follows.
	\begin{Theorem}[MISCAT for deconvolution]\label{Thm:Conv2}
		Suppose  model \eqref{model2} with convolution kernel $\psf$ satisfying Assumption \eqref{D1} and $\xi_{\bj}$ satisfying Assumption \ref{Ass:Noise}. Let $\bt_i\in[\pmb{\rho}+\bh_i,\be-\pmb{\rho}]$, where $\pmb{\rho}=(\rho,\ldots,\rho)^T,\,\rho>0.$ Consider the dictionary $\mathcal{W}$, given by 
		\eqref{DictConvolution} such that Assumption \ref{Ass:Dict} is satisfied and, in addition, 
		$\phi$ belongs to a Sobolev space $H^{2a+\gamma\vee 1/2}(\mathbb{R}^d)$. Assume further that
		the  approximation error of $\langle \mathbb{E}[Y],\Phi_i\rangle_n$ is asymptotically negligible, i.\,e., \eqref{bias} holds.
		\begin{enumerate}
			\item The results of Theorem \ref{Th:1}a) carry over to this general  convolution setting.\\
			\item Furthermore, let the  grids of positions $\bt$ and scales $\bh$ sufficiently fine, i.e. satisfy \eqref{Grid1} and \eqref{Grid2}. Then there exist  positive constants $ \underline{D}_{\gamma}$ and  $ \overline{D}_{\gamma}$ such that for any fixed $\lambda\in\R$ 
			\begin{align*}
			e ^{-\underline{D}_{\gamma} \displaystyle e^{-\lambda}}\leq\lim_{n\to\infty}\mathbb{P}_0\bigl[\mathcal{S}(Y)\leq\lambda\bigr]\leq e ^{-\overline{D}_{\gamma} \displaystyle e^{-\lambda}}.
			\end{align*}		
			Hence, under $H_0$, $\mathcal{S}(Y)$ is asymptotically non-degenerate.
			\item 	In the situation of (b), let $\bh_i=(h_i,\ldots,h_i)$ for all $i\in I_N$ and assume that \eqref{eq:scale_logs} holds true with $0 < \delta < \Delta \leq 1$. If the stronger condition \eqref{eq:Phi_h_convergence} holds, then with the calibration $w(K,1+d)$ we obtain
			\begin{align*}
			\lim_{n\to\infty}\mathbb{P}_0\bigl(\mathcal{S}(Y)\leq\lambda\bigr)=\exp\left(-\exp\left(-\lambda\right)\cdot\tfrac{H_{2\gamma}\det\left(D_{\Xi}^{-1}\right) \log(\Delta/\delta)}{\sqrt{2\pi}K}\right),
			\end{align*}
		\end{enumerate}
	\end{Theorem}
	\begin{Remark}\label{BT3}
		\begin{enumerate}
			\item	In Theorem \ref{Thm:Conv2} we need to exclude a small boundary region of the observations from the analysis since, in general, the functions $\Phi_{\bh_i}$ in $\mathcal{W}$ might be of unbounded support. Then the results of Theorem \ref{Th:1} transfer to this setting.
			\item The results from Theorem \ref{Thm:Conv2} (c) require assumption \eqref{eq:Phi_h_convergence} which basically means that the convolution kernel $\psf$ should decay exactly like a polynomial if $\|\bxi\|_2\to\infty$ in contrast to the weaker assumption \eqref{D1} which only requires upper and lower polynomial bounds and can hence only ensure upper and lower Gumbel bounds. In Section \ref{Sec:Appl} we provide a specific example for which both \eqref{D1} and \eqref{eq:Phi_h_convergence} are satisfied.
		\end{enumerate}
	\end{Remark}
	
	\subsubsection{Optimal detection in deconvolution} \label{Sec:OptDec}
	In this section we discuss and specify the results from Sections \ref{Sec:GT3} and \ref{Sec:Opt} for deconvolution. The results given in Lemma \ref{Lemma:Power1} also hold in the general deconvolution setting. 
	The following lemma contains a related result in the situation of \eqref{Def:PSF} concerning the support inference about the signal $f$ itself.
	
	\begin{Lemma}
		\label{Lemma:Power3}Given observations from model \eqref{model2} with random noise satisfying Assumption \ref{Ass:Noise} and $\psf$ as in \eqref{Def:PSF} and given a non-negative function $\phi\in \mathcal{R}(T^*)$, define the  dictionary  $\mathcal{W}$ as  in  \eqref{DictConvolution}. Suppose that the signal $f$ is non-negative as well.
		Let further $\mathcal{I}_{\alpha}(f)$ denote the set
		\begin{align*}
		\mathcal{I}_{\alpha}(f):= \bigl\{i\,\big|\,f\big|_{\mathrm{supp(\varphi_i)}}>2q_{i,1-\alpha}\|\sigma\Phi_{i}\|_2/(h_{i,1} h_{i,2}n^{d/2})\bigr\}.
		\end{align*}
		Then, under the assumptions of Theorem \ref{Th:1},
		\begin{align*}
		\lim_{n\rightarrow\infty}\mathbb{P}\bigl(\langle\Phi_i\,,\,Y\rangle_n>q_{i,1-\alpha}\|\sigma\Phi_i\|_2/n\quad\text{for all}\quad i\in\mathcal{I}_{\alpha}(f)\bigr)\geq1-\alpha.
		\end{align*} 
	\end{Lemma}
	The result above immediately shows that the choice of $\varphi$ in \eqref{eq:dec_dict} has a high influence on the detection properties of the corresponding test via the variances $\Vert \sigma \Phi_{i}\Vert_2^2$. Extending an argument from \citet{SHMunkDuem2013} for $d = 1$ to general $d$, we can provide a mother wavelet $\varphi$ which minimizes the asymptotic variance of the test statistic over all tensor-type probe functions. It only depends on the polynomial order of decay  of the convolution kernel in Fourier space ($\hat =$ degree of ill-posedness) and is (for $d = 2$) given by
	\begin{align}\label{Probe}
	\varphi\left(x,y\right) = x^{\beta_1 + 1}\left(1-x\right)^{\beta_1 + 1}y^{\beta_2 + 1}\left(1-y\right)^{\beta_2 + 1} \mathbf 1_{\left(0,1\right)} \left(x\right)\mathbf 1_{\left(0,1\right)} \left(y\right),
	\end{align}
	where the two parameters $\beta_1, \beta_2 \in \mathbb N$ equal the polynomial order of decay  of the convolution kernel in $x$ and $y$ direction. This choice will be considered in the following. 
	
	The previous lemma implies the consistency of the testing procedure for the signal itself, i.\,e., testing $f=0$ versus $f>0$, if the minimal scale satisfies $h_{\min}\gtrsim(\log(n)/n)^{1/(4a+1)}.$
	Moreover, in the situation of Theorem \ref{Thm:Conv1} (c) the optimality results of Section \ref{Sec:Opt} carry over to the deconvolution setting.
	For a comparison consider the rate of estimation of the $2a$-th derivative of a H\"older $\beta$ function w.r.t $L^\infty$ risk in $d = 1$. We restrict to this case as otherwise the deconvolution is no longer equivalent to estimating derivatives, cf. \eqref{eq:tildePhi}. This is possible with minimax rate
	\[
	\left(\frac{\log n}{n}\right)^{\beta/(2 \beta + 4a + 1)},
	\]
	which is attained for $h \sim \left(\log n/n\right)^{1/(2 \beta + 4a + 1)}$ \citep[see e.\,g.][]{jkpr04}, i.e. such a function can be distinguished from $0$  by means of estimation on a box $\left[t-h, t\right]$ as long as it is asymptotically larger than $h^\beta$. Posing the same question to MISCAT, the above result show that for $f\big|_{\left[t-h,t\right]} \sim h^\beta$ and $h \sim \left(\log n/n\right)^{1/(2 \beta + 4a + 1)}$ it recognizes $\left[t-h,t\right]$ as active with (asymptotic) probability $\geq 1-\alpha$. Consequently any support points found by estimation will also be found by MISCAT.
	
	\section{Simulations and real data applications}\label{Sec:Appl}
	
	In this section we investigate the finite sample properties of the proposed multiscale test. To this end, we apply MISCAT in a $2$-dimensional mildly ill-posed deconvolution problem. In Section \ref{Application} we then analyze experimental STED data to locate single DNA origami in a sample.
	
	Specifying the setting described in Section \ref{Sec:Dec} to this situation, the data is given by \eqref{model2}. The convolution kernel $\psf$ is chosen from the  family $\left\{\psf_{a,b} ~\big|~ a \in \mathbb N, b >0\right\}$ defined in Fourier space via
	\begin{align}\label{Def:PSF}
	\left(\mathcal F_2 \psf_{a,b}\right) \left(\pmb{\xi}\right)= (1 + b^2 \left\Vert \pmb{\xi}\right\Vert_2^2)^{-a}, \qquad \pmb{\xi} \in \mathbb R^2.
	\end{align}
	Model \eqref{Def:PSF} is a $2$-dimensional generalization of the one-dimensional family of auto-convolu\-tions of a scaled version of the density of the Laplace distribution with itself with radially symmetric PSF. 
	
	For any convolution kernel $\psf_{a,b}$ Assumption \eqref{D1} is obviously satisfied and we obtain
	\begin{equation}\label{eq:tildePhi}
	\Phi_{\bh_i}=\sum_{j=0}^a\sum_{k=0}^j\binom{a}{j}\binom{j}{k}\Bigl(\frac{b}{h_{i,1}}\Bigr)^{2k}\Bigl(\frac{b}{h_{i,2}}\Bigr)^{2(j-k)}\partial^{(2k\,,\,2(j-k))}\varphi.
	\end{equation}
	This shows that a compactly supported function $\phi$ results in a dictionary $\mathcal{W}$ which consists of compactly supported functions as well. Consequently, the results from Theorem \ref{Thm:Conv2} can be obtained even without excluding a small boundary region, and furthermore a Gumbel limit theorem can be obtained as follows. Let
	\begin{align}\label{Def:XiDec}
	\Xi=\frac{\widetilde{\Xi}}{\|\widetilde{\Xi}\|_2},\quad\text{where}\quad\widetilde{\Xi}=b^{2a}\sum_{k=0}^{a}\binom{a}{k}\partial^{2k,2(a-k)}\varphi.
	\end{align}
	and consider the case $\bh_i=(h_i,h_i)$ for all $i\in I_N$.
	Then
	\begin{align}\label{eq:XiDec1}
	\|\Phi_{\bh_i}\|_2=\bigg(\frac{1}{h_i}\bigg)^{2a}\|\Xi+h_i^2\Xi_{n,i}\|_2\quad\text{and}\quad\frac{\Phi_{\bh_i}}{\|\Phi_{\bh_i}\|_2}=\frac{\Xi+h_i^2\Xi_{n,i}}{\|\Xi+h_i^2\Xi_{n,i}\|_2},
	\end{align}
	where
	\begin{align}\label{eq:XiDec2}
	\Xi_{n,i}:=\sum_{j=0}^{a-1}h_i^{2(a-1-j)}\sum_{k=0}^j\binom{j}{k}\binom{a}{j}\partial^{2k,2(j-k)}\varphi.
	\end{align}
	In this setting, it is easy to verify that condition \eqref{eq:Phi_h_convergence} holds. 
	\begin{Theorem}[MISCAT for our application]\label{Thm:Conv1}~\\
		Suppose that we have access to observations following model \eqref{model2} with convolution kernel $\psf_{a,b}$ satisfying Assumption \eqref{Def:PSF}, $d=2$ and random noise satisfying Assumption \ref{Ass:Noise}. Assume that the dictionary is given by 
		\eqref{DictConvolution} with dictionary functions $\Phi_{\bh_i}$ defined in \eqref{eq:tildePhi} such that Assumption \ref{Ass:Dict} is satisfied, and that \eqref{bias} holds.
		\begin{enumerate} 
			\item The results of Theorem \ref{Th:1}(a) carry over to this particular convolution setting.
			\item If $\varphi \in H^{2a+\gamma \wedge 1/2} \left(\mathbb R^2\right)$ and if the grids of positions $\bt$ and scales $\bh$ are sufficiently fine, i.e. satisfy \eqref{Grid1} and \eqref{Grid2}, then the results of Theorem \ref{Thm:Conv2}(b) carry over to this particular convolution setting.
			\item Suppose furthermore that $\bh_i=(h_i,h_i)$ for all $i\in I_N$, that \eqref{eq:scale_logs} holds true with $0 < \delta < \Delta \leq 1$ and that the grids of positions $\bt$ and scales $\bh$ are sufficiently fine, i.e. satisfy \eqref{Grid1} and \eqref{Grid2}. If in addition, $\varphi$ is $(2a+1)$-times differentiable in $L^{2}(\R^d)$, let $\varphi_{\balpha}=\sum_{k=0}^a\binom{a}{k}\partial^{2k+\alpha_1,2(a-k)+\alpha_2}\varphi$, $\balpha\in\{0,1\}^2,\,|\balpha|=1.$ Then, for
			\begin{align}\label{eq:CalDec}
			\omega(K,1+d)\quad\text{with}\quad K=b^{4a}\log(\Delta/\delta)(2\pi)^{-\frac{3}{2}}\|\widetilde{\Xi}\|_2^{-1}\sqrt{\|\varphi_{0,1}\|_2^2\|\varphi_{1,0}\|_2^2-\langle\varphi_{0,1}\varphi_{1,0}\rangle}
			\end{align}
			(see \eqref{eq:calibration} and Remark \ref{BT1}(b)), we obtain
			\[
			\lim_{n\to\infty}\mathbb{P}_0\bigl[\mathcal{S}(Y)\leq\lambda\bigr]=e^{-e^{-\lambda}}.
			\]
		\end{enumerate}
	\end{Theorem}
	
	\subsection{$2$-dimensional support inference}\label{subsec:2D_simulations}
	
	Now we will infer on the support of a testfunction of realistic size from simulated data. In accordance with our subsequent data example we choose $n = 512$, and the kernel parameters are set to $a = 2$ and $b = 0.0243$. This implies that the order of decay  of the convolution kernel in Fourier space is $4$ and the full width at half maximum (FWHM) of $\psf_{a,b}$ is about 10 pixels in the above setting (cf. Appendix \ref{appA}). 
	
	In this situation we investigate three tests, which differ in the choice of the probe function $\varphi$ acting as mother wavelet. We use three different setups:
	\begin{itemize}
		\item[$\Psi_{\mathrm c}$] \textbf{Correct setup.} In this case, the test uses the true ill-posedness, which corresponds to choosing $\beta_1 = \beta_2 = 2a = 4$.
		\item[$\Psi_{\mathrm o}$] \textbf{Oversmoothing setup.} In this case, the probe function is smoother than necessary, i.e. the ill-posedness is overspecified with $\beta_1 = \beta_2 = 10$. From a practitioners point of view, this choice can be considered to be pessimistic.
		\item[$\Psi_{\mathrm u}$] \textbf{Undersmoothing setup.} In this case, the probe function is not smooth enough, i.e. the ill-posedness is underestimated with $\beta_1 = \beta_2 = 1$. From a practitioners point of view, this choice can be considered to be optimistic.
	\end{itemize}
	In all of these cases, the constant $K$, appearing in the definition of the calibration $\omega_{i}$ in \eqref{eq:calibration} is chosen by a numerical approximation of \eqref{eq:CalDec}.
	
	All our tests use 196 different scales defined by boxes consisting of $k_x \times k_y$ pixels, $k_x, k_y = 4,6, ..., 30$. Concerning the positions $\bt$ we  use again all possible upper left points of boxes fitting in the image. 
	
	\subsubsection{Distribution of the approximating Gaussian version}
	
	Figure \ref{fig:2D_hist} shows histograms of $10.000$ runs of the approximating Gaussian test statistic $\mathcal S \left(W\right)$. Note again that this statistic is independent of all unknown quantities including the variance due to the standardization of the local test statistics. The simulations suggest a stable behavior of the distribution of the maximum statistic for all three tests. 
	
	\begin{figure}[!htb]
		\setlength\fheight{4cm} \setlength\fwidth{4cm}
		\centering
		\tiny
		\subfigure[test statistic for $\Psi_{\mathrm c}$]{
			\label{subfig:2D_hist_correct_stat}
			% This file was created by matlab2tikz v0.4.7 running on MATLAB 7.11.
% Copyright (c) 2008--2014, Nico Schlömer <nico.schloemer@gmail.com>
% All rights reserved.
% Minimal pgfplots version: 1.3
% 
% The latest updates can be retrieved from
%   http://www.mathworks.com/matlabcentral/fileexchange/22022-matlab2tikz
% where you can also make suggestions and rate matlab2tikz.
% 
%
% defining custom colors
\definecolor{mycolor1}{rgb}{0.00000,0.00000,0.56250}%
\begin{tikzpicture}

\begin{axis}[%
width=\fwidth,
height=\fheight,
area legend,
scale only axis,
xmin=0,
xmax=14,
ymin=0,
ymax=1000
]
\addplot[ybar,bar width=0.0181389359998107\fwidth,draw=black,fill=mycolor1] plot table[row sep=crcr] {%
2.26797695247994	4\\
2.44936631247805	29\\
2.63075567247616	61\\
2.81214503247426	127\\
2.99353439247237	216\\
3.17492375247048	359\\
3.35631311246859	478\\
3.53770247246669	592\\
3.7190918324648	680\\
3.90048119246291	763\\
4.08187055246102	743\\
4.26325991245912	774\\
4.44464927245723	724\\
4.62603863245534	680\\
4.80742799245345	604\\
4.98881735245155	550\\
5.17020671244966	478\\
5.35159607244777	392\\
5.53298543244588	336\\
5.71437479244398	299\\
5.89576415244209	221\\
6.0771535124402	166\\
6.25854287243831	142\\
6.43993223243641	109\\
6.62132159243452	94\\
6.80271095243263	90\\
6.98410031243074	55\\
7.16548967242884	51\\
7.34687903242695	42\\
7.52826839242506	35\\
7.70965775242317	23\\
7.89104711242127	17\\
8.07243647241938	23\\
8.25382583241749	6\\
8.4352151924156	11\\
8.6166045524137	2\\
8.79799391241181	6\\
8.97938327240992	2\\
9.16077263240803	0\\
9.34216199240613	2\\
9.52355135240424	4\\
9.70494071240235	3\\
9.88633007240046	4\\
10.0677194323986	0\\
10.2491087923967	0\\
10.4304981523948	2\\
10.6118875123929	0\\
10.793276872391	0\\
10.9746662323891	0\\
11.1560555923872	1\\
};
\addplot [color=black,solid,forget plot]
  table[row sep=crcr]{%
2	0\\
12	0\\
};
\end{axis}
\end{tikzpicture}% 
		}
		\subfigure[test statistic for $\Psi_{\mathrm o}$]{
			\label{subfig:2D_hist_oversmooth_stat}
			% This file was created by matlab2tikz v0.4.7 running on MATLAB 7.11.
% Copyright (c) 2008--2014, Nico Schlömer <nico.schloemer@gmail.com>
% All rights reserved.
% Minimal pgfplots version: 1.3
% 
% The latest updates can be retrieved from
%   http://www.mathworks.com/matlabcentral/fileexchange/22022-matlab2tikz
% where you can also make suggestions and rate matlab2tikz.
% 
%
% defining custom colors
\definecolor{mycolor1}{rgb}{0.00000,0.00000,0.56250}%
\begin{tikzpicture}

\begin{axis}[%
width=\fwidth,
height=\fheight,
area legend,
scale only axis,
xmin=0,
xmax=14,
ymin=0,
ymax=1000
]
\addplot[ybar,bar width=0.016965444819422\fwidth,draw=black,fill=mycolor1] plot table[row sep=crcr] {%
5.20171166131466	2\\
5.32046977505061	8\\
5.43922788878657	32\\
5.55798600252252	81\\
5.67674411625847	168\\
5.79550222999443	295\\
5.91426034373038	513\\
6.03301845746633	627\\
6.15177657120229	760\\
6.27053468493824	836\\
6.3892927986742	897\\
6.50805091241015	839\\
6.6268090261461	878\\
6.74556713988206	689\\
6.86432525361801	608\\
6.98308336735397	544\\
7.10184148108992	420\\
7.22059959482587	360\\
7.33935770856183	326\\
7.45811582229778	248\\
7.57687393603374	184\\
7.69563204976969	147\\
7.81439016350564	127\\
7.9331482772416	90\\
8.05190639097755	63\\
8.17066450471351	56\\
8.28942261844946	44\\
8.40818073218541	28\\
8.52693884592137	31\\
8.64569695965732	23\\
8.76445507339328	17\\
8.88321318712923	17\\
9.00197130086518	9\\
9.12072941460114	5\\
9.23948752833709	6\\
9.35824564207305	7\\
9.477003755809	5\\
9.59576186954495	2\\
9.71451998328091	4\\
9.83327809701686	2\\
9.95203621075282	0\\
10.0707943244888	0\\
10.1895524382247	0\\
10.3083105519607	1\\
10.4270686656966	0\\
10.5458267794326	0\\
10.6645848931685	0\\
10.7833430069045	0\\
10.9021011206404	0\\
11.0208592343764	1\\
};
\addplot [color=black,solid,forget plot]
  table[row sep=crcr]{%
5	0\\
12	0\\
};
\end{axis}
\end{tikzpicture}% 
		}
		\subfigure[test statistic for $\Psi_{\mathrm u}$]{
			\label{subfig:2D_hist_undersmooth_stat}
			% This file was created by matlab2tikz v0.4.7 running on MATLAB 7.11.
% Copyright (c) 2008--2014, Nico Schlömer <nico.schloemer@gmail.com>
% All rights reserved.
% Minimal pgfplots version: 1.3
% 
% The latest updates can be retrieved from
%   http://www.mathworks.com/matlabcentral/fileexchange/22022-matlab2tikz
% where you can also make suggestions and rate matlab2tikz.
% 
%
% defining custom colors
\definecolor{mycolor1}{rgb}{0.00000,0.00000,0.56250}%
\begin{tikzpicture}

\begin{axis}[%
width=\fwidth,
height=\fheight,
area legend,
scale only axis,
xmin=0,
xmax=14,
ymin=0,
ymax=1000
]
\addplot[ybar,bar width=0.0161648850911559\fwidth,draw=black,fill=mycolor1] plot table[row sep=crcr] {%
1.2187485261771	1\\
1.44505691745329	3\\
1.67136530872947	30\\
1.89767370000566	88\\
2.12398209128184	142\\
2.35029048255802	296\\
2.57659887383421	505\\
2.80290726511039	673\\
3.02921565638658	823\\
3.25552404766276	857\\
3.48183243893895	896\\
3.70814083021513	895\\
3.93444922149131	827\\
4.1607576127675	759\\
4.38706600404368	629\\
4.61337439531987	565\\
4.83968278659605	435\\
5.06599117787223	345\\
5.29229956914842	278\\
5.5186079604246	224\\
5.74491635170079	185\\
5.97122474297697	131\\
6.19753313425316	120\\
6.42384152552934	72\\
6.65014991680552	59\\
6.87645830808171	39\\
7.10276669935789	29\\
7.32907509063408	36\\
7.55538348191026	18\\
7.78169187318644	2\\
8.00800026446263	12\\
8.23430865573881	5\\
8.460617047015	4\\
8.68692543829118	2\\
8.91323382956737	2\\
9.13954222084355	3\\
9.36585061211973	2\\
9.59215900339592	2\\
9.8184673946721	2\\
10.0447757859483	0\\
10.2710841772245	1\\
10.4973925685007	0\\
10.7237009597768	1\\
10.950009351053	0\\
11.1763177423292	0\\
11.4026261336054	1\\
11.6289345248816	0\\
11.8552429161578	0\\
12.0815513074339	0\\
12.3078596987101	1\\
};
\addplot [color=black,solid,forget plot]
  table[row sep=crcr]{%
0	0\\
14	0\\
};
\end{axis}
\end{tikzpicture}% 
		}
		\caption{Histograms for $10^4$ runs of the test statistic applied to pure Gaussian white noise.}
		\label{fig:2D_hist}
	\end{figure}
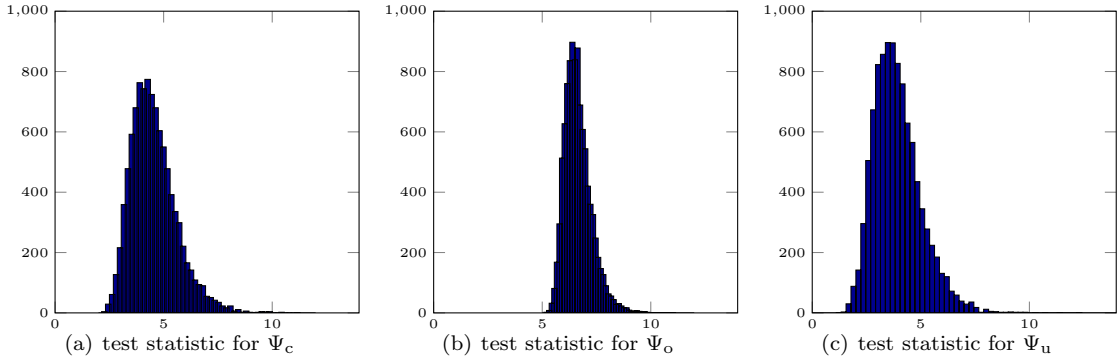
	
	\subsubsection{Detection properties}
	
	To investigate more detailed the influence of the proper specification of ill-posedness via the three tests $\Psi_{\mathrm c}$, $\Psi_{\mathrm o}$ and $\Psi_{\mathrm u}$, empirical quantiles $q_\alpha$ are computed from the simulations in Figure \ref{fig:2D_hist}, which are also shown in Table \ref{tab:2D_quantiles}. 
	
	\begin{table}[!htb]
		\centering
		\begin{tabular}{c|c|cccccc}
			\toprule[2pt]
			&$\alpha$ & $0.1$ & $0.5$ &$0.8$ &$0.9$ &$0.95$ &$0.99$ \\
			\midrule[1pt]
			$\Psi_{\mathrm c}$&$q_\alpha$ & $3.3420$ & $4.4006$ & $5.3194$ & $5.8828$ & $6.4722$ & $7.6442$\\
			\midrule[1pt]
			$\Psi_{\mathrm o}$&$q_\alpha$ & $5.9541$ & $6.5553$ & $7.1012$ & $7.4512$ & $7.7888$ & $8.5752$\\
			\midrule[1pt]
			$\Psi_{\mathrm u}$&$q_\alpha$ & $2.6635$ & $3.7673$ & $4.7307$ & $5.3649$ & $5.9227$ & $7.1658$\\
			\bottomrule[2pt]
		\end{tabular}
		\caption{Empirical quantiles of the test statistic for the three different test in Figure \ref{fig:2D_hist}.}
		\label{tab:2D_quantiles}
	\end{table}
	
	We expect the test to reach a twofold aim: To detect weak signals and to distinguish between different signals, which are strong enough (see Lemma \ref{Lemma:Power1}). To illustrate the difference in power of the three tests, we show an exemplary simulation based on a synthetic testfunction which is a binary image (see Figure \ref{fig:2D_testfunction}(a)). In agreement with our prospects, the testfunction consists of some comparably large and well-distributed circles (3 in the top row, 6 in the second row, 12 in the third row), of a line of squares in the fourth row, and of a continuous line of same height in the bottom row at different sizes.
	
	\begin{figure}[!htb]
		\setlength\fheight{5cm} \setlength\fwidth{5cm}
		\centering
		\begin{tabular}{ll}
			% This file was created by matlab2tikz v0.4.7 running on MATLAB 7.11.
% Copyright (c) 2008--2014, Nico Schlömer <nico.schloemer@gmail.com>
% All rights reserved.
% Minimal pgfplots version: 1.3
% 
% The latest updates can be retrieved from
%   http://www.mathworks.com/matlabcentral/fileexchange/22022-matlab2tikz
% where you can also make suggestions and rate matlab2tikz.
% 
\begin{tikzpicture}[baseline]

\begin{axis}[%
width=\fwidth,
height=\fheight,
axis on top,
scale only axis,
xmin=0.5,
xmax=512.5,
trim axis left,
trim axis right,
xtick = \empty,
ytick = \empty,
ymin=0.5,
ymax=512.5,
colormap={mymap}{[1pt] rgb(0pt)=(1,1,1); rgb(63pt)=(0,0,0)},
colorbar,
point meta min=0,
point meta max=1
]
\addplot [forget plot] graphics [xmin=0.5,xmax=512.5,ymin=0.5,ymax=512.5] {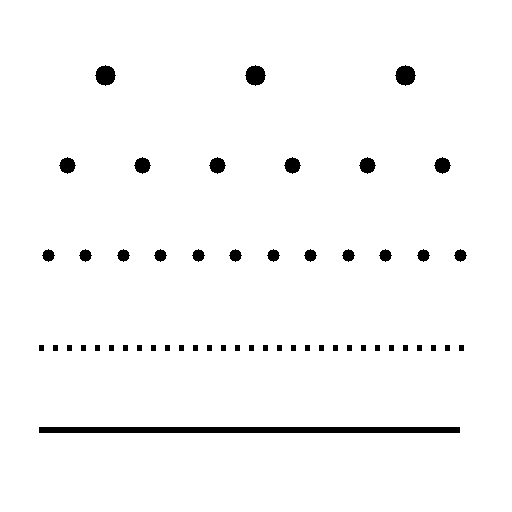};
\end{axis}
\end{tikzpicture}% 
			&
			% This file was created by matlab2tikz v0.4.7 running on MATLAB 7.11.
% Copyright (c) 2008--2014, Nico Schlömer <nico.schloemer@gmail.com>
% All rights reserved.
% Minimal pgfplots version: 1.3
% 
% The latest updates can be retrieved from
%   http://www.mathworks.com/matlabcentral/fileexchange/22022-matlab2tikz
% where you can also make suggestions and rate matlab2tikz.
% 
\begin{tikzpicture}[baseline]

\begin{axis}[%
width=\fwidth,
height=\fheight,
axis on top,
scale only axis,
xmin=0.5,
xmax=512.5,
trim axis left,
trim axis right,
ymin=0.5,
ymax=512.5,
xtick = \empty,
ytick = \empty,
colormap={mymap}{[1pt] rgb(0pt)=(1,1,1); rgb(63pt)=(0,0,0)},
colorbar,
point meta min=-2.42042138563527,
point meta max=2.30336918456527
]
\addplot [forget plot] graphics [xmin=0.5,xmax=512.5,ymin=0.5,ymax=512.5] {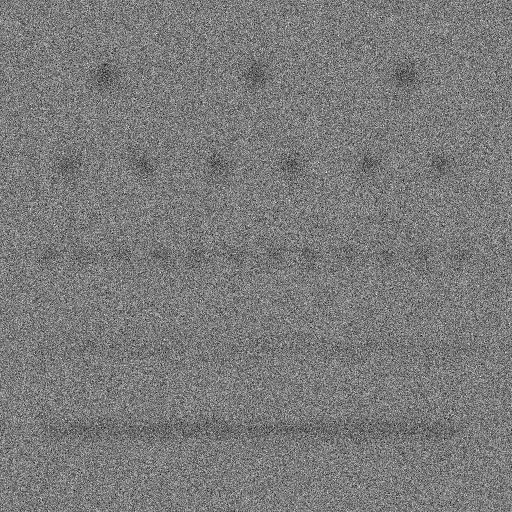};
\end{axis}
\end{tikzpicture}%
			\\
			\hspace*{1cm}(a) testfunction
			&
			\hspace*{1cm}(b) normal noise, $\sigma = 0.5$
			\\[0.2cm]
			% This file was created by matlab2tikz v0.4.7 running on MATLAB 7.11.
% Copyright (c) 2008--2014, Nico Schlömer <nico.schloemer@gmail.com>
% All rights reserved.
% Minimal pgfplots version: 1.3
% 
% The latest updates can be retrieved from
%   http://www.mathworks.com/matlabcentral/fileexchange/22022-matlab2tikz
% where you can also make suggestions and rate matlab2tikz.
% 
\begin{tikzpicture}[baseline]

\begin{axis}[%
width=\fwidth,
height=\fheight,
axis on top,
scale only axis,
xmin=0.5,
xmax=512.5,
trim axis left,
trim axis right,
ymin=0.5,
ymax=512.5,
xtick = \empty,
ytick = \empty,
colormap={mymap}{[1pt] rgb(0pt)=(1,1,1); rgb(63pt)=(0,0,0)},
colorbar,
point meta min=-0.231716263892619,
point meta max=0.741994574433032
]
\addplot [forget plot] graphics [xmin=0.5,xmax=512.5,ymin=0.5,ymax=512.5] {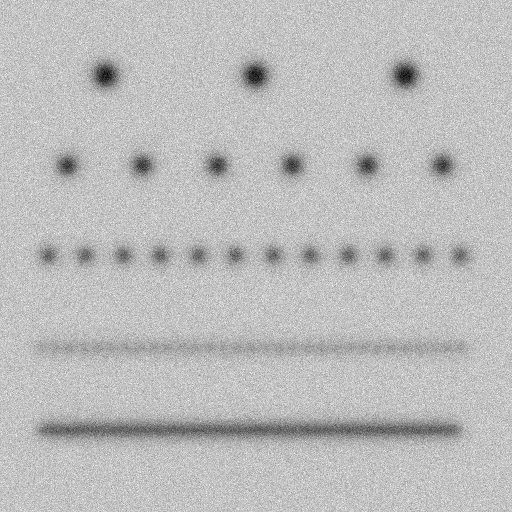};
\end{axis}
\end{tikzpicture}%
			&
			% This file was created by matlab2tikz v0.4.7 running on MATLAB 7.11.
% Copyright (c) 2008--2014, Nico Schlömer <nico.schloemer@gmail.com>
% All rights reserved.
% Minimal pgfplots version: 1.3
% 
% The latest updates can be retrieved from
%   http://www.mathworks.com/matlabcentral/fileexchange/22022-matlab2tikz
% where you can also make suggestions and rate matlab2tikz.
% 
\begin{tikzpicture}[baseline]

\begin{axis}[%
width=\fwidth,
height=\fheight,
axis on top,
scale only axis,
xmin=0.5,
xmax=512.5,
trim axis left,
trim axis right,
ymin=0.5,
ymax=512.5,
xtick = \empty,
ytick = \empty,
colormap={mymap}{[1pt] rgb(0pt)=(1,1,1); rgb(63pt)=(0,0,0)},
colorbar,
point meta min=-0.0231670093725303,
point meta max=0.636341983959635
]
\addplot [forget plot] graphics [xmin=0.5,xmax=512.5,ymin=0.5,ymax=512.5] {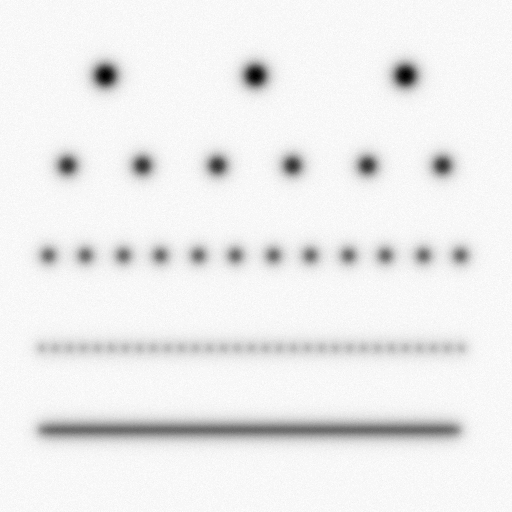};
\end{axis}
\end{tikzpicture}%
			\\
			\hspace*{1cm}(c) normal noise, $\sigma = 0.05$
			&
			\hspace*{1cm}(d) normal noise, $\sigma = 0.005$
		\end{tabular}
		\caption{Synthetic testfunction and simulated data for three different noise levels $\sigma \in \left\{0.005,0.05,0.5\right\}$.}
		\label{fig:2D_testfunction}
	\end{figure}
	
	For three different noise levels $\sigma \in \left\{0.005,0.05,0.5\right\}$ we simulate from a homogeneous Gaussian model (see Figure \ref{fig:2D_testfunction}(b)--(d)). To avoid masking effects caused by variance estimation, we always use the true variance in our test statistic. In Figure \ref{fig:2D_significance_maps} the data and the resulting significance maps for the three different tests are depicted. The significance map color-codes for each pixel the smallest scale on which it was significant. For example, in the top left subfigure the optimal test $\Psi_{\mathrm c}$ marked several boxes as significant, and the smallest scale on which significant boxes were found contains $624$ pixels. Any pixel contained in such a box is marked in dark red, which is the case in the third object in the top row and in objects number $2$ and $4$ in the second row of the top left subfigure. The other colors indicate pixels which belong to larger boxes marked as significant. In case that a pixel belongs to significant boxes of different scales, the smallest one is indicated by the color coding.
	
	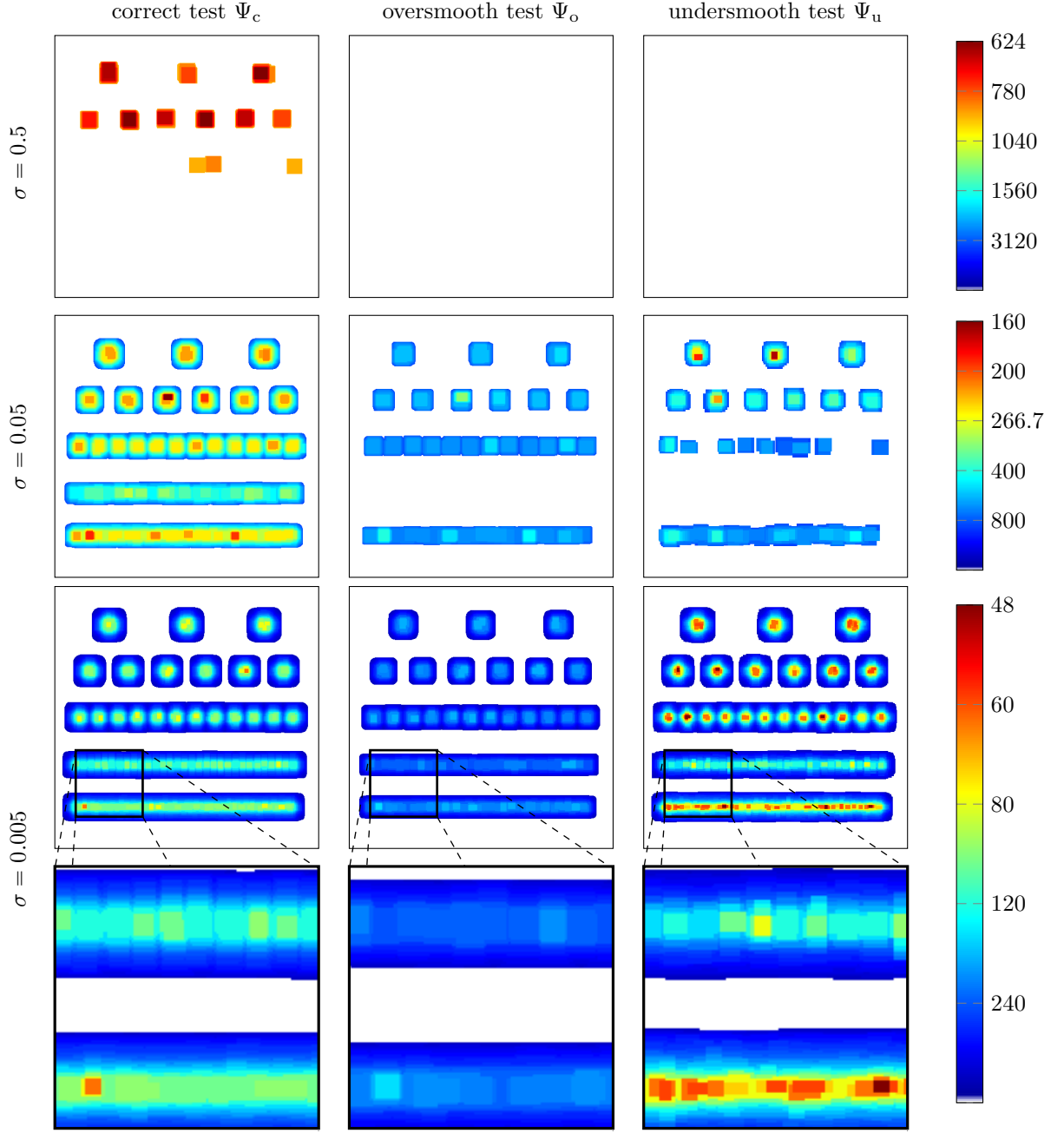
\begin{figure}[!htb]
		\setlength\fheight{4cm} \setlength\fwidth{4cm}
		\centering
		\begin{tabular}{rcccl}
			
			&
			correct test $\Psi_{\mathrm c}$
			&
			oversmooth test $\Psi_{\mathrm o}$
			&
			undersmooth test $\Psi_{\mathrm u}$
			&
			\\
			
			\raisebox{1.5cm}{\rotatebox{90}{$\sigma = 0.5$}}
			&
			% This file was created by matlab2tikz v0.4.7 running on MATLAB 7.11.
% Copyright (c) 2008--2014, Nico Schlömer <nico.schloemer@gmail.com>
% All rights reserved.
% Minimal pgfplots version: 1.3
% 
% The latest updates can be retrieved from
%   http://www.mathworks.com/matlabcentral/fileexchange/22022-matlab2tikz
% where you can also make suggestions and rate matlab2tikz.
% 
\begin{tikzpicture}

\begin{axis}[%
width=\fwidth,
height=\fheight,
trim axis left,
trim axis right,
axis on top,
scale only axis,
xmin=0.5,
xmax=512.5,
ymin=0.5,
ymax=512.5,
xtick = \empty,
ytick = \empty
]
\addplot [forget plot] graphics [xmin=0.5,xmax=512.5,ymin=0.5,ymax=512.5] {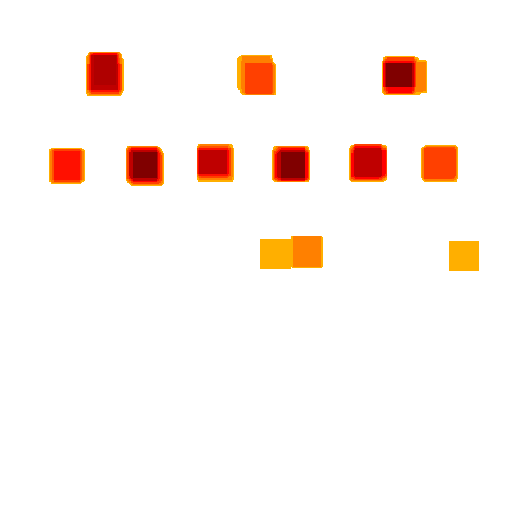};
\end{axis}
\end{tikzpicture}% 
			& 
			% This file was created by matlab2tikz v0.4.7 running on MATLAB 7.11.
% Copyright (c) 2008--2014, Nico Schlömer <nico.schloemer@gmail.com>
% All rights reserved.
% Minimal pgfplots version: 1.3
% 
% The latest updates can be retrieved from
%   http://www.mathworks.com/matlabcentral/fileexchange/22022-matlab2tikz
% where you can also make suggestions and rate matlab2tikz.
% 
\begin{tikzpicture}

\begin{axis}[%
width=\fwidth,
height=\fheight,
axis on top,
scale only axis,
xmin=0.5,
xmax=512.5,
trim axis left,
trim axis right,
ymin=0.5,
ymax=512.5,
xtick = \empty,
ytick = \empty
]
\addplot [forget plot] graphics [xmin=0.5,xmax=512.5,ymin=0.5,ymax=512.5] {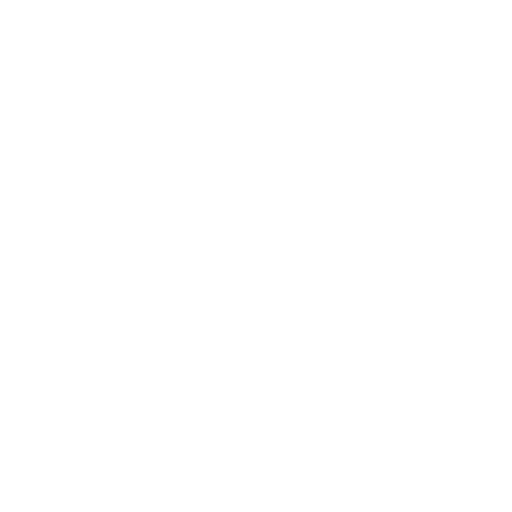};
\end{axis}
\end{tikzpicture}% 
			& 
			% This file was created by matlab2tikz v0.4.7 running on MATLAB 7.11.
% Copyright (c) 2008--2014, Nico Schlömer <nico.schloemer@gmail.com>
% All rights reserved.
% Minimal pgfplots version: 1.3
% 
% The latest updates can be retrieved from
%   http://www.mathworks.com/matlabcentral/fileexchange/22022-matlab2tikz
% where you can also make suggestions and rate matlab2tikz.
% 
\begin{tikzpicture}

\begin{axis}[%
width=\fwidth,
height=\fheight,
axis on top,
scale only axis,
xmin=0.5,
xmax=512.5,
trim axis left,
trim axis right,
ymin=0.5,
ymax=512.5,
xtick = \empty,
ytick = \empty
]
\addplot [forget plot] graphics [xmin=0.5,xmax=512.5,ymin=0.5,ymax=512.5] {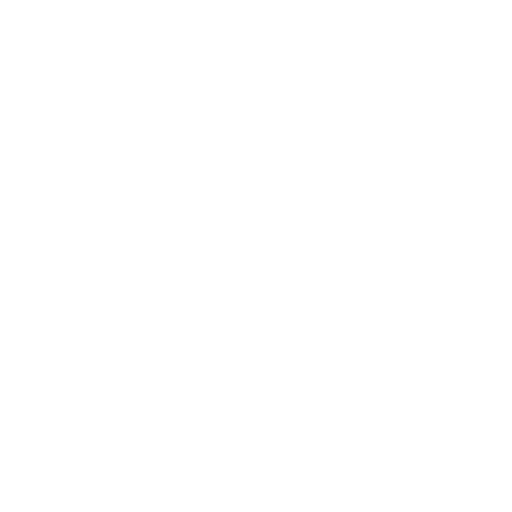};
\end{axis}
\end{tikzpicture}% 
			&
			\begin{tikzpicture}
\begin{axis}[
hide axis,
scale only axis,
height=0pt,
width=0pt,
colormap={mymap}{[1pt] rgb(0pt)=(1,1,1); rgb(1pt)=(0,0,0.625); rgb(7pt)=(0,0,1); rgb(23pt)=(0,1,1); rgb(39pt)=(1,1,0); rgb(55pt)=(1,0,0); rgb(63pt)=(0.5,0,0)},
colorbar,
colorbar/width = .4cm,
colorbar style={
ytick={0,0.000320512820512821,0.000641025641025641,0.000961538461538462,0.00128205128205128,0.0016025641025641},
yticklabels={{},{3120},{1560},{1040},{780},{624}},
scaled ticks = false,
width = .4cm,
height = .95\fheight
},
point meta min=0,
point meta max=0.0016025641025641
]
\addplot [draw = none] coordinates {(0,0)};
\end{axis}
\end{tikzpicture} \\
			
			\raisebox{1.35cm}{\rotatebox{90}{$\sigma = 0.05$}}
			&
			% This file was created by matlab2tikz v0.4.7 running on MATLAB 7.11.
% Copyright (c) 2008--2014, Nico Schlömer <nico.schloemer@gmail.com>
% All rights reserved.
% Minimal pgfplots version: 1.3
% 
% The latest updates can be retrieved from
%   http://www.mathworks.com/matlabcentral/fileexchange/22022-matlab2tikz
% where you can also make suggestions and rate matlab2tikz.
% 
\begin{tikzpicture}

\begin{axis}[%
width=\fwidth,
height=\fheight,
axis on top,
scale only axis,
xmin=0.5,
xmax=512.5,
trim axis left,
trim axis right,
ymin=0.5,
ymax=512.5,
xtick = \empty,
ytick = \empty
]
\addplot [forget plot] graphics [xmin=0.5,xmax=512.5,ymin=0.5,ymax=512.5] {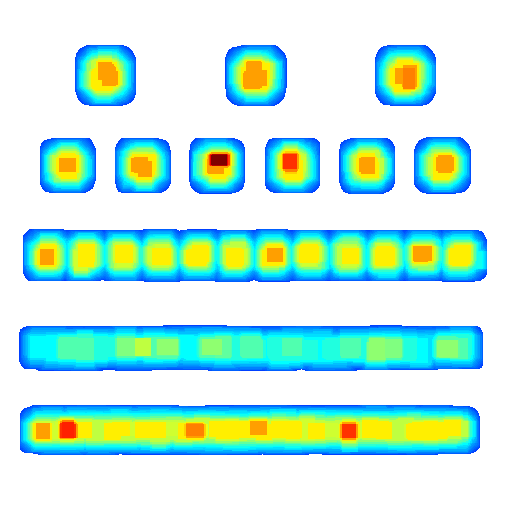};
\end{axis}
\end{tikzpicture}% 
			& 
			% This file was created by matlab2tikz v0.4.7 running on MATLAB 7.11.
% Copyright (c) 2008--2014, Nico Schlömer <nico.schloemer@gmail.com>
% All rights reserved.
% Minimal pgfplots version: 1.3
% 
% The latest updates can be retrieved from
%   http://www.mathworks.com/matlabcentral/fileexchange/22022-matlab2tikz
% where you can also make suggestions and rate matlab2tikz.
% 
\begin{tikzpicture}

\begin{axis}[%
width=\fwidth,
height=\fheight,
axis on top,
scale only axis,
xmin=0.5,
xmax=512.5,
trim axis left,
trim axis right,
ymin=0.5,
ymax=512.5,
xtick = \empty,
ytick = \empty
]
\addplot [forget plot] graphics [xmin=0.5,xmax=512.5,ymin=0.5,ymax=512.5] {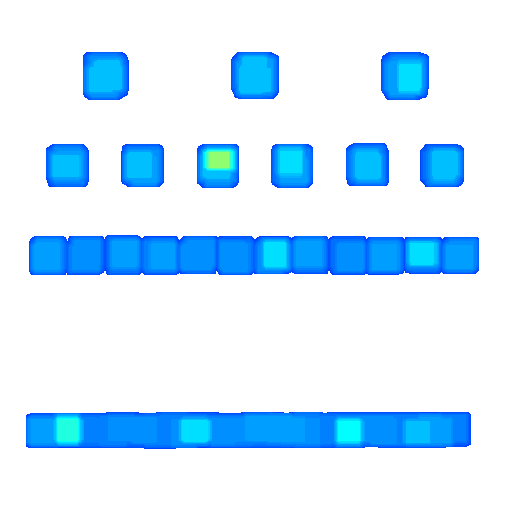};
\end{axis}
\end{tikzpicture}% 
			& 
			% This file was created by matlab2tikz v0.4.7 running on MATLAB 7.11.
% Copyright (c) 2008--2014, Nico Schlömer <nico.schloemer@gmail.com>
% All rights reserved.
% Minimal pgfplots version: 1.3
% 
% The latest updates can be retrieved from
%   http://www.mathworks.com/matlabcentral/fileexchange/22022-matlab2tikz
% where you can also make suggestions and rate matlab2tikz.
% 
\begin{tikzpicture}

\begin{axis}[%
width=\fwidth,
height=\fheight,
axis on top,
scale only axis,
xmin=0.5,
xmax=512.5,
trim axis left,
trim axis right,
ymin=0.5,
ymax=512.5,
xtick = \empty,
ytick = \empty
]
\addplot [forget plot] graphics [xmin=0.5,xmax=512.5,ymin=0.5,ymax=512.5] {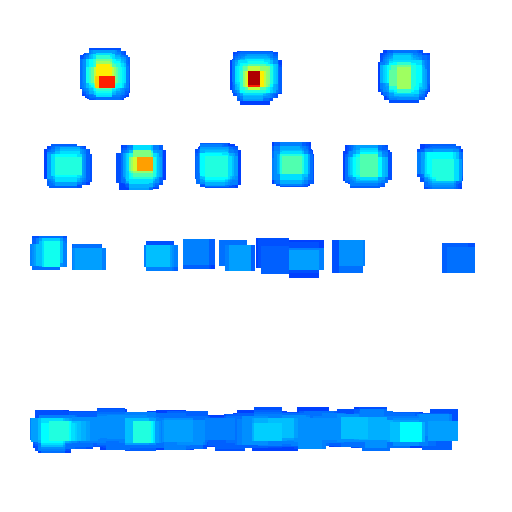};
\end{axis}
\end{tikzpicture}% 
			&
			\begin{tikzpicture}
\begin{axis}[
hide axis,
scale only axis,
height=0pt,
width=0pt,
colormap={mymap}{[1pt] rgb(0pt)=(1,1,1); rgb(1pt)=(0,0,0.625); rgb(7pt)=(0,0,1); rgb(23pt)=(0,1,1); rgb(39pt)=(1,1,0); rgb(55pt)=(1,0,0); rgb(63pt)=(0.5,0,0)},
colorbar,
colorbar/width = .4cm,
colorbar style={
ytick={0,0.00125,0.0025,0.00375,0.005,0.00625},
yticklabels={{},{800},{400},{266.7},{200},{160}},
scaled ticks = false,
width = .4cm,
height = .95\fheight
},
point meta min=0,
point meta max=0.00625
]
\addplot [draw = none] coordinates {(0,0)};
\end{axis}
\end{tikzpicture} \\
			
			\raisebox{-.9cm}{\rotatebox{90}{$\sigma = 0.005$}}
			&
			\begin{tikzpicture}[baseline]

\begin{axis}[
width=4cm,
height=4cm,
axis on top,
scale only axis,
xmin=0.5,
xmax=512.5,
trim axis left,
trim axis right,
ymin=0.5,
ymax=512.5,
xtick = \empty,
ytick = \empty
]
\addplot [forget plot] graphics [xmin=0.5,xmax=512.5,ymin=0.5,ymax=512.5] {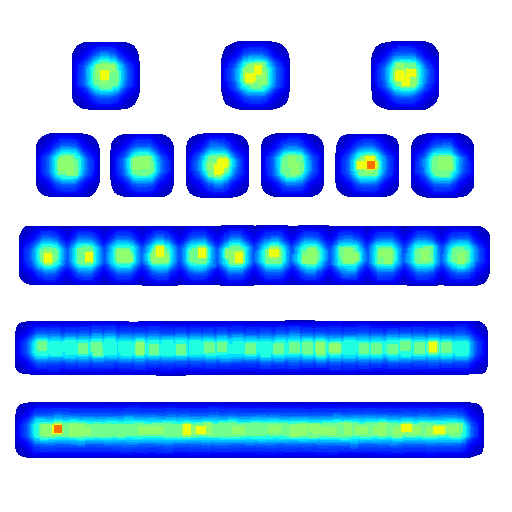};
\end{axis}

\coordinate (C) at (40.5/128,192/128);
\coordinate (D) at (40.5/128,62/128);
\coordinate (E) at (170.5/128,192/128);
\coordinate (F) at (170.5/128,62/128);
\coordinate (G) at (0,-40/128);
\coordinate (H) at (0,-552/128);
\coordinate (I) at (4,-40/128);
\coordinate (J) at (4,-552/128);

\draw[solid, draw=black, line width=1.0pt] (C) rectangle (F);
\draw[dashed, draw=black, line width=.5pt] (C) -- (G);
\draw[dashed, draw=black, line width=.5pt] (D) -- (H);
\draw[dashed, draw=black, line width=.5pt] (E) -- (I);
\draw[dashed, draw=black, line width=.5pt] (F) -- (J);

\begin{axis}[%
width=\fwidth,
height=\fheight,
axis on top,
scale only axis,
xmin=0.5,
xmax=131.5,
trim axis left,
trim axis right,
ymin=0.5,
ymax=131.5,
xtick = \empty,
ytick = \empty,
axis line style = very thick,
at={(0,-140)}
]
\addplot [forget plot] graphics [xmin=0.5,xmax=131.5,ymin=0.5,ymax=131.5] {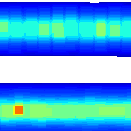};
\end{axis}

\end{tikzpicture}% 
			& 
			\begin{tikzpicture}[baseline]

\begin{axis}[
width=4cm,
height=4cm,
axis on top,
scale only axis,
xmin=0.5,
xmax=512.5,
trim axis left,
trim axis right,
ymin=0.5,
ymax=512.5,
xtick = \empty,
ytick = \empty
]
\addplot [forget plot] graphics [xmin=0.5,xmax=512.5,ymin=0.5,ymax=512.5] {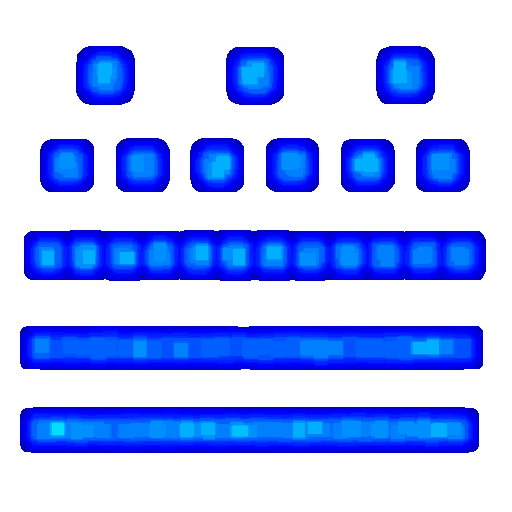};
\end{axis}

\coordinate (C) at (40.5/128,192/128);
\coordinate (D) at (40.5/128,62/128);
\coordinate (E) at (170.5/128,192/128);
\coordinate (F) at (170.5/128,62/128);
\coordinate (G) at (0,-40/128);
\coordinate (H) at (0,-552/128);
\coordinate (I) at (4,-40/128);
\coordinate (J) at (4,-552/128);

\draw[solid, draw=black, line width=1.0pt] (C) rectangle (F);
\draw[dashed, draw=black, line width=.5pt] (C) -- (G);
\draw[dashed, draw=black, line width=.5pt] (D) -- (H);
\draw[dashed, draw=black, line width=.5pt] (E) -- (I);
\draw[dashed, draw=black, line width=.5pt] (F) -- (J);

\begin{axis}[%
width=\fwidth,
height=\fheight,
axis on top,
scale only axis,
xmin=0.5,
xmax=131.5,
trim axis left,
trim axis right,
ymin=0.5,
ymax=131.5,
xtick = \empty,
ytick = \empty,
axis line style = very thick,
at={(0,-140)}
]
\addplot [forget plot] graphics [xmin=0.5,xmax=131.5,ymin=0.5,ymax=131.5] {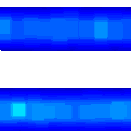};
\end{axis}

\end{tikzpicture}% 
			& 
			\begin{tikzpicture}[baseline]

\begin{axis}[
width=4cm,
height=4cm,
axis on top,
scale only axis,
xmin=0.5,
xmax=512.5,
trim axis left,
trim axis right,
ymin=0.5,
ymax=512.5,
xtick = \empty,
ytick = \empty
]
\addplot [forget plot] graphics [xmin=0.5,xmax=512.5,ymin=0.5,ymax=512.5] {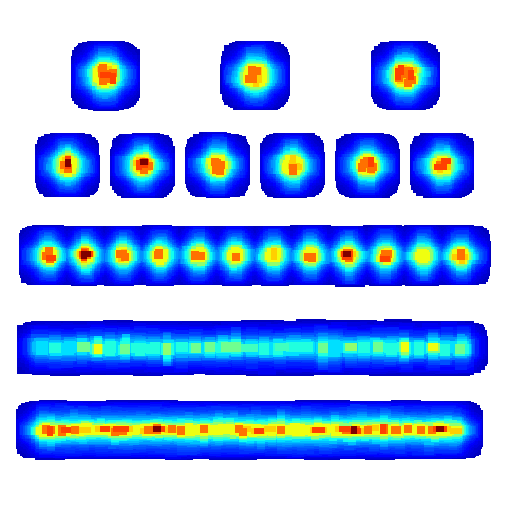};
\end{axis}

\coordinate (C) at (40.5/128,192/128);
\coordinate (D) at (40.5/128,62/128);
\coordinate (E) at (170.5/128,192/128);
\coordinate (F) at (170.5/128,62/128);
\coordinate (G) at (0,-40/128);
\coordinate (H) at (0,-552/128);
\coordinate (I) at (4,-40/128);
\coordinate (J) at (4,-552/128);

\draw[solid, draw=black, line width=1.0pt] (C) rectangle (F);
\draw[dashed, draw=black, line width=.5pt] (C) -- (G);
\draw[dashed, draw=black, line width=.5pt] (D) -- (H);
\draw[dashed, draw=black, line width=.5pt] (E) -- (I);
\draw[dashed, draw=black, line width=.5pt] (F) -- (J);

\begin{axis}[%
width=\fwidth,
height=\fheight,
axis on top,
scale only axis,
xmin=0.5,
xmax=131.5,
trim axis left,
trim axis right,
ymin=0.5,
ymax=131.5,
xtick = \empty,
ytick = \empty,
axis line style = very thick,
at={(0,-140)}
]
\addplot [forget plot] graphics [xmin=0.5,xmax=131.5,ymin=0.5,ymax=131.5] {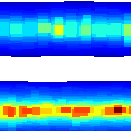};
\end{axis}

\end{tikzpicture}% 
			&
			\raisebox{-4cm}{\begin{tikzpicture}
\begin{axis}[
hide axis,
scale only axis,
height=0pt,
width=0pt,
colormap={mymap}{[1pt] rgb(0pt)=(1,1,1); rgb(1pt)=(0,0,0.625); rgb(7pt)=(0,0,1); rgb(23pt)=(0,1,1); rgb(39pt)=(1,1,0); rgb(55pt)=(1,0,0); rgb(63pt)=(0.5,0,0)},
colorbar,
colorbar/width = .4cm,
colorbar style={
ytick={0,0.00416666666666667,0.00833333333333333,0.0125,0.0166666666666667,0.0208333333333333},
yticklabels={{},{240},{120},{80},{60},{48}},
scaled ticks = false,
width = .4cm,
height = 1.9\fheight
},
point meta min=0,
point meta max=0.0208333333333333,
at={(0,0)}
]
\addplot [draw = none] coordinates {(0,0)};
\end{axis}
\end{tikzpicture}}
		\end{tabular}
		\caption{Support detected by the three different tests at $90\%$-significance level on the data from Figure \ref{fig:2D_testfunction}. Each column corresponds to the results of a fixed test for different noise levels, each row to the results for the different tests for a fixed noise level. The color-coding shows the smallest scale (in pixels) on which the corresponding pixel was significant.}
		\label{fig:2D_significance_maps}
	\end{figure}
	
	In conclusion we find that the test $\Psi_{\mathrm c}$ with correctly specified degree of ill-posedness measures up to our expectations. It detects the large objects even in the large noise regime Figure \ref{fig:2D_testfunction}(b) (where the other two tests do not find any significant boxes), and from the zoomed plot in Figure \ref{fig:2D_significance_maps}, bottom left, it is apparent that $\Psi_{\mathrm c}$ is able to distinguish between the line in the bottom row and the sequence of squares in the second last row, i.e. it is able to separate objects which have a distance of $9$ pixels, which is even less than the FWHM.
	
	It is immediately apparent that the undersmoothing test $\Psi_{\mathrm u}$ is more sensitive on small scales compared to $\Psi_{\mathrm c}$ and $\Psi_{\mathrm o}$, but both $\Psi_{\mathrm o}$ and $\Psi_{\mathrm u}$ cannot keep up with the large scale detection power of $\Psi_{\mathrm c}$, as this test is the only one to detect any signal in the large noise situation, $\sigma = 0.5$. From our theory (cf. Lemma \ref{Lemma:Power3}) we know that a box $B_i$ will be detected as soon as $f_{|_{B_i}} \geq 2\biggl(\frac{q_{1-\alpha}}{\omega_i}+ \omega_i\biggr)\frac{\|\sigma \Phi_{\bh_i}\|_2}{\sqrt{h_{i,1} h_{i,2}n^d}}$ with probability $\geq 1-\alpha$ (the penalties $\omega_i$ are as in Theorem \ref{Th:1} and $ \Phi_{\bh_i}$ is the corresponding transformed mother wavelet as in Section \ref{Sec:Dec}). This provides the largest $\sigma$ which still ensures detection with probability $\geq 1-\alpha$, 
	\[
	\sigma  = \frac{\sqrt{h_{i,1} h_{i,2}n^2}}{\|\Phi_{\bh_i}\|_2} \left(2\biggl(\frac{q_{1-\alpha}}{\omega_i}+ \omega_i\biggr)\right)^{-1}.
	\]
	Numerical approximations of the values of the right-hand side are depicted in Figure \ref{fig:2D_detection_bound} against the size (being $h_{i,1} h_{i,2}$) of the considered scales with $\alpha = 0.1$. For visual appeal, the computed values have been smoothed.
	Comparing $\Psi_{\mathrm c}$ and $\Psi_{\mathrm u}$ it becomes apparent that $\Psi_{\mathrm u}$ allows for more detections on small scales, but will necessarily oversee several weak signals on large scales. Comparing with $\Psi_{\mathrm c}$ and $\Psi_{\mathrm u}$ we find that $\Psi_{\mathrm o}$ has less detection power on all considered scales. Concerning $\Psi_{\mathrm o}$ and $\Psi_{\mathrm u}$ it seems that $\Psi_{\mathrm u}$ has better detection properties on nearly all scales but the largest ones. In conclusion, misspecification of ill-posedness does not provide false detections, but loss in detection power, where underestimating the ill-posedness has less severe effects to MISCAT than overestimation.
	
	\begin{figure}[!htb]
		\setlength\fheight{3cm} \setlength\fwidth{8cm}
		\centering
		\input{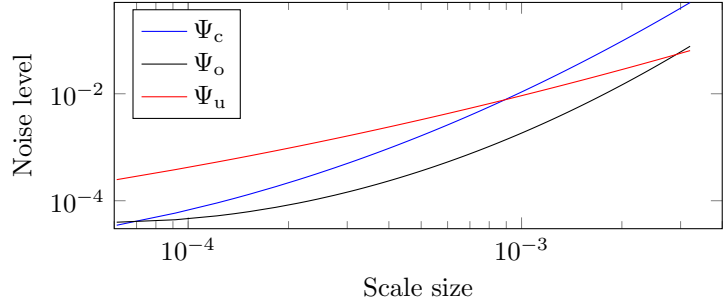}
		\caption{Upper bound for $\sigma$ ensuring a $90\%$ probability for detection of a signal of intensity $1$ on the corresponding scale in the three different test setups.}
		\label{fig:2D_detection_bound}
	\end{figure}
	
	\subsubsection{Robustness to the noise distribution}
	
	For the test $\Psi_{\mathrm c}$ with correctly specified degree of smoothness of the convolution we furthermore investigate the empirical level using the exact variance for different data setups, cf. Table \ref{tab:2D_levels}. We consider three scenarios:
	\begin{itemize}
		\item[(1)] data with Gaussian additive noise $\mathcal N \left(0,\sigma^2\right)$,
		\item[(2)] data with Student's t additive noise $t\left(\nu\right)$ with different degrees of freedom $\nu$,
		\item[(3)] data with mixed Poisson and Gaussian noise.
	\end{itemize}
	
	For scenario (3) the data is generated as follows: For specified parameters $t, \sigma, b>0$ we have
	\[
	Y_{\bj} \stackrel{\text{independent}}{\sim} \frac{1}{t}\text{Poi} \left(t \left[\left(\psf \ast f\right) \left(\bx_{\bj}\right) + b\right] \right) - b + \mathcal N \left(0, \sigma^2\right), \qquad \bj \in \left\{1,...,512\right\}^2.
	\]
	Here $\text{Poi} \left(\lambda\right)$ denotes the Poisson distribution with parameter $\lambda >0$, and the $\mathcal N \left(0, \sigma^2\right)$ distribution is independent of this. This model aims for mimicking the data coming from CCD sensors \citep[see e.\,g.][]{SWH:93,SHLFW:95}, modeling the (known) photon background by $b$, the (known) observation time by $t$, and the read-out errors by $\mathcal N \left(0, \sigma^2\right)$ with known variance $\sigma^2>0$. 
	
	\begin{table}[!htb]
		\centering
		\begin{tabular}{cll}
			\toprule[2pt]
			Noise scenario & Parameters & false positive detections in $\%$ \\
			\midrule[1pt]
			(1) & & $8.8$ \\
			\midrule[1pt]
			\multirow{7}{*}{(2)}
			& $\nu = 3$ & $100$ \\
			& $\nu = 6$ & $94.7$ \\
			& $\nu = 7$ & $72.3$ \\
			& $\nu = 11$ & $21.8$ \\
			& $\nu = 15$ & $15.7$ \\
			& $\nu = 19$ & $13.0$ \\
			& $\nu = 23$ & $13.3$ \\
			\midrule[1pt]
			\multirow{3}{*}{(3)}
			& $t = 100,~ b = 0.5,~~ \sigma = 0.01$ & $9.8$\\
			& $t = 1000, b = 0.005, \sigma = 0.01$ & $8.1$\\
			& $t = 100,~ b = 0.005, \sigma = 0.01$ & $14.5$\\
			\bottomrule[2pt]
		\end{tabular}
		\caption{Some empirical levels computed in $1000$ runs for different data settings. We always choose $\alpha = 0.1$, i.e. (asymptotically) there should $\leq 10\%$ false positives.}
		\label{tab:2D_levels}
	\end{table}
	
	Note that in scenario (1) the level will be independent of $\sigma$ even for finite $n$, because of the standardization of the local test statistics.
	
	In conclusion we find that the heavy tail behavior of the student's t distribution strongly influences the level of the test. Note that $t(\nu)$ has $\nu-1$ moments and hence does not satisfy \eqref{Moments} for any $\nu$. In contrast to that, in the mixture model all moments exist and satisfy \eqref{Moments} whenever $b>0$. Our simulations suggest that the test keeps its level quite stable over a large range of parameters. In the situation of a very low Poisson intensity ($t = 100,~ b = 0.005$), the heavier tail behavior of the Poisson distribution dominates.
	
	\subsection{Locating fluorescent markers in STED super-resolution microscopy} \label{Application}
	
	Based on the results from Section \ref{subsec:2D_simulations}, we are now able to rigorously treat the real world application from Section \ref{intro:appl} from $2$-dimensional STED (stimulated emission depletion) super-resolution microscopy \citep{HW:94,KH:99,H:07}. A brief overview over the experimental setup is already given in the introduction, and for a detailed mathematical model we refer to Appendix \ref{appB}, where we argue there that our measurements are described reasonably by
	\[
	Y_{\bj} \stackrel{\text{independent}}{\sim} \text{Bin}\left(t, \left(\psf_{2,0.016} \ast f\right) \left(\bx_{\bj}\right)\right), \qquad \bj \in \left\{1,...,600\right\}^2.
	\]
	Here $\text{Bin}\left(t,p\right)$ denotes the Binomial distribution with parameters $t \in \mathbb N$ and $p \in \left[0,1\right]$, observations are obtained on the grid $\bigl\{\bx_{\bj} ~\big|~ \bj \in \left\{1, ..., 600\right\}^2\bigr\}$ and $f \left(\bx\right)$ is the probability that a photon emitted at grid point $\bx$ is recorded at the detector in a single excitation pulse. The kernel $\psf_{2,0.016}$ is as in \eqref{Def:PSF}, and in actual experiments $t$ is roughly $10^3$.
	
	With this kernel we design a test using the optimal probe function $\varphi$ in \eqref{Probe} (i.e. $\beta_1 = \beta_2 = 4$) and a set of scales defined by boxes of size $k_x \times k_y$ pixels, $k_x, k_y = 4,6,...,20$. The variances $\sigma_i^2$ in \eqref{scan2} used in the test statistic are estimated from the data point-wise using a maximum likelihood estimator. Furthermore we ease the problem by neglecting all boxes in which no photons where observed, i.e. we drop all pairs $\left(\bt_i, \bh_i\right)$ such that $Y_{\bj}=0$ for all $\bx_{\bj} \in [\bt_i-\bh_i,\bt_i]$. Even though this choice is data dependent and hence random, the uniformity over all pairs $\left(\bt_i, \bh_i\right)$ of our confidence statements ensures that those stay valid. 
	
	With this test we analyze the data shown in Figure \ref{fig:exp_data}, cf. Section \ref{intro:appl} for details. For a comparison, we also use two different tests, namely an analog of MISCAT using only one single scale of size $4 \times 6$ pixels (these are the smallest boxes found by MISCAT, see Remark \ref{BT1}(b)), and the multiscale scanning test ignoring the deconvolution ($T = \text{id}$), boiling down to the test statistic of \citet{DuemSpok2001}:
	\[
	\max_{i} \frac{\sqrt{\log \left(\frac{3}{\bh_{i}^{\mathbf 1}}\right)}}{\log\left(\log\left(\frac{3}{
			\bh_{i}^{\mathbf 1}}\right)\right)} \left[ \frac{1}{\sqrt{\bh_{i}^{\mathbf 1}}}\sum\limits_{\bx_{\bj} \in [\bt_i-\bh_i,\bt_i]} Y_{\bj} - \sqrt{2\log\left(\frac{3}{
			\bh_{i}^{\mathbf 1}}\right)} \right].
	\]
	For all tests we again use empirical quantiles computed in $10^4$ runs of the test statistics applied to Gaussian white noise.
	
	The full result is depicted in Figure \ref{fig:exp_result}. As mentioned in the introduction, MISCAT is able (at least for some of the single DNA origamis) to infer on position and rotation as indicated in the first panel in Figure \ref{fig:exp_result}. Remarkably, this information is not visible by eye, cf. Figure \ref{fig:appl_intro}. 
	
	\begin{figure}[!htb]
		\centering
		\input{exp_result.tikz} 
		\begin{tikzpicture}

\begin{axis}[
width=6cm,
height=6cm,
axis on top,
scale only axis,
xmin=0.5,
xmax=600.5,
xtick = \empty,
ytick = \empty,
y dir = reverse,
ymin=0.5,
ymax=600.5,
colormap={mymap}{[1pt] rgb(0pt)=(1,1,1); rgb(1pt)=(0,0,0.625); rgb(7pt)=(0,0,1); rgb(23pt)=(0,1,1); rgb(39pt)=(1,1,0); rgb(55pt)=(1,0,0); rgb(63pt)=(0.5,0,0)},
colorbar,
colorbar style={
ytick={0,0.0125,0.025,0.0375,0.05,0.0625},
yticklabels={{Not significant},{8000},{4000},{2667},{2000},{1600}},
scaled ticks = false
},
point meta min=0,
point meta max=0.0625
]
\addplot [forget plot] graphics [xmin=0.5,xmax=600.5,ymin=0.5,ymax=600.5] {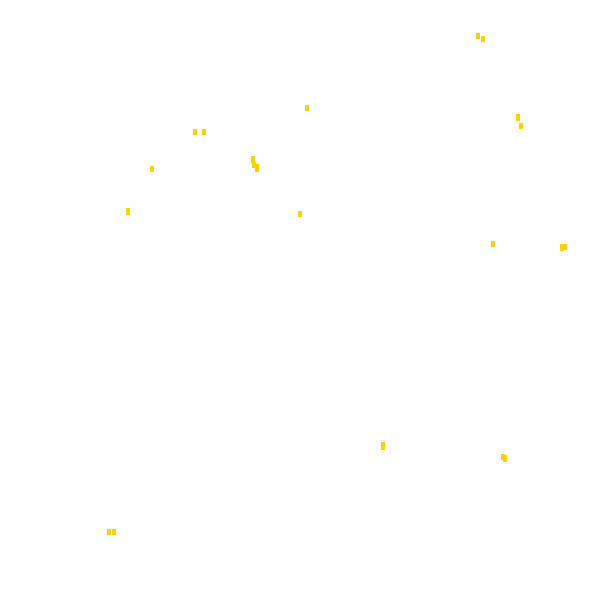};

\coordinate (A) at (150,55);
\coordinate (B) at (199,55);

\draw[color=black,solid,line width=2.0pt] (A) -- (B) node[midway,above] {\tiny 500nm};

\end{axis}

\coordinate (C) at (1.2,3.3);
\coordinate (D) at (2.7,4.8);
\coordinate (E) at (1.2,4.8);
\coordinate (F) at (2.7,3.3);
\coordinate (G) at (-4.775,3.6);
\coordinate (H) at (-1.175,3.6);
\coordinate (I) at (-4.775,0);
\coordinate (J) at (-1.175,0);

\draw[solid, draw=black, line width=1.0pt] (C) rectangle (D);
\draw[dashed, draw=black, line width=.5pt] (C) -- (I);
\draw[dashed, draw=black, line width=.5pt] (D) -- (H);
\draw[dashed, draw=black, line width=.5pt] (E) -- (G);
\draw[dashed, draw=black, line width=.5pt] (F) -- (J);

\begin{axis}[
width=3.6cm,
height=3.6cm,
axis on top,
scale only axis,
xmin=0.5,
xmax=151.5,
axis line style = very thick,
xtick = \empty,
ytick = \empty,
ymin=0.5,
ymax=151.5,
at={(-200,0)}
]
\addplot [forget plot] graphics [xmin=0.5,xmax=151.5,ymin=0.5,ymax=151.5] {exp_single_scale_result_snap-1.png};

\coordinate (K) at (20,135);
\coordinate (L) at (39,135);

\draw[color=black,solid,line width=2.0pt] (K) -- (L) node[midway,above] {\tiny 200nm};
\end{axis}

\node [rotate=90] at (8.4,3.2) {Size of smallest significance in nm$^2$};
\end{tikzpicture}% 
		\\
		\begin{tikzpicture}

\begin{axis}[
width=6cm,
height=6cm,
axis on top,
scale only axis,
xmin=0.5,
xmax=600.5,
xtick = \empty,
ytick = \empty,
y dir = reverse,
ymin=0.5,
ymax=600.5,
colormap={mymap}{[1pt] rgb(0pt)=(1,1,1); rgb(1pt)=(0,0,0.625); rgb(7pt)=(0,0,1); rgb(23pt)=(0,1,1); rgb(39pt)=(1,1,0); rgb(55pt)=(1,0,0); rgb(63pt)=(0.5,0,0)},
colorbar,
colorbar style={
ytick={0,0.0125,0.025,0.0375,0.05,0.0625},
yticklabels={{Not significant},{8000},{4000},{2667},{2000},{1600}},
scaled ticks = false
},
point meta min=0,
point meta max=0.0625
]
\addplot [forget plot] graphics [xmin=0.5,xmax=600.5,ymin=0.5,ymax=600.5] {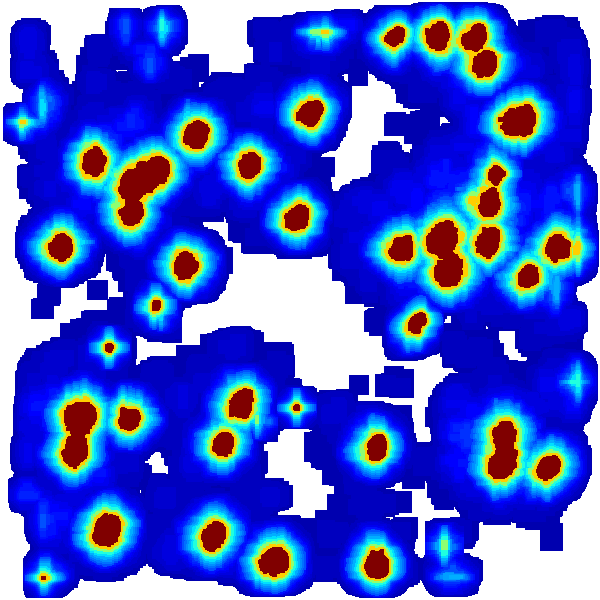};

\coordinate (A) at (150,55);
\coordinate (B) at (199,55);

\draw[color=black,solid,line width=2.0pt] (A) -- (B) node[midway,above] {\tiny 500nm};

\end{axis}

\coordinate (C) at (1.2,3.3);
\coordinate (D) at (2.7,4.8);
\coordinate (E) at (1.2,4.8);
\coordinate (F) at (2.7,3.3);
\coordinate (G) at (-4.775,3.6);
\coordinate (H) at (-1.175,3.6);
\coordinate (I) at (-4.775,0);
\coordinate (J) at (-1.175,0);

\draw[solid, draw=black, line width=1.0pt] (C) rectangle (D);
\draw[dashed, draw=black, line width=.5pt] (C) -- (I);
\draw[dashed, draw=black, line width=.5pt] (D) -- (H);
\draw[dashed, draw=black, line width=.5pt] (E) -- (G);
\draw[dashed, draw=black, line width=.5pt] (F) -- (J);

\begin{axis}[
width=3.6cm,
height=3.6cm,
axis on top,
scale only axis,
xmin=0.5,
xmax=151.5,
axis line style = very thick,
xtick = \empty,
ytick = \empty,
ymin=0.5,
ymax=151.5,
at={(-200,0)}
]
\addplot [forget plot] graphics [xmin=0.5,xmax=151.5,ymin=0.5,ymax=151.5] {exp_scanning_result_snap-1.png};

\coordinate (K) at (20,135);
\coordinate (L) at (39,135);

\draw[color=black,solid,line width=2.0pt] (K) -- (L) node[midway,above] {\tiny 200nm};
\end{axis}

\node [rotate=90] at (8.4,3.2) {Size of smallest significance in nm$^2$};

\end{tikzpicture}% 
		\caption{$90\%$ significance maps and excerpts for different tests computed from the data in Figure \ref{fig:exp_data}. From top to bottom: MISCAT, a single scale test with deconvolution, and the standard multiscale test without deconvolution.}
		\label{fig:exp_result}
	\end{figure}
	
	\section{Multiscale Extreme Value Theory}\label{Sec:Aux}
	The following theorem guarantees that the Gaussian approximation $\mathcal{S}(W)$ is asymptotically  bounded from above almost surely. 	
	\begin{Theorem}[MISCAT: a.s. boundedness]\label{GWSMain}
		Let $\Xi$ be a normed function, i.e. $\|\Xi\|_2=1$, supported on $[0,1]^d$ such that \eqref{S1} holds. Let $(\bt,\bh)\in\mathcal{H}\times\mathcal{T}_{\bh}\subset[\bh_{\min},\bh_{\max}]\times[\bh,\be]$, where $ h_{\max}\leq n^{-\delta}$. There exists a function $F$ which is independent of $n$, such that $\lim_{\lambda\to\infty}F(\lambda)=0$ and for $\lambda>0$
		\begin{align*}	\mathbb{P}\biggl(\sup_{\bh\in\mathcal{H}}\sup_{\bt\in\mathcal{T}_{\bh}}\omega_{\bh}\biggr(\int\Xi\bigl(\bt-\bz\bigr)\,dW_{\bz}-\omega_{\bh}\biggr)>\lambda\biggr)\leq F(\lambda).
		\end{align*} 
		This implies in particular that $\mathcal{S}(W)$ is almost surely bounded. Furthermore, there exists a positive constant $ \underline{D}_{\gamma}$  such that for any fixed $\lambda\in\R$
		\begin{align*}
		e ^{-\underline{D}_{\gamma} \displaystyle e^{-\lambda}}\leq\lim_{n\to\infty}\mathbb{P}\biggl(\sup_{\bh\in\mathcal{H}}\sup_{\bt\in\mathcal{T}_{\bh}}\omega_{\bh}\biggr(\int\Xi\bigl(\bt-\bz\bigr)\,dW_{\bz}-\omega_{\bh}\biggr)\leq\lambda\biggr).
		\end{align*}
	\end{Theorem}
	The following theorem yields a weak limit for multiscale statistics of the type $\mathcal{S}(W)$.		
	\begin{Theorem}[A general multiscale Gumbel limit theorem]\label{Thm:Gumbel} Let $\Xi$ be a normed function, i.e. $\|\Xi\|_2=1$, supported on $[0,1]^d$ such that \eqref{S2} holds.
		Let $K>0$ be a fixed, positive constant.
		If $-\log(h_{\max})=\delta\log(n)+o(\log(n))$ and $-\log(h_{\min})=\Delta\log(n)+o(\log(n))$, with $0 < \delta < \Delta \leq 1$, it holds
		\begin{align*}
		\lim_{n\to\infty}\mathbb{P}\biggl(\sup_{\bh\in[h_{\min},h_{\max}]^d}\sup_{\bt\in[\bh,\be]}\omega_{\bh}\biggr(\int\Xi\bigl(\bt-\bz\bigr)\,dW_{\bz}-\omega_{\bh}\biggr)\leq\lambda\biggr)=e^{-e^{-\lambda}\cdot\frac{H_{2\gamma}|\det D_{\Xi}^{-1}|I_d(\delta,\Delta)}{\sqrt{2\pi}K}},
		\end{align*}
		where  $\omega_{\bh}$ and $I_d(\delta,\Delta)$ are defined as in \eqref{eq:calibration} and \eqref{Def:Id}, respectively.
	\end{Theorem}
	
	Corollary \ref{Cor} below follows immediately from the proofs of the previous Theorems. Two  special cases are discussed in 
	Remark \ref{BT1}.
	\begin{Corollary}\label{Cor}
		Suppose that the assumptions of Theorem \ref{GWSMain} hold.  Assume that $\bh_i\in\mathcal{H}_1\times\ldots\times\mathcal{H}_d$, where possibly $\mathcal{H}_i\neq\mathcal{H}_j$ for $i\neq j.$ Let for 
		$\mathcal{P}:=\bigl\{\lfloor\log(1/h_{\max})\rfloor-2,\lfloor\log(1/h_{\max})\rfloor-1,\ldots,\lceil\log(1/h_{\min})\rceil\bigr\}$,  and $j\in\{1,\ldots,d\}$
		\begin{align*}
		\mathcal{P}_j:=\big\{p\in\mathcal{P}\,\big|\,\exists\;h_{i,j}\in\mathcal{H}_j:\;h_{i,j}\in \bigl[e^{-(p+1)},e^{-p}\bigr)\big\}.
		\end{align*}
		Choose the constant $C_d$ in \eqref{eq:calibration} such that there exist positive constants $\underline{d}$ and $\overline{D} $ such that
		\begin{align*}
		\underline{d}\leq\log(n)^{-\frac{C_d-d/\gamma+1}{2}}\bigl|\mathcal{P}_1\times\ldots\times\mathcal{P}_d\bigr|\leq\overline{D}.
		\end{align*}
		\begin{itemize}
			\item[(a)]  The results of Theorem \ref{GWSMain} remain valid.
			\item[(b)] If, in addition, the grid of positions $\bt$ is sufficiently fine, i.e.
			\eqref{Grid1} holds and for each $j$, the sets $\mathcal{P}_j=\mathcal{P}_{j,n}$ are increasing sets with respect to $n\in\mathbb{N}$, i.\,e., $|\mathcal{P}_{j,n}|\leq|\mathcal{P}_{j,n+1}|$ and $\sum_{p_{j,n}\in\mathcal{P}_{j,n}}p_{j,n}$ is increasing, there exists a constant $C_{\mathcal{P}}>0$ such that
			\begin{align*}
			\lim_{n\to\infty}\mathbb{P}\biggl(\sup_{\bh\in \mathcal{H}_1\times\ldots\times\mathcal{H}_d}\sup_{\bt\in\mathcal{T}_{\bh}}\omega_{\bh}\biggr(\int\Xi\bigl(\bt-\bz\bigr)\,\mathrm{d}W_{\bz}-\omega_{\bh}\biggr)\leq\lambda\biggr)=e^{-e^{-\lambda}\cdot\frac{H_{2\gamma}|\det D_{\Xi}^{-1}|C_{\mathcal{P}}}{\sqrt{2\pi}K}}.
			\end{align*}
			
		\end{itemize}
	\end{Corollary} 
\section{Proofs}\label{Proofs}

Throughout the proofs the letter $C$ without a subscript denotes a generic, positive constant whose value may vary from line to line.

We will first prove the auxiliary results from Section \ref{Sec:Aux}, as they are needed for the proofs of the main results.
\begin{proof}[Proof of Theorem \ref{Thm:Gumbel}]
	Recall that we denote
	\begin{align*}
	\bh^{\pmb{\alpha}}=h_1^{\alpha_1}\cdot\ldots\cdot h_d^{\alpha_d},\quad\bh_{\bp}=(h_1,\ldots,h_d)^T,\quad\frac{\be}{\bh_{\mathbf{p}}}=\biggl(\frac{1}{h_1},\ldots,\frac{1}{h_d}\biggr)^T
	\end{align*}
	and inequalities between vectors or multi-indices are meant component-wise.\\
	\textbf{Step I: Proof for scales on a dyadic grid.}~\\
	We show later in Step II.3 that the supremum over $[h_{\min}/\grid,h_{\max}]$ and supremum over $[h_{\min},h_{\max}]$ are asymptotically equivalent and consider first the supremum over the slightly enlarged set.
	Define a dyadic grid $\mathcal{H}_{\mathrm{dyad}}\subset[h_{\min}/\grid,h_{\max}]$ as follows:
	\begin{align*}
	\mathcal{H}_{\mathrm{dyad}}:=\{2^{-p}\,|\,p\in\mathcal{P}\}, \quad\mathcal{P}=\big\{\big\lceil b_{\gamma}(h_{\max})\Big\rceil,\ldots,\big\lfloor b_{\gamma}(h_{\min})\big\rfloor\big\},
	\end{align*}
	where $b_{\gamma}(h):=\log\big(\frac{\grid}{h}\big)/\log(2)$.
	Here and below $\log$ denotes the natural logarithm. Define $p_{\min}:=\min\mathcal{P}$ and $p_{\max}:=\max\mathcal{P}$. 
	Let
	\begin{align}\label{Mn}
	&\max_{\bh\in\mathcal{H}_{\mathrm{dyad}}^d}\sup_{\bt\in[\bh,\be]}\omega_{\bh}\biggl(\frac{1}{\sqrt{\bh^{\be}}}\int\Xi\Bigl(\frac{\bt-\bz}{\bh}\Bigr)\,\mathrm{d}W_{\bz}-\omega_{\bh}\biggr)\notag\\
	&\stackrel{\mathcal{D}}{=}\max_{\bh\in\mathcal{H}_{\mathrm{dyad}}^d}\sup_{\bt\in[\be,\be/\bh]}\omega_{\bh}\biggl(\int\Xi(\bt-\bz)\,\mathrm{d}W_{\bz}-\omega_{\bh}\biggr)=:M_n,
	\end{align}
	by stationarity of $Z_{\bt,\bh}$ for fixed $\bh$.
	We now consider the term $M_n.$\\
	\textbf{Step I.1: Partition of the parameter set}~\\
	The form of $M_n$ in \eqref{Mn} reveals a redundancy pattern that we will exploit later on. Observe that the suprema with respect to $\bt$ of the rescaled version $M_n$ are taken over subsets of the rectangle $[\be,\be/\bh_{\min}].$ For smaller scales, the supremum with respect to $\bt$ is taken over larger sets. Obviously, for $\bp\in\mathcal{P}^d,$
	\begin{align*}
	[\be/\bh_{\bp},\be/\bh_{\bp+\be}]\subset [\be,\be/\bh_{\bs}]\quad\forall\quad \;\bs>\bp+\be.
	\end{align*}   
	In order to exploit this fact we partition the parameter set $[\be,\be/\bh_{\min}]$ into suitable blocks, i.e. into blocks $B_{\bp+\be,\bq+\be}$ that are approximately equal to $[\be/\bh_{\bp},\be/\bh_{\bp+\be}]$
	in order to split the suprema with respect to $\bt$ into suitable sub-suprema. To achieve that those  sub-suprema are independent, we separate the blocks by small bands of width 1. This ensures independence, since supp$(\Xi)\subset[0,1]^d.$ The bands only yield a contribution which is asymptotically negligible, which we will show in Step I.3 below.
	
	To be precise, we define subsets of $[\be,\be/\bh_{\min}]$ as follows
	\begin{align}\label{partition1}
	B_{\bp}&:=\bigl[\tfrac{\be}{\bh_{\bp-\be}},\tfrac{\be}{\bh_{\bp}}-\be\bigr],\;\;\text{and}\;\;
	R_{\bp}=\bigl[\tfrac{\be}{\bh_{\bp-\be}},\tfrac{\be}{\bh_{\bp}}\bigr]\backslash B_{\bp},
	\end{align}
	where $\bh_{\bp_{\min}-1}:=\be.$		
	The large blocks $B_{\bp}$  yield the main contributions. The sets $R_{\bp}$ are asymptotically negligible (see Step I.3 below).
	Define further for $\bq\in\mathcal{P}^d$
	\begin{align*}
	\mathcal{B}_{\bq}:=\bigcup_{\substack{ \bp\in\mathcal{P}^d,\, \bp\leq \bq}}B_{\bp}\quad\text{and}\quad M_{\mathcal{B}}:=\max_{\bp\in\mathcal{P}^d}\omega_{\bh_{\bp}}\Bigl(M_{\mathcal{B}_{\bp}}-\omega_{\bh_{\bp}}\Bigr).
	\end{align*}
	Write
	\begin{align*}
	M_{\mathcal{B}}&=\max_{\bq\in\mathcal{P}^d}\max_{\bp\leq\bq}\sup_{\bt\in B_{\bp}}\omega_{\bh_{\bq}}\biggl(\int\Xi(\bt-\bz)\,\mathrm{d}W_{\bz}-\omega_{\bh_{\bq}}\biggr)\\
	&=\max_{\bp\in\mathcal{P}^d}\max_{\bq\geq\bp}\sup_{\bt\in B_{\bp}}\omega_{\bh_{\bq}}\biggl(\int\Xi(\bt-\bz)\,\mathrm{d}W_{\bz}-\omega_{\bh_{\bq}}\biggr).
	\end{align*}
	Fix $\lambda\in\R.$ Since the blocks $B_{\bp}$ are constructed such that the sub-maxima over different blocks are independent, we find
	\begin{align*}
	\mathbb{P}\bigl(M_{\mathcal{B}}\leq\lambda\bigr)
	&=\prod_{\bp\in\mathcal{P}^d}\mathbb{P}\biggr(\max_{\bp\leq\bq}\sup_{\bt\in B_{\bp}}\omega_{\bh_{\bq}}\biggl(\int\Xi(\bt-\bz)\,\mathrm{d}W_{\bz}-\omega_{\bh_{\bq}}\biggr) \leq \lambda\biggr)\\
	&=\prod_{\bp\in\mathcal{P}^d}\mathbb{P}\biggr(\sup_{\bt\in B_{\bp}}\int\Xi(\bt-\bz)\,\mathrm{d}W_{\bz} \leq \Lambda_{\min,\bp}\biggr),
	\end{align*}
	where
	\begin{align*}
	\Lambda_{\min,\bp}:=\min_{\bp\leq\bq}\biggl(\frac{\lambda}{\omega_{\bh_{\bq}}}+\omega_{\bh_{\bq}}\biggr).
	\end{align*}
	
	Now we have broken the proof down into $|\mathcal{P}|^d$ \textquotedblleft one-scale\textquotedblright~ extreme value problems and  use standard results for those. Let $\mathrm{Leb}(B_{\bp})$ denote the Lebesgue-measure of $B_{\bp}$ and let $\Lambda_{\bp}$ denote
	\begin{align}\label{Lambdakl}
	\Lambda_{\bp}:=\frac{\lambda}{\omega_{\bh_{\bp}}}+\omega_{\bh_{\bp}}.
	\end{align}
	For any fixed $\lambda\in\mathbb{R}$ we have that
	\begin{align*}
	\Lambda_{\min,\bp}=\Lambda_{\bp},
	\end{align*}
	for sufficiently large $n.$ Thus,
	\begin{align*}
	\mathbb{P}\bigl(M_{\mathcal{B}}\leq\lambda\bigr)
	&=\prod_{\bp\in\mathcal{P}^d}\mathbb{P}\biggr(\sup_{\bt\in B_{\bp}}\int\Xi(\bt-\bz)\,\mathrm{d}W_{\bz} \leq \Lambda_{\bp}\biggr)\\
	&=\prod_{\bp\in\mathcal{P}^d}\biggl(1-\mathbb{P}\biggr(\sup_{\bt\in B_{\bp}}\int\Xi(\bt-\bz)\,\mathrm{d}W_{\bz} > \Lambda_{\bp}\biggr)\biggr).
	\end{align*}
	\textbf{Step I.2: Derivation of the weak limit on the dyadic grid.}~\\
	To proceed, we need the following definition and theoretical result from the monograph by \citet{piterbarg1996}.
	\begin{Definition}[slowly blowing up sets]
		A system of sets $A_u,\,u>0$ is said to blow up slowly with rate $\kappa>0$ if each of the sets contains a unit cube and
		\begin{align*}
		\Leb(A_u)=O(e^{\kappa u^2/2}),\quad\text{as}\quad n\to\infty.
		\end{align*}
	\end{Definition}
	\begin{Theorem}[Theorem 7.2 in \cite{piterbarg1996}]\label{Thm:Pit}
		If Assumption \eqref{S2} holds, there exists a constant $\kappa>0$ such that for any system of closed Jordan-sets, blowing up slowly with the rate $\kappa$ we have
		\begin{align}\label{eq:Pit}
		\mathbb{P}\Big(\max_{\bt \in A_u}X(\bt)>u\Big)=H_{2\gamma}\mathrm{mes}(A_u)\big|\mathrm{det}D_{\Xi}^{-1}\big|u^{\frac{d}{\gamma}}\tail(u)(1+o(1)),
		\end{align}
		where $H_{2\gamma}$ denotes Pickands' constant for the rescaled process, i.\,e. the process satisfying \eqref{S2} with $D_{\Xi}=\mathrm{id}$ and $\tail$ denotes the tail function of the standard normal distribution.
	\end{Theorem}

	Next, we estimate
	\begin{align*}
	P_{n,\bp}(\lambda):=\mathbb{P}\biggr(\sup_{\bt\in B_{\bp}}\int\Xi(\bt-\bz)\,\mathrm{d}W_{\bz} \leq \Lambda_{\bp}\biggr)
	\end{align*}
	using Theorem \ref{Thm:Pit}, where $u\triangleq\Lambda_{\bp}.$
	By Theorem \ref{Thm:Pit} we can conclude that there exists a constant $\kappa>0$ such that \eqref{eq:Pit} holds. Since
	\begin{align*}
	e^{\kappa\frac{\Lambda_{k}}{2}}\sim e^{\kappa\lambda}\Bigl(\frac{K}{\bh_{\bp}^{\be}}\Bigr)^{\kappa}\biggl(2\log\Bigl(\frac{K}{\bh_{\bp}^{\be}}\Bigr)\biggr)^{\frac{\kappa\cdot C_d}{2}}
	\end{align*}
	and
	\begin{align*} 
	\mathrm{mes}(B_{\mathbf{p}})=\Leb(B_{\mathbf{p}})\sim\prod_{j=1}^{d}\Bigl(\frac{1}{h_{p_j-1}}-\frac{1}{h_{p_j}}\Bigr)=\frac{1}{2^d}\prod_{j=1}^{d}\Bigl(\frac{1}{h_{p_j}}\Bigr),
	\end{align*}
	the sets $B_{\mathbf{p}}$ are too large if $\kappa<1,$ hence we need to split them up into subsets such that Theorem \ref{Thm:Pit} is applicable. 
	Divide each side of $B_{\bp}$ into $N_{\kappa}^{k_j}:=\lceil\frac{1}{2}h_{p_j}^{\kappa-1}\rceil$ subintervals of equal length $L_{\kappa}^{p_j}:=\frac{1}{2}\frac{1}{h_{p_j}}\lceil\frac{1}{2}h_{p_j}^{\kappa-1}\rceil^{-1}$. This yields a division of $B_{\bp}$ into $N_{\kappa}^{\bp}:=\prod_{j=1}^dN_{\kappa}^{p_j}$ sub-blocks $SB_{\bp,1},\ldots,SB_{\bp,N_{\kappa}^{\bp}}$. As we will show in Step I.3, we obtain for some $\delta_{SB}>0,$
	\begin{align*}
	P_{n,\bp}(\lambda)&=\prod_{j=1}^{N_{\kappa}^{\bp}}\mathbb{P}\biggl(\sup_{\bt\in SB_{\bp,j}}\int\Xi(\bt-\bz)\,\mathrm{d}W_{\bz} \leq \Lambda_{\bp}\biggr)+o\big(n^{-\delta_{SB}}\big)\\
	&=\prod_{j=1}^{N_{\kappa}^{\bp}}\mathbb{P}\biggl(\sup_{\bt\in [\mathbf{0},L_{\kappa}^{\mathbf{k}}]}\int\Xi(\bt-\bz)\,\mathrm{d}W_{\bz} \leq \Lambda_{\bp}\biggr)+o\big(n^{-\delta_{SB}}\big)\\
	&=\biggl(\mathbb{P}\biggl(\sup_{\bt\in [\mathbf{0},L_{\kappa}^{\mathbf{k}}]}\int\Xi(\bt-\bz)\,\mathrm{d}W_{\bz} \leq \Lambda_{\bp}\biggr)\biggr)^{N_{\kappa}^{\bp}}+o\big(n^{-\delta_{SB}}\big),
	\end{align*} 
	where $L_{\kappa}^{\bp}=\prod_{j=1}^dL_{\kappa}^{k_j},$ since for each $\bp$ all sub-blocks are of equal length. Hence,
	\begin{align*}
	P_n(\lambda)&:=\prod_{\bp\in\mathcal{P}^d}P_{n,\bp}(\lambda)=\prod_{\bp\in\mathcal{P}^d}\biggl\{\biggl(1-\mathbb{P}\biggl(\sup_{\bt\in [\mathbf{0},L_{\kappa}^{\mathbf{k}}]}\int\Xi(\bt-\bz)\,\mathrm{d}W_{\bz} > \Lambda_{\bp}\biggr)\biggr)^{N_{\kappa}^{\bp}}+o\big(n^{-\delta_{SB}}\big)\biggr\}\\
	&=\biggl\{\prod_{\bp\in\mathcal{P}^d}\biggl(1-\mathbb{P}\biggl(\sup_{\bt\in [\mathbf{0},L_{\kappa}^{\mathbf{k}}]}\int\Xi(\bt-\bz)\,\mathrm{d}W_{\bz} > \Lambda_{\bp}\biggr)\biggr)^{N_{\kappa}^{\bp}}\biggr\}\cdot\Bigl\{1+o\Bigl(\#\mathcal{P}^dn^{-\delta_{SB}}\Bigr)\Bigr\}.
	\end{align*}
	An application of Theorem \ref{Thm:Pit} yields
	\begin{align*}
	P_n(\lambda)=\prod_{\bp\in\mathcal{P}^d}\biggl(1-(1+o(1))\tail(\Lambda_{\bp})\Lambda_{\bp}^{\frac{d}{\gamma}}L_{\kappa}^{\bp}H_{2\gamma}|\det D_{\Xi}^{-1}|\biggr)\biggr)^{N_{\kappa}^{\bp}}(1+o(1)).
	\end{align*}
	Using 
	\begin{align*}
	\tail(u)=\frac{1}{\sqrt{2\pi}}\exp\Bigl(-\frac{1}{2}u^2\Bigr)u^{-1}\big(1+o(1)\big),\quad\text{as}\quad u\to\infty,
	\end{align*}
	and inserting $\Lambda_{\bp}=\frac{\lambda}{\omega_{\bh_{\bp}}}+\omega_{\bh_{\bp}}$ yields
	\begin{align*}
	P_n(\lambda)&=\prod_{\bp\in\mathcal{P}^d}\biggl(1-(1+o(1))e^{-\frac{1}{2}\Lambda_{\bp}^2}\Lambda_{\bp}^{\frac{d}{\gamma}-1}L_{\kappa}^{\bp}\frac{H_{2\gamma}|\det D_{\Xi}^{-1}|}{\sqrt{2\pi}}\biggr)\biggr)^{N_{\kappa}^{\bp}}(1+o(1))\\
	&=\prod_{\bp\in\mathcal{P}^d}\biggl(1-(1+o(1))e^{-\lambda}\frac{h_{\bp}^{\be}}{K}\Bigl(2\log\Bigl(\tfrac{K}{\bh_{\bp}^{\be}}\Bigr)\Bigr)^{-\frac{C_d}{2}}\Lambda_{\bp}^{\frac{d}{\gamma}-1}L_{\kappa}^{\bp}\frac{H_{2\gamma}|\det D_{\Xi}^{-1}|}{\sqrt{2\pi}}\biggr)\biggr)^{N_{\kappa}^{\bp}}(1+o(1))\\
	&=\prod_{\bp\in\mathcal{P}^d}\biggl(1-(1+o(1))\frac{1}{2^d}e^{-\lambda}\Bigl(2\log\Bigl(\tfrac{K}{\bh_{\bp}^{\be}}\Bigr)\Bigr)^{-\frac{C_d}{2}}\Lambda_{\bp}^{\frac{d}{\gamma}-1}\frac{H_{2\gamma}|\det D_{\Xi}^{-1}|}{K\sqrt{2\pi}}\frac{1}{N_{\kappa}^{\bp}}\biggr)\biggr)^{N_{\kappa}^{\bp}}(1+o(1))\\
	&=\prod_{\bp\in\mathcal{P}^d}\biggl(1-(1+o(1))\frac{1}{2^d}e^{-\lambda}\Bigl(2\log\Bigl(\tfrac{K}{\bh_{\bp}^{\be}}\Bigr)\Bigr)^{\frac{d}{2\gamma}-\frac{C_d-1}{2}}\frac{H_{2\gamma}|\det D_{\Xi}^{-1}|}{K\sqrt{2\pi}}\frac{1}{N_{\kappa}^{\bp}}\biggr)\biggr)^{N_{\kappa}^{\bp}}(1+o(1)).
	\end{align*}
	Plugging in  $C_d=2d+d/\gamma-1$, yields
	\begin{align*}
	P_n(\lambda)&\sim\prod_{\bp\in\mathcal{P}^d}\biggl(1-(1+o(1))e^{-\lambda}\Bigl(\log\Bigl(\frac{K}{\bh_{\bp}^{\be}}\Bigr)\Bigr)^{-d}\frac{H_{2\gamma}|\det D_{\Xi}^{-1}|}{K\sqrt{2\pi}}\frac{1}{N_{\kappa}^{\bp}}\biggr)\biggr)^{N_{\kappa}^{\bp}}.
	\end{align*}
	For sufficiently large $n$ and any fixed $\lambda\in\R$ we have that
	\begin{align*}
	0<\big(1+o(1)\big)e^{-\lambda}\bigg(\log\bigg(\frac{K}{\bh_{\bp}^{\be}}\bigg)\bigg)^{-d}\frac{H_{2\gamma}|\mathrm{det}D_{\Xi}^{-1}|}{K\sqrt{2\pi}}=:x_{n,\bp}(\lambda)<1.
	\end{align*}
	Hence,
	\begin{align*}
	&\bigg|P_n(\lambda)-\exp\bigg(-\sum_{\bp\in\mathcal{P}^d}x_{n,\bp}(\lambda)\bigg)\bigg|
	=\bigg|\exp\bigg(\sum_{\bp\in\mathcal{P}^d}N_{\kappa}^{\mathbf{p}}\log\bigg(1-\frac{x_{n,\mathbf{p}}}{N_{\kappa}^{\mathbf{p}}}\bigg)\bigg)-\exp\bigg(-\sum_{\bp\in\mathcal{P}^d}x_{n,\bp}(\lambda)\bigg)\bigg|\\
	&~\leq\biggl|\exp\bigg(-\sum_{\bp\in\mathcal{P}^d}N_{\kappa}^{\mathbf{p}}\bigg(\frac{x_{n,\mathbf{p}}(\lambda)}{N_{\kappa}^{\mathbf{p}}}+3\frac{x_{n,\mathbf{p}}(\lambda)^2}{N_{\kappa}^{2\mathbf{p}}}\bigg)\bigg)-\exp\bigg(-\sum_{\bp\in\mathcal{P}^d}x_{n,\bp}(\lambda)\bigg)\biggr|\\
	&~\leq\exp\bigg(-\sum_{\bp\in\mathcal{P}^d}x_{n,\bp}(\lambda)\bigg)\bigg|1+\exp\bigg(-3\sum_{\bp\in\mathcal{P}^d}\frac{x_{n,\mathbf{p}}^2(\lambda)}{N_{\kappa}^{2\mathbf{p}}}\bigg)\bigg|,
	\end{align*}
	where we used the series expansion of $x\mapsto\log(1+x),$ i.\,e.
	\begin{align*}
	\log(1+x)=\sum_{k=1}^{\infty}(-1)^{k+1}\frac{x^k}{k},\quad |x|<1,
	\end{align*}
	and the estimate
	\begin{align*}
	\sum_{k=3}^{\infty}\frac{1}{k}\biggl(\frac{x_{n}(\lambda)}{N_{\kappa}^{\mathbf{p}}}\biggr)^{k}
	=\biggl(\frac{x_{n}(\lambda)}{N_{\kappa}^{\mathbf{p}}}\biggr)^{2}\sum_{k=3}^{\infty}\frac{1}{k}\biggl(\frac{x_{n}(\lambda)}{N_{\kappa}^{\mathbf{p}}}\biggr)^{k-2}\leq \biggl(\frac{x_{n}(\lambda)}{N_{\kappa}^{\mathbf{p}}}\biggr)^{2}\sum_{k=3}^{\infty}\frac{1}{k^2}
	\leq\frac{5}{2}\biggl(\frac{x_{n}(\lambda)}{N_{\kappa}^{\mathbf{p}}}\biggr)^{2},
	\end{align*}
	for sufficiently large $n\in\mathbb{N}.$
	Thus, we obtain 
	\begin{align*}
	P_n(\lambda)&\sim\prod_{\bp\in\mathcal{P}^d}\exp\biggl(-(1+o(1))e^{-\lambda}\Bigl(\log\Bigl(\frac{K}{\bh_{\bp}^{\be}}\Bigr)\Bigr)^{-d}\frac{H_{2\gamma}|\det D_{\Xi}^{-1}|}{K\sqrt{2\pi}}\biggr)\\
	&=\exp\biggl(-e^{-\lambda}\frac{H_{2\gamma}|\det D_{\Xi}^{-1}|}{K\sqrt{2\pi}}\sum_{\bp\in\mathcal{P}^d}\Bigl(\log\Bigl(\frac{K}{\bh_{\bp}^{\be}}\Bigr)\Bigr)^{-d}\biggr)\Bigl(1+\exp(-\exp(-\lambda)n^{-\delta_e})\Bigr),
	\end{align*}
	for some $\delta_e>0.$
	We now estimate the sum $\sum_{\bp\in\mathcal{P}^d}\Bigl(\log\Bigl(\frac{K}{\bh_{\bp}^{\be}}\Bigr)\Bigr)^{-d}.$ Write
	\begin{align*}
	\sum_{\bp\in\mathcal{P}^d}\Bigl(\log\Bigl(\frac{K}{\bh_{\bp}^{\be}}\Bigr)\Bigr)^{-d} =\sum_{p_1=p_{\min}}^{p_{\max}}\sum_{p_2=p_{\min}}^{p_{\max}}\ldots\sum_{p_d=p_{\min}}^{p_{\max}}\biggl(\frac{1}{\log(K)+p_1+\cdots+p_d}\biggr)^d.
	\end{align*}
	Note that for a positive, monotonically decreasing function $g$ we have that
	\begin{align}\label{eq:EstSumInt}
	\int_{p_{\min}}^{p_{\max}+1}g(x)\,\mathrm{d}x\leq\sum_{p=p_{\min}}^{p_{\max}}g(p)\leq\int_{p_{\min}-1}^{p_{\max}}g(x)\,\mathrm{d}x.
	\end{align}
	Applying \eqref{eq:EstSumInt} subsequently to the above sums yields
	\begin{align*}
	\sum_{\bp\in\mathcal{P}^d}\Bigl(\log\Bigl(\frac{K}{\bh_{\bp}^{\be}}\Bigr)\Bigr)^{-d} &\sim\int_{[p_{\min},p_{\max}]^d}\biggl(\frac{1}{\log(K)+z_1+\cdots+z_d}\biggr)^d\,\mathrm{d}\bz\\
	&\sim\int_{[\delta\log(n),\Delta\log(n)]^d}\biggl(\frac{1}{\log(K)+z_1+\cdots+z_d}\biggr)^d\,\mathrm{d}\bz=:I_{n,d}
	\end{align*}
	By induction with respect to $d\in\mathbb{N}$ we now show that 
	\begin{align}\label{eq:induction}
	I_{n,d}=\frac{(-1)^{d-1}}{(d-1)!}\sum_{k=0}^d(-1)^k\binom{d}{k}\log(\log(K)+(k\delta+(d-k)\Delta)\log(n)).
	\end{align}
	For $d=1$, we find
	\begin{align*}
	\int_{[\delta\log(n),\Delta\log(n)]}\biggl(\frac{1}{\log(K)+z}\biggr)^d\,\mathrm{d}z= \log\big(\log(K)+\Delta\log(n)\big)-\log\big(\log(K)+\delta\log(n)\big).
	\end{align*}
	Hence,  assertion \eqref{eq:induction} holds for $d=1.$ Next, we consider $I_{n,d+1}.$ We have
	\begin{align*}
	I_{n,d+1}&=\int_{[\delta\log(n),\Delta\log(n)]^{d+1}}\biggl(\frac{1}{\log(K)+z_1+\cdots+z_d+z_{d+1}}\biggr)^{d+1}\,\mathrm{d}\bz\\
	&=-\frac{1}{d}\bigg(\int_{[\delta\log(n),\Delta\log(n)]^d}\biggl(\frac{1}{\log(K)+\Delta\log(n)+z_1+\cdots+z_d}\biggr)^d\,\mathrm{d}(z_1,\ldots,z_d)\\
	&\qquad-\int_{[\delta\log(n),\Delta\log(n)]^d}\biggl(\frac{1}{\log(K)+\delta\log(n)+z_1+\cdots+z_d}\biggr)^d\,\mathrm{d}(z_1,\ldots,z_d)\bigg).
	\end{align*}
	Now we plug in \eqref{eq:induction}  twice and obtain
	\begin{align*}
	I_{n,d+1}&=\frac{(-1)^{d}}{(d)!}\biggl(\sum_{k=0}^d(-1)^k\binom{d}{k}\log(\log(K)+(k\delta+(d+1-k)\Delta)\log(n))\\
	&\qquad-\sum_{k=0}^d(-1)^k\binom{d}{k}\log(\log(K)+((k+1)\delta+(d-k)\Delta)\log(n))\biggr).
	\end{align*}
	An index shift in the second sum yields
	\begin{align*}
	I_{n,d+1}&=\frac{(-1)^{d}}{(d)!}\biggl(\sum_{k=1}^d(-1)^k\biggl[\binom{d}{k}-\binom{d}{k-1}\biggr]\log(\log(K)+(k\delta+(d+1-k)\Delta)\log(n))\\
	&\qquad+\log(\log(K)+(d+1)\Delta\log(n))-(-1)^{d+2}\log(\log(K)+(d+1)\delta\log(n))\biggr)\\
	&=\frac{(-1)^{d}}{(d)!}\sum_{k=0}^{d+1}(-1)^k\binom{d+1}{k}\log(\log(K)+(k\delta+(d+1-k)\Delta)\log(n)).
	\end{align*}
	The last identity follows using the recursive relation of the binomial coefficient:
	\begin{align*}
	\binom{d}{j}=\binom{d-1}{j-1}+\binom{d-1}{j}
	\end{align*}
	and concludes the proof of \eqref{eq:induction}.
	Furthermore,
	\begin{align*}
	\frac{(d-1)!}{(-1)^d}I_{n,d}=\log\left(\frac{\log(n)^{\sum\limits_{j \text{ even}}\binom{d}{j}}\displaystyle\prod_{j\text{ even}}(k\delta+(d-k)\Delta)^{\binom{d}{j}}}{\log(n)^{\sum\limits_{j \text{ odd}}\binom{d}{j}}\displaystyle\prod_{j\text{ odd}}(k\delta+(d-k)\Delta)^{\binom{d}{j}}}\right)=\log\left(\frac{\displaystyle\prod_{j\text{ even}}(k\delta+(d-k)\Delta)^{\binom{d}{j}}}{\displaystyle\prod_{j\text{ odd}}(k\delta+(d-k)\Delta)^{\binom{d}{j}}}\right),
	\end{align*}
	since  $\sum_{j\text{ even}}\binom{d}{j}=\sum_{j\text{ odd}}\binom{d}{j}=2^{d-1}$, which follows using the recursive relation of the binomial coefficient again.
	Hence, the statement of the theorem holds true for scales on the dyadic grid.\\
	\smallskip
	
	\textbf{Step I.3: Negligibility of the remainder terms.}~\\
	We first show that, asymptotically, the slight enlargement of the domain of the scales from the beginning of Step I  does not have an impact. 
	\begin{align*}
	&\mathbb{P}\Big( \Big|\sup_{\bh\in[2^{-p_{\min}},h_{\max}]^d}\sup_{\bt\in[\bh,\be]}\omega_{\bh}(Z_{\bt,\bh}-\omega_{\bh})-\sup_{\bh\in[h_{\min},h_{\max}]^d}\sup_{\bt\in[\bh,\be]}\omega_{\bh}(Z_{\bt,\bh}-\omega_{\bh})\Big|>\varepsilon\Big)\\
	&=\mathbb{P}\Big( \sup_{\bh\in[2^{-p_{\min}},h_{\max}]^d}\sup_{\bt\in[\bh,\be]}\omega_{\bh}(Z_{\bt,\bh}-\omega_{\bh})-\sup_{\bh\in[h_{\min},h_{\max}]^d}\sup_{\bt\in[\bh,\be]}\omega_{\bh}(Z_{\bt,\bh}-\omega_{\bh})>\varepsilon\Big)\\
	&~\leq2\mathbb{P}\Big( \sup_{\bh\in[2^{-p_{\min}},h_{\min}]^d}\sup_{\bt\in[\bh,\be]}\omega_{\bh}(Z_{\bt,\bh}-\omega_{\bh})>\varepsilon\Big).
	\end{align*}
	Furthermore,
	\begin{align*}
	&\mathbb{P}\Big( \sup_{\bh\in[2^{-p_{\min}},h_{\min}]^d}\sup_{\bt\in[\bh,\be]}\omega_{\bh}(Z_{\bt,\bh}-\omega_{\bh})>\varepsilon\Big)\\
	&~\leq
	\mathbb{P}\Big( \sup_{\bh\in[2^{-p_{\min}},2^{-p_{\min}+\lceil4\log\log(n)\rceil}]^d}\sup_{\bt\in[\bh,\be]}\omega_{\bh}(Z_{\bt,\bh}-\omega_{\bh})>\varepsilon\Big)=o(1),
	\end{align*}
	which is an immediate consequence of the previous calculations since the region $[2^{-p_{\min}},2^{-p_{\min}+\lceil4\log\log(n)\rceil}]^d$ is covered by a dyadic grid of cardinality $O\big(\log\log(n)^d\big).$
	Next, we show that the contribution of the separating regions, $\mathcal{R}_{\bp},\,\bp\in\mathcal{P}^d$ are asymptotically negligible. We write
	\begin{align*}
	M_{n,\mathrm{dyad}}=\max\Bigl\{\sup_{\bh\in\mathcal{H}^d_{\mathrm{dyad}}}\sup_{\bt\in\mathcal{B}}\omega_{\bh}(Z_{\bt,\bh}-\omega_{\bh})\,,\,\sup_{\bh\in\mathcal{H}_{\mathrm{dyad}}}\sup_{\bt\in\mathcal{R}}\omega_{\bh}(Z_{\bt,\bh}-\omega_{\bh})\Bigr\}.
	\end{align*}
	The second term converges to $-\infty$. To see this, consider
	\begin{align*}
	\mathbb{P}\Bigl(\sup_{\bh\in\mathcal{H}^d_{\mathrm{dyad}}}\sup_{\bt\in\mathcal{R}}\omega_{\bh}(Z_{\bt,\bh}-\omega_{\bh})>-\sqrt{\log(n)}\Bigr).
	\end{align*}
	Let, for $J\subset\{1,\ldots,d\},\,0\leq|J|<d$ and $\bp\in\mathcal{P}^d,$ $\mathcal{I}_{\bp,J}=\bigtimes_{j=1}^dI_{p_j,J},$ where $\mathrm{Leb}^1(I_{p_j,J})=1$ if $j\in J^C$ and $\mathrm{Leb}(I_{p_j,J})=O\Big(\frac{1}{h_{p_j}}\Big)$ if $j\in J.$ Then
	\begin{align*}
	&\mathbb{P}\Bigl(\sup_{\bh\in\mathcal{H}^d_{\mathrm{dyad}}}\sup_{\bt\in\mathcal{R}}\omega_{\bh}(Z_{\bt,\bh}-\omega_{\bh})>-\sqrt{\log(n)}\Bigr)
	\leq\sum_{\bp\in\mathcal{P}^d}\sum_{\substack{J\subset\{1,\ldots,d\}\\0\leq |J|<d}}\mathbb{P}\Bigl(\sup_{\bt\in\mathcal{I}_{\bp,J}}\omega_{\bh_{\bp}}(Z_{\bt}-\omega_{\bh_{\bp}})>-\sqrt{\log(n)}\Bigr)\\
	&~\leq\sum_{\bp\in\mathcal{P}^d}\sum_{\substack{J\subset\{1,\ldots,d\}\\0\leq |J|<d}}\mathbb{P}\Bigl(\sup_{\bt\in\mathcal{I}_{\bp,J}}Z_{\bt}>-C_1+\omega_{\bh_{\bp}}\Bigr),
	\end{align*}
	for a constant $C_1>0.$ An application of Borel's inequality yields the existence of constants $C_2,\delta_{\mathrm{B}}$ such that
	\begin{align*}
	\sum_{\bp\in\mathcal{P}^d}\sum_{\substack{J\subset\{1,\ldots,d\}\\0\leq |J|<d}}\mathbb{P}\Bigl(\sup_{\bt\in\mathcal{I}_{\bp,J}}Z_{\bt}>-C_1+\omega_{\bh_{\bp}}\Bigr)\leq\exp\bigg(-C\bigg(\sqrt{\log\Big(\frac{1}{\bh_{\bp}^{\be}}\Big)}-\sqrt{\log\Big(\prod_{j\in J}\frac{1}{h_{p_j}}\Big)}\bigg)^2\bigg)\leq n^{-\delta_{\mathrm{B}}}.
	\end{align*}
	This yields
	\begin{align*}
	\mathbb{P}\bigg(\sup_{\bh\in\mathcal{H}^d_{\mathrm{dyad}}}\sup_{\bt\in\mathcal{B}}\omega_{\bh}(Z_{\bt,\bh}-\omega_{\bh})\bigg)\geq\mathbb{P}\big(M_{n,\mathrm{dyad}}\leq\lambda\big)\geq\mathbb{P}\bigg(\sup_{\bh\in\mathcal{H}^d_{\mathrm{dyad}}}\sup_{\bt\in\mathcal{B}}\omega_{\bh}(Z_{\bt,\bh}-\omega_{\bh})\bigg)-o(n^{-\delta_{\text{B}}}).
	\end{align*}
	~\\\textbf{Step II: The dyadic grid is sufficiently dense.}~\\
	We now show
	\begin{align*}
	\Delta_{\gamma,n}=\bigg|\max_{\bh\in\mathcal{H}^d_{\mathrm{dyad}}}\sup_{\bt\in\mathcal{T}}\omega_{\bh}\big(Z_{\bt,\bh}-\omega_{\bh}\big)-\sup_{\bh\in[h_{\mathrm{min}},h_{\mathrm{max}}]^d}\sup_{\bt\in\mathcal{T}}\omega_{\bh}\big(Z_{\bt,\bh}-\omega_{\bh}\big)\bigg|=o_{\mathbb{P}}(1).
	\end{align*}
	Let $\varepsilon>0$. 
	\begin{align*}
	\mathbb{P}\bigl(\Delta_{n,\gamma}>\varepsilon\bigr)&\leq\mathbb{P}\Bigl(\max_{\mathbf{p}\in\mathcal{P}}\Bigl|\sup_{\bt\in\mathcal{T}}\omega_{\bh_{\mathbf{p}}}\big(Z_{\bt,\bh_{\mathbf{p}}}-\omega_{\bh_{\mathbf{p}}}\big)-\max_{\bh\in[\bh_{\mathbf{p}},\bh_{\mathbf{p+1}}]}\sup_{\bt\in\mathcal{T}}\omega_{\bh}\big(Z_{\bt,\bh}-\omega_{\bh}\big)\Bigr|>\varepsilon \Bigr)\\
	&\leq\mathbb{P}\Bigl(\max_{\mathbf{p}\in\mathcal{P}}\Bigl|\omega_{\bh_{\mathbf{p}}}\Bigl(\sup_{\bt\in\mathcal{T}}\bigl|Z_{\bt,\bh_{\mathbf{p}}}-Z_{\bt,\bh}\bigr|+\max_{\bh\in[\bh_{\mathbf{p}},\bh_{\mathbf{p+1}}]} \bigl|\omega_{\bh}-\omega_{\bh_{\mathbf{p}}}\bigr|\Big)\Bigr|>\varepsilon \Bigr).
	\end{align*}
	\textbf{Step II.1: Fineness of the dyadic grid.}~\\
	Let $h\in[h_{\min},h_{\max}].$ Set $p=\lfloor\log(\grid/h)\rfloor$ and assign the element $h_{\text{dyad}}$ of the dyadic grid to $h$:
	\begin{align}\label{eq:hdyad}
	h_{\text{dyad}}=\mathrm{argmin}\big\{|g-h|\,|\,g\in\{2^{-p},\ldots,2^{-p_{\min}}\}\big\}.
	\end{align}
	Obviously $2^p|h_{\text{dyad}}-h|\leq1/2$, hence
	\begin{align*}
	2^ph_{\text{dyad}}&\in\bigg[2^ph-\frac{1}{2},2^ph+\frac{1}{2}\bigg]\cap\mathbb{N}
	\subset\bigg[\grid-\frac{1}{2},2\grid+\frac{1}{2}\bigg]\cap\mathbb{N}\\
	&\subset\{1,\ldots\lceil2\grid\rceil\}.
	\end{align*}
	Since $2^{-p+1}\grid\geq h\geq2^{-p}\grid$ for sufficiently large $n$, we find
	\begin{align*}
	h_{\text{dyad}}h&\geq h^2-2^{-p-1}h\geq2^{-2p}(\grid)^2-2^{-p-1}h\\
	&\geq2^{-2p}\big((\grid)^2-\grid/2\big).
	\end{align*}
	This yields
	\begin{align}\label{eq:gridwidth}
	\frac{|h-h_{\text{dyad}}|}{\sqrt{hh_{\text{dyad}}}}&\leq\frac{2^{-p-1}}{2^{-p}\sqrt{(\grid)^2-\grid/2}}\notag\\
	&=\frac{1}{2\sqrt{(\grid)^2-\grid/2}}.
	\end{align}
	Let $\bh\in[h_{\min},h_{\max}]^d$ and define $\bh_{\mathrm{dyad}}$ component-wise via \eqref{eq:hdyad}. Then
	\begin{align*}
	|\omega_{\bh}-\omega_{\bh_{\mathrm{dyad}}}|&=\Bigg|\sqrt{2\log\big(\tfrac{K}{\bh^ {\be}}\big)}+\tfrac{\log(\sqrt{2\log(\frac{K}{\bh^ {\be}})})}{\sqrt{2\log(\frac{K}{\bh^ {\be}})}}-\sqrt{2\log\big(\tfrac{K}{\bh_{\mathrm{dyad}}^ {\be}}\big)}-\tfrac{\log(\sqrt{2\log(\frac{K}{\bh_{\mathrm{dyad}}^ {\be}})})}{\sqrt{2\log(\frac{K}{\bh_{\mathrm{dyad}}^ {\be}})}}\Bigg|\\
	&=O\bigg(\frac{1}{\sqrt{\log(n)}}\bigg)\biggl(\biggl|\log\big(\tfrac{K}{\bh^ {\be}}\big)-\log\big(\tfrac{K}{\bh_{\mathrm{dyad}}^ {\be}}\big)\biggr| +o(1)\biggr)\\
	&=O\bigg(\frac{1}{\sqrt{\log(n)}}\bigg)\biggl(\sum_{j=1}^d\frac{|h_{j,\mathrm{dyad}}-h_j|}{\sqrt{h_{j}h_{j,\mathrm{dyad}}}}+o(1)\biggr)=o\bigg(\frac{1}{\sqrt{\log(n)}}\bigg), 
	\end{align*}
	where the last estimate follows from \eqref{eq:gridwidth}.\\
	
	\textbf{Step II.2: Estimation of the covering numbers.}~\\
	We now show that  there exists a constant $C_{\mathrm{cov}}$, depending only on the dimension $d$ and the function $\Xi$ via the constants $L$ and $\gamma$ from condition \eqref{S1} such that for $\varepsilon\in(0,d),$ 
	\begin{align}\label{eq:covering}
	\mathcal{N}\big(\mathcal{T}\times\mathcal{H},\rho,\varepsilon\big)\leq C_{\mathrm{cov}} \biggl(\frac{1}{\varepsilon}\biggr)^{\frac{2d}{\gamma}}\biggl(\frac{1}{h_{\min}}-\frac{1}{h_{\max}}\biggr)^d,
	\end{align}
	where 
	\begin{align}\label{eq:rho}
	\rho^2\big((\bt,\bh),(\bs,\bl)\big)=\mathbb{E}|Z_{\bt,\bh}-Z_{\bs,\bl}|^2
	\end{align}
	and $\mathcal{N}\big(\mathcal{T}\times\mathcal{H},\rho,\varepsilon\big)$ denotes the covering numbers of $\mathcal{T}\times\mathcal{H}$ with respect to $\rho.$ To this end, we first show that 
	\begin{align}\label{eq:dist}
	\rho^2\big((\bt,\bh),(\bs,\bl)\big)\leq2L\sum_{j=1}^d\biggl|\frac{t_j-s_j}{h_j}\biggr|^{2\gamma}+4\cdot2^{d-1}\sum_{j=1}^d\bigg|\frac{l_j-h_j}{\sqrt{h_jl_j}}\bigg|^2+4\cdot2^{d-1}L\sum_{j=1}^d\bigg|\frac{l_j-h_j}{\sqrt{h_jl_j}}\bigg|^{2\gamma},
	\end{align}
	where $L$ is the constant from Assumption \ref{S1}. In a second step we construct an $\varepsilon$-covering with respect to $\rho$ which satisfies inequality \eqref{eq:covering}.
	\begin{align*}
	&\rho^2\big((\bt,\bh),(\bs,\bl)\big)=\int\biggl|\frac{1}{\sqrt{\bh^{\be}}}\Xi\biggl(\frac{\bz-\bt}{\bh}\biggr)-\frac{1}{\sqrt{\bl^{\be}}}\Xi\biggl(\frac{\bz-\bs}{\bl}\biggr)\biggr|^2\,\mathrm{d}\bz\\
	&~~\leq2\biggl(\int\biggl|\frac{1}{\sqrt{\bh^{\be}}}\Xi\biggl(\frac{\bz-\bt}{\bh}\biggr)-\frac{1}{\sqrt{\bh^{\be}}}\Xi\biggl(\frac{\bz-\bs}{\bh}\biggr)\biggr|^2\,\mathrm{d}\bz
	+\int\biggl|\frac{1}{\sqrt{\bh^{\be}}}\Xi\biggl(\frac{\bz-\bs}{\bh}\biggr)-\frac{1}{\sqrt{\bl^{\be}}}\Xi\biggl(\frac{\bz-\bs}{\bl}\biggr)\biggr|^2\,\mathrm{d}\bz\biggr)\\
	&~~=:2\rho_1^2+2\rho_2^2.
	\end{align*}
	Assumption \eqref{S1} immediately gives
	\begin{align*}
	\rho_1^2\leq L\biggl\|\frac{\bt-\bs}{\bh}\biggr\|_{2}^{2\gamma}.
	\end{align*}
	Let $\pmb{\iota}^{j}:=(l_1,\ldots,l_j,h_{j+1},\ldots,h_d),\; j=2,\ldots,d.$ Without loss of generality $\bh\leq\bl$ (else consider $\pmb{\iota}^{j}:=(l_1,\ldots,l_j,h_{j+1}\wedge l_{j+1},\ldots,h_d\wedge l_d)$). we find
	\begin{align*}
	\rho_2^2\leq2^{d-1}\sum_{j=1}^d\int\biggl|\frac{1}{\sqrt{\bh^{\be}}}\Xi\biggl(\frac{\bz-\bs}{\bh}\biggr)-\frac{1}{\sqrt{(\pmb{\iota}^{j})^{\be}}}\Xi\biggl(\frac{\bz-\bs}{\pmb{\iota}^{j}}\biggr)\biggr|^2\,\mathrm{d}\bz=:2^{d-1}\bigl(\rho_{2,1}^2+\ldots+\rho_{2,d-1}^2\bigr).
	\end{align*}
	We now estimate $\rho_{2,1}^2.$
	\begin{align*}
	\rho_{2,1}^2&=\frac{1}{h_2\cdots h_d}\int\biggl|\frac{1}{\sqrt{h_1}}\Xi\biggl(\frac{\bz-\bs}{\bh}\biggr)-\frac{1}{\sqrt{l_1}}\Xi\biggl(\frac{\bz-\bs}{\pmb{\iota}^1}\biggr)\biggr|^2\,\mathrm{d}\bz\\
	&=\int\biggl|\Xi\biggl(\bz-\frac{\bs}{\bh}\biggr)-\frac{\sqrt{h_1}}{\sqrt{l_1}}\Xi\biggl(\frac{z_1h_1-s_1}{l_1},z_2-\frac{s_2}{h_2},\ldots,z_d-\frac{s_d}{h_d}\biggr)\biggr|^2\,\mathrm{d}\bz\\
	&\leq 2\biggl(1-\sqrt{\frac{h_1}{l_1}}\biggr)^2\int\biggl|\Xi\biggl(\bz-\frac{\bs}{\bh}\biggr)\biggr|^2\,\mathrm{d}\bz
	+2\frac{h_1}{l_1}\int\biggl|\Xi\biggl(\bz-\frac{\bs}{\bh}\biggr)-\Xi\biggl(\frac{z_1h_1-s_1}{l_1},z_2-\frac{s_2}{h_2},\ldots,z_d-\frac{s_d}{h_d}\biggr)\biggr|^2\,\mathrm{d}\bz\\
	&=:2\biggl(\frac{\sqrt{l_1}-\sqrt{h_1}}{\sqrt{l_1}}\biggr)^2+2\widetilde{\rho}_{2,1}^2\leq 2\biggl|\frac{l_1-h_1}{\sqrt{l_1h_1}}\biggr|^2+2\widetilde{\rho}_{2,1}^2
	\end{align*}
	We now estimate the term $\widetilde{\rho}_{2,1}^2.$ 
	\begin{align*}
	\widetilde{\rho}_{2,1}^2\leq\frac{h_1}{l_1}\sup_{\mathbf{u}\in[-|h_1-l_1|+s_1,|h_1-l_1|+s_1]\times\{s_2\}\times\cdots\times\{s_d\}}\int\biggl|\Xi\bigg(\bz-\frac{\bs}{\bh}\bigg)-\Xi\bigg(\bz-\frac{\mathbf{u}}{\bl}\bigg)\biggr|^2\,\mathrm{d}\bz.
	\end{align*}
	Now we can apply \eqref{S1} again and find
	\begin{align*}
	\widetilde{\rho}_{2,1}^2&\leq\frac{h_1}{l_1}\sup_{\mathbf{u}\in[-|h_1-l_1|+s_1,|h_1-l_1|+s_1]\times\{s_2\}\times\cdots\times\{s_d\}}\biggl\|\frac{\bs}{\bh}-\frac{\mathbf{u}}{\bl}\biggr\|_{2}^{2\gamma}\leq\frac{h_1}{l_1}\biggl|\frac{h_1-l_1}{h_1}\biggr|_{l_2}^{2\gamma}\\
	&=\frac{h_1}{l_1}\frac{l_1^{\gamma}}{h_1^{\gamma}}\biggl|\frac{h_1-l_1}{\sqrt{h_1l_1}}\biggr|^{2\gamma}\leq\biggl|\frac{h_1-l_1}{\sqrt{h_1l_1}}\biggr|^{2\gamma},
	\end{align*}
	where the last estimate follows since $\gamma\leq1$ and $\bh\leq\bl.$
	The terms $\rho^2_{2,2},\ldots,\rho^2_{2,d-1}$ can be estimated analogously. In total, we now obtain the estimate \eqref{eq:dist}.
	Based on this, we  now construct an $\varepsilon$-covering with respect to $\rho$ which satisfies inequality \eqref{eq:covering}.\\
	
	To this end, define the Grid
	\begin{align*}
	\mathcal{H}_{a}:=\bigg\{a^k\,\bigg|\,k\in\bigg\{\left\lceil\frac{\log(h_{\min})}{\log(a)}\right\rceil,\ldots,\left\lfloor\frac{\log(h_{\max})}{\log(a)}\right\rfloor\bigg\}\bigg\},
	\end{align*}
	where
	\begin{align*}
	a:=1+\biggl(\frac{(F(\varepsilon))^{\frac{2}{\gamma}}}{2}-\sqrt{(F(\varepsilon))^{\frac{2}{\gamma}}+\frac{1}{4}(F(\varepsilon))^{\frac{4}{\gamma}}}\biggr),
	\end{align*}
	and $F(\varepsilon)=\varepsilon/(2d2^{\frac{d-1}{2}}(\sqrt{L}+1)).$
	
	The grid $\mathcal{H}_a$ is constructed such that 
	\begin{align*}
	\biggl|\frac{a^k-a^{k+1}}{\sqrt{a^ka^{k+1}}}\biggr|^{\gamma}<F(\varepsilon).
	\end{align*}
	
	Notice that there exists $a_0>0$ such that $a\in(a_0,1)$,  if $\varepsilon\in(0,d).$
	Let 
	\begin{align*}
	\mathcal{G}:=\bigg\{(\bt,\bh)\,|\,\bh\in\mathcal{H}_a^d,\;\bt\in\bigtimes_{j=1}^d\Big\{k\cdot G(\varepsilon)^{1/\gamma}h_j\,|\,1\leq k\leq G(\varepsilon)^{-1/\gamma}h_j^{-1}\Big\}\bigg\},
	\end{align*}
	where $G(\varepsilon):=\frac{\varepsilon a_0}{d\sqrt{2L}}.$ 
	We have
	\begin{align*}
	\big|\mathcal{G}\big|^{1/d}\leq\sum_{j=\left\lceil\frac{\log(h_{\min})}{\log(a)}\right\rceil}^{\left\lfloor\frac{\log(h_{\max})}{\log(a)}\right\rfloor}G(\varepsilon)^{-1/\gamma}a^{-j}
	& =G(\varepsilon)^{-1/\gamma}a\frac{\bigl(\frac{1}{a}\big)^{1+\left\lceil\frac{\log(h_{\min})}{\log(a)}\right\rceil}-\bigl(\tfrac{1}{a}\big)^{1+\left\lfloor\frac{\log(h_{\max})}{\log(a)}\right\rfloor}}{a-1}\\
	&\leq G(\varepsilon)^{-1/\gamma}\biggl(\frac{1}{a^2}\frac{1}{h_{\min}}-\frac{1}{h_{\max}}\biggr)\leq(2\varepsilon)^{-1/\gamma}\frac{1}{a_0^2}\biggl(\frac{1}{h_{\min}}-\frac{1}{h_{\max}}\biggr).
	\end{align*}
	
	Fix $(\bt_0,\bh_0)\in\mathcal{T}\times\mathcal{H}$ and set
	\begin{align*}
	(\bt_{\mathrm{grid}},\bh_{\mathrm{grid}}):=\mathrm{argmin}\{\rho((\bt,\bh);(\bt_0,\bh_0))\,|\,(\bt,\bh)\in\mathcal{G}\}.
	\end{align*}
	Hence
	\begin{align*}
	\biggl|\frac{h_{0,j}-h_{\mathrm{grid},j}}{\sqrt{h_{0,j}h_{\mathrm{grid},j}}}\biggr|^{\gamma}<\frac{1}{2}F(\varepsilon)\quad\text{and thus}\quad\biggl|\frac{h_{0,j}-h_{\mathrm{grid},j}}{\sqrt{h_{0,j}h_{\mathrm{grid},j}}}\biggr|<\biggl|\frac{h_{0,j}-h_{\mathrm{grid},j}}{\sqrt{h_{0,j}h_{\mathrm{grid},j}}}\biggr|^{\gamma}.
	\end{align*}
	Also $|t_{0,j}-t_{\mathrm{grid}}|^{\gamma}<G(\varepsilon) h_{\mathrm{grid},j},$ which implies
	\begin{align*}
	\biggl|\frac{t_{0,j}-t_{\mathrm{grid},j}}{h_{\mathrm{grid},j}\wedge h_{0,j}}\biggr|^{\gamma}<\frac{1}{2}G(\varepsilon) \biggl(\frac{h_{\mathrm{grid,j}}}{h_{\mathrm{grid},j}\wedge h_{0,j}}\biggr)^{\gamma}\leq \frac{1}{2}G(\varepsilon) \biggl(\frac{1}{a_0}\biggr)^{\gamma}\leq\varepsilon\frac{1}{2d\sqrt{2 L}}.
	\end{align*}
	Then
	\begin{align*}
	\rho((\bt_{\mathrm{grid}},\bh_{\mathrm{grid}});(\bt_0,\bh_0))\leq \varepsilon,
	\end{align*}
	hence the grid $\mathcal{G}$ defines an $\varepsilon$-covering with the desired properties.\\\\
	\textbf{Step II.3: Proof of $\max_{\mathbf{p}\in\mathcal{P}^d}\omega_{\bh_{\mathbf{p}}}\bigl|\sup_{\bt\in\mathcal{T}}Z_{\bt,\bh_{\mathbf{p}}}-\sup_{\bh\in[\bh_{\mathbf{p}},\bh_{\mathbf{p+1}}]}\sup_{\bt\in\mathcal{T}}Z_{\bt,\bh}\bigr|=o(1)$.}~\\
	First, we estimate
	\begin{align*}
	&\mathbb{P}\biggl(\max_{\mathbf{p}\in\mathcal{P}^d}\omega_{\bh_{\mathbf{p}}}\bigl|\sup_{\bt\in\mathcal{T}}Z_{\bt,\bh_{\mathbf{p}}}-\sup_{\bh\in[\bh_{\mathbf{p}},\bh_{\mathbf{p+1}}]}\sup_{\bt\in\mathcal{T}}Z_{\bt,\bh}\bigr|>\varepsilon\biggr)\\
	&\leq\mathbb{P}\biggl(\max_{\mathbf{p}\in\mathcal{P}^d}\sup_{\bh\in[\bh_{\mathbf{p}},\bh_{\mathbf{p+1}}]}\sup_{\bt\in\mathcal{T}}\;\omega_{\bh_{\mathbf{p}}}\bigl|Z_{\bt,\bh_{\mathbf{p}}}-Z_{\bt,\bh}\bigr|>\varepsilon\biggr).
	\end{align*}
Now, we use \eqref{eq:gridwidth}, which yields $\rho\big((\bt,\bh);(\bt,\bh_{\bp})\big)\leq C_{\Delta_n,\gamma}/\grid^{\gamma}=:\widetilde C_{n,\gamma}.$ Thus,
\begin{align*}
&\mathbb{P}\biggl(\max_{\mathbf{p}\in\mathcal{P}^d}\omega_{\bh_{\mathbf{p}}}\bigl|\sup_{\bt\in\mathcal{T}}Z_{\bt,\bh_{\mathbf{p}}}-\sup_{\bh\in[\bh_{\mathbf{p}},\bh_{\mathbf{p+1}}]}\sup_{\bt\in\mathcal{T}}Z_{\bt,\bh}\bigr|>\varepsilon\biggr)
\leq\varepsilon^{-1}\omega_{\bh_{\min}}\mathbb{E}\biggl[\sup_{\rho((\bt,\bh);(\bs,\bl))\leq\widetilde C_{n,\gamma}}\bigl|Z_{\bs,\bl}-Z_{\bt,\bh}\bigr|\biggr],
\end{align*}
by Markov's inequality. An application of Dudley's Theorem
yields
\begin{align}\label{eq:Dudley}
\mathbb{E}\Biggl[\sup_{\rho((\bt,\bh);(\bs,\bl))\leq\widetilde C_{n,\gamma}}\bigl|Z_{\bs,\bl}-Z_{\bt,\bh}\bigr|\Biggr]\leq C\int_{0}^{\widetilde C_{n,\gamma}}\sqrt{\log\big(\mathcal{N}\big(\mathcal{T}\times\mathcal{H},\rho,\eta\big)\big)}\,\mathrm{d}\eta.
\end{align}
By inequality \eqref{eq:covering},
\begin{align*}
\mathbb{P}\biggl(\max_{\mathbf{p}\in\mathcal{P}^d}\omega_{\bh_{\mathbf{p}}}\bigl|\sup_{\bt\in\mathcal{T}}Z_{\bt,\bh_{\mathbf{p}}}-\sup_{\bh\in[\bh_{\mathbf{p}},\bh_{\mathbf{p+1}}]}\sup_{\bt\in\mathcal{T}}Z_{\bt,\bh}\bigr|>\varepsilon\biggr)=o(1)\quad\text{as}\quad n\to\infty.
\end{align*}
Hence, the supremum over the dyadic grid and the supremum over the full range $[h_{\min},h_{\max}]^d$ have the same limit.
\end{proof}
\begin{proof}[Proof of Theorem \ref{GWSMain}]
The first claim is an asymptotic statement as $\lambda\to\infty$. Let $\lambda>0$. As in the proof of Theorem \ref{Thm:Gumbel} write
	\begin{align*}
	&\mathbb{P}\biggl(\sup_{\bh\in\mathcal{H}}\sup_{\bt\in\mathcal{T}_{\bh}}\omega_{\bh}\bigl(Z_{\bt,\bh}-\omega_{\bh}\bigr)>2\lambda\biggr)\leq\mathbb{P}\biggl(\sup_{\bh\in[\bh_{\min},\bh_{\max}]}\sup_{\bt\in[\bh,\be]}\omega_{\bh}\bigl(Z_{\bt,\bh}-\omega_{\bh}\bigr)>2\lambda\biggr)\\
	&\leq\mathbb{P}\biggl(\sup_{\bh\in\mathcal{H}_{\mathrm{dyad}}^d}\sup_{\bt\in\mathcal{T}_{\bh}}\omega_{\bh}\bigl(Z_{\bt,\bh}-\omega_{\bh}\bigr)>\lambda\biggr)+\frac{C}{\lambda},
	\end{align*}
where the last estimate follows as in \eqref{eq:Dudley}.	
We showed in Step I.3 of the proof of Theorem \ref{Thm:Gumbel} that $\sup_{\bh\in\mathcal{H}^d_{\mathrm{dyad}}}\sup_{\bt\in\mathcal{R}}\omega_{\bh}(Z_{\bt,\bh}-\omega_{\bh})\to-\infty$ in probability. Therefore,
	\begin{align*}
	\mathbb{P}\biggl(\sup_{\bh\in\mathcal{H}_{\mathrm{dyad}}^d}\sup_{\bt\in\mathcal{T}_{\bh}}\omega_{\bh}\bigl(Z_{\bt,\bh}-\omega_{\bh}\bigr)>\lambda\biggr)=\mathbb{P}\bigl(M_{\mathcal{B}}>\lambda\bigr)+o(1),
	\end{align*}
as $\lambda\to\infty.$
	 As in Step I.1 in the proof of Theorem \ref{Thm:Gumbel}, we write
	\begin{align}\label{eq:MB}
	\mathbb{P}\bigl(M_{\mathcal{B}}\leq\lambda\bigr)
	=\prod_{\bp\in\mathcal{P}^d}\mathbb{P}\biggr(\sup_{\bt\in B_{\bp}}\int\Xi(\bt-\bz)\,\mathrm{d}W_{\bz} \leq \Lambda_{\min,\bp}\biggr),
	\end{align}
	where
	$
	\Lambda_{\min,\bp}:=\min_{\bp\leq\bq}\bigl(\frac{\lambda}{\omega_{\bh_{\bq}}}+\omega_{\bh_{\bq}}\bigr).
	$
	By Theorem 4.1.2 in \citet{adltay2007} we deduce that there exists a constant $D_{\gamma,1}$, depending only on the degree of average H\"older smoothness, $\gamma$, (see \eqref{S1} in Assumption \ref{Ass:Dict}) such that
	\begin{align*}
	\mathbb{P}\bigl(M_{\mathcal{B}}\leq\lambda\bigr)
	&\geq\prod_{\bp\in\mathcal{P}^d}\biggr(1-D_{\gamma,1}\mathrm{Leb}(B_{\bp})\Lambda_{\min,\bp}^{d/\gamma-1}\exp\Bigl(-\frac{1}{2}\Lambda_{\min, \bp}^2\Bigr)\biggr)\\
	&\geq\prod_{\bp\in\mathcal{P}^d}\biggr(1-D_{\gamma,1}\mathrm{Leb}(B_{\bp})\Lambda_{\bp}^{d/\gamma-1}\exp\Bigl(-\frac{1}{2}\Lambda_{\min, \bp}^2\Bigr)\biggr),
	\end{align*}
	since $\Lambda_{\min,\bp}\leq\Lambda_{\bp}.$ Also, $\Lambda_{\min,\bp}=\biggl(\frac{\lambda}{\omega_{\bh_{\bp_0}}}+\omega_{\bh_{\bp_0}}\biggr)$ for some $\bp\leq\bp_0=\bp_0(\lambda).$ Hence, $\Lambda_{\min,\bp}^2=(\lambda/\omega_{\bh_{\bp_0}})^2+2\lambda+\omega_{\bh_{\bp_0}}^2\geq  2\lambda+\omega_{\bh_{\bp}}^2,$ where we used that $\bh_\bp>\bh_{\bp+\be}\;\forall\; \bp\in\mathcal{P}^d.$ We obtain
	\begin{align*}
	\mathbb{P}\bigl(M_{\mathcal{B}}\leq\lambda\bigr)
	\geq \prod_{\bp\in\mathcal{P}^d}\biggr(1-C\mathrm{Leb}(B_{\bp})\bh_{\bp}^{\be}e^{-\lambda}\biggl(\Bigl(\tfrac{\lambda^2}{2\log(\frac{K}{\bh_{\bp}^{\be}})}\Bigr)^{\frac{d}{2\gamma}-\frac{1}{2}}+\omega_{\bh_{\bp}}^{\frac{d}{\gamma}-1}\biggr)\biggl(\tfrac{1}{\log(\frac{K}{\bh_{\bp}^{\be}})}\biggr)^{\frac{C_d}{2}}\biggr).
	\end{align*}
	We find, for sufficiently large $n,$
	\begin{align*}
	&\mathbb{P}\bigl(M_{\mathcal{B}}\leq\lambda\bigr)
	\geq \prod_{\bp\in\mathcal{P}^d}\biggr(1-Ce^{-\lambda}\biggl(\tfrac{1}{2\log(K/(\bh_{\bp}^{\be}))}\biggr)^{\frac{C_d}{2}}\biggl(\Bigl(\tfrac{\lambda^2}{2\log(\frac{K}{\bh_\bp^{\be}})}\Bigr)^{\frac{d}{2\gamma}-\frac{1}{2}}+\omega_{\bh_\bp}^{\frac{d}{\gamma}-1}\biggr)\\
	&~\geq\exp\biggl(\sum_{\bp\in\mathcal{P}^d}\log\biggl(1-Ce^{-\lambda}\biggl[\biggl(\tfrac{1}{2\log(K/(\bh_\bp^{\be}))}\biggr)^{\frac{C_d-d/\gamma+1}{2}}+\biggl(\tfrac{|\lambda|^{\frac{d}{\gamma}-1}}{(2\log(K/\bh_\bp^{\be}))^{\frac{C_d+d/\gamma-1}{2}}}\biggr)\biggr]\biggr).
	\end{align*}
	Recall that $C_d=2d+d/\gamma-1$ and hence $\frac{C_d-d/\gamma+1}{2}=d$ as well as $\frac{C_d+d/\gamma-1}{2}\geq2d-1\geq d.$
	\begin{align*}
	\mathbb{P}\bigl(M_{\mathcal{B}}\leq\lambda\bigr)
	&\geq\exp\Biggl(\sum_{\bp\in\mathcal{P}^d}\log\biggl(1-Ce^{-\lambda/2}\frac{2e^{-\lambda/2}|\lambda|^{\frac{d}{\gamma}-1}}{(2\log(K/(\bh_\bp^{\be})))^d}\biggr)\Biggr)\\&\geq\exp\Biggl(\sum_{\bp\in\mathcal{P}^d}\log\biggl(1-C\frac{e^{-\lambda/2}}{(2\log(K/(\bh_\bp^{\be})))^d}\biggr)\Biggr).
	\end{align*}
	Since $\log(1-x)\leq-x$ for $x<1$, we find, for sufficiently large $n$,
	\begin{align*}
	\mathbb{P}\bigl(M_{\mathcal{B}}\leq\lambda\bigr)
	\geq\exp\Biggl(-Ce^{-\lambda/2}\sum_{\bp\in\mathcal{P}^d}\frac{1}{(2\log(K/(\bh_\bp^{\be})))^d}\Biggr)\geq\exp\Biggl(-Ce^{-\lambda/2}\biggl(\frac{p_{\max}}{p_{\min}}\biggr)^d\Biggr).
	\end{align*}
	Since $p_{\max}/p_{\min}$ is bounded we find $F(\lambda)=e^{-Ce^{-\lambda/2}}$ for some $C>0$ which is independent of $n.$
	This concludes the proof of the first claim of this theorem. In order to proof the lower Gumbel-bound we proceed similar.
For $\lambda\in\R$, we have that
	\begin{align*}
	&\mathbb{P}\biggl(\sup_{\bh\in\mathcal{H}}\sup_{\bt\in\mathcal{T}_{\bh}}\omega_{\bh}\bigl(Z_{\bt,\bh}-\omega_{\bh}\bigr)\leq\lambda\biggr)\geq\mathbb{P}\biggl(\sup_{\bh\in[\bh_{\min},\bh_{\max}]}\sup_{\bt\in[\bh,\be]}\omega_{\bh}\bigl(Z_{\bt,\bh}-\omega_{\bh}\bigr)\leq\lambda\biggr)\\
	&=\mathbb{P}\biggl(\sup_{\bh\in\mathcal{H}_{\mathrm{dyad}}^d}\sup_{\bt\in[\bh,\be]}\omega_{\bh}\bigl(Z_{\bt,\bh}-\omega_{\bh}\bigr)\leq\lambda\biggr)+o(1)=\mathbb{P}\bigl(M_{\mathcal{B}}\leq\lambda\bigr)+o(1),
	\end{align*}
	where the estimates follows from the proof of Theorem \ref{Thm:Gumbel}. We make use of \eqref{eq:MB} again but
now, in contrast to before, we consider fixed $\lambda\in\R$. Hence, for sufficiently large $n$, $\Lambda_{\max,\bp}=\Lambda_\bp.$
	Again, by Theorem 4.1.2 in \citet{adltay2007}, we deduce that there exists a constant $D_{\gamma,1}$, depending only on the degree of average H\"older smoothness, $\gamma$, (see \eqref{S1} in Assumption \ref{Ass:Dict}) such that
	\begin{align*}
	\mathbb{P}\bigl(M_{\mathcal{B}}\leq\lambda\bigr)
	\geq\prod_{\bp\in\mathcal{P}^d}\biggr(1-D_{\gamma,1}\mathrm{Leb}(B_{\bp})\Lambda_{\bp}^{d/\gamma-1}\exp\Bigl(-\frac{1}{2}\Lambda_{\bp}^2\Bigr)\biggr).
	\end{align*} 
	From here the second claim now follows as in the proof of Theorem \ref{Thm:Gumbel}. 
	
\end{proof}

\subsection{Localization}
\begin{Lemma}\label{Le:Loc}
	If Assumption \ref{Ass:Dict} (c) is satisfied, the following holds true for a normed and uniformly bounded test-function $\Xi\in L^2[0,1]^d.$
	\begin{align*}
	\max_{i\in\{1,\ldots,N\}}\biggl|\frac{1}{\sqrt{\bh_{i}^{\be}}}\int(\sigma(\bz)-\sigma(\bt_i))\Xi\Big(\frac{\bt_i-\bz}{\bh_i}\Big)\,\mathrm{d}W_{\bz}\biggr|=O_{\mathbb{P}}\bigl(\sqrt{\log(n)h_{\max}}\bigr)=o_\mathbb{P}\Bigl(\log(n)^{-1/2}\Bigr).
	\end{align*}    
\end{Lemma}
\begin{proof}[Proof of Lemma \ref{Le:Loc}]
	Define
	\begin{align*}
	Y_{\bt_i,\bh_i}^{(1)}:=\frac{1}{\sqrt{h_{i,1}\cdot\ldots\cdot h_{i,d}}}\int(\sigma(\bz)-\sigma(\bt_i))\Xi\Big(\frac{\bt_i-\bz}{\bh_i}\Big)\,\mathrm{d}W_{\bz},
	\end{align*}
	\begin{align*}
	Y_{\bt_i,\bh_i}^{(2)}:=\frac{1}{\sqrt{h_{i,1}\cdot\ldots\cdot h_{i,d}}}\int\mathrm{grad}(\sigma)( \bt_i)(\bz-\bt_i)\Xi\Big(\frac{\bt_i-\bz}{\bh_i}\Big)\,\mathrm{d}W_{\bz},
	\end{align*}
	and, for $l=1,2,$ $(i,j)\in\{1,\ldots,N\}^2,$ let
	$
	\gamma_{i,j}^{(l)}:=\mathbb{E}\bigl[(Y_{\bt_i,\bh_i}^{(l)}-Y_{\bt_j,\bh_j}^{(l)})^2\bigr].
	$
	
	We have that
	\begin{align}\label{abc}
	\sup_{(i,j)\in\{1,\ldots,N\}^2}|\gamma_{i,j}^{(1)}-\gamma_{i,j}^{(2)}|\leq C(\|\Xi\|_2^2+\|\Xi\|_{\infty}^2)h_{\max}.
	\end{align}
	Using Assumption \eqref{growth} and \eqref{abc}, an application of Theorem 2.2.5 in \citet{adltay2007} yields
	\begin{align}\label{def}
	\Bigg|\mathbb{E}\biggl[\max_{i\in\{1,\ldots,N\}}Y_{\bt_i,\bh_i}^{(1)}\biggr]-\mathbb{E}\biggl[\max_{i\in\{1,\ldots,N\}}Y_{\bt_i,\bh_i}^{(2)}\biggr]\Bigg|\leq C\sqrt{h_{\max}\log(N)}=O\bigl(\sqrt{h_{\max}\log(n)}\bigr).
	\end{align}
	Define for $k\in\{1,\ldots,d\}$ the quantities
	\begin{align*}
	Y_{\bt_i,\bh_i}^{(2,k)}:&=\frac{(\mathrm{grad}(\sigma))_k( \bt_i)}{\sqrt{h_{i,1}\cdot\ldots\cdot h_{i,d}}}\int(z_k-t_{i,k})\Xi\Big(\frac{\bt_i-\bz}{\bh_i}\Big)\,\mathrm{d}W_{\bz}\\
	&=\frac{h_{i,k}(\mathrm{grad}(\sigma))_k( \bt_i)}{\sqrt{h_{i,1}\cdot\ldots\cdot h_{i,d}}}\int\frac{z_k-t_{i,k}}{h_{i,k}}\Xi\Big(\frac{\bt_i-\bz}{\bh_i}\Big)\,\mathrm{d}W_{\bz}=:\frac{h_{i,k}(\mathrm{grad}(\sigma))_k( \bt_i)}{\sqrt{h_{i,1}\cdot\ldots\cdot h_{i,d}}}\int\widetilde\Xi_k\Big(\frac{\bt_i-\bz}{\bh_i}\Big)\,\mathrm{d}W_{\bz},
	\end{align*}
	where $\widetilde{\Xi}_k(\bz)=z_k\Xi(\bz).$
	By Theorem \ref{GWSMain} we obtain that 
	\begin{align*}
	\mathbb{E}\biggl[\max_{1\leq i\leq N}\frac{1}{\sqrt{h_{i,1}\cdot\ldots\cdot h_{i,d}}}\int\widetilde\Xi\Big(\frac{\bt_i-\bz}{\bh_i}\Big)\,\mathrm{d}W_{\bz}\biggr]=O\bigl(\sqrt{\log(n)}\bigr).
	\end{align*}
	Since $h_{i,k}(\mathrm{grad}(\sigma))_k( \bt_i)\leq Ch_{\max}$
	we conclude that $\mathbb{E}\bigl[\max_{i\in\{1,\ldots,N\}}Y_{\bt_i,\bh_i}^{(2)}\bigr]\leq C\sqrt{\log(n)}h_{\max}$ and hence, by \eqref{def}, there exists a positive constant $C_{Y^{(1)}}$ such that $\mathbb{E}\bigl[\max_{i\in\{1,\ldots,N\}}Y_{\bt_i,\bh_i}^{(1)}\bigr]\leq C_{Y^{(1)}}\sqrt{\log(n)}h_{\max}$. It follows by an application of Borell's inequality that, for $\lambda>C_{Y^{(1)}}$,
	\begin{align*}
	\mathbb{P}\biggl(\bigl|\max_{i\in\{1,\ldots,N\}}Y_{\bt_i,\bh_i}^{(1)}\bigr|>\lambda\sqrt{\log(n)h_{\max}}\biggr)\leq\exp\bigl(-(\lambda- C_{Y^{(1)}})^2\log(n)\bigr)=n^{-(\lambda- C_{Y^{(1)}})^2},
	\end{align*}
	
	where we also used that $\max_{i\in\{1,\ldots,N\}}$Var$(Y_{\bt_i,\bh_i}^{(1)})\leq h_{\max}^2\|\Xi\|_2^2=h_{\max}^2$ since $\Xi$ is normed. The assertion of the lemma now follows. 
\end{proof}   	
\subsection{A continuous limit}

\begin{Lemma}\label{Le:SumInt}
	Let $\Xi$ be a test function satisfying Assumption \ref{Ass:Dict} (c) and \ref{Ass:Dict} (d) and let $\{\zeta_{\bk}\,|\,\bk\in\mathbb{N}^ d\}$ be a field of independent, standard normally distributed random variables. Then
	\begin{align*}
	\max_{1\leq i\leq N}\biggl\{(\sqrt{\bh_{i}^{\be}})^{-1}\biggl(n^{-\frac{d}{2}}\sum_{\bj\in I_n^d}\zeta_{\bj}\Xi\Bigl(\frac{\bt_i-\bx_{\bj}}{\bh_i}\Bigr)-\int \Xi\Bigl(\frac{\bt_i-\bx_{\bj}}{\bh_i}\Bigr)\,\mathrm{d}W_{\bz}\biggr)\biggr\}=o_{\mathbb{P}}\biggl(\frac{1}{\sqrt{\log(n)}}\biggr).
	\end{align*} 
\end{Lemma}
\begin{proof}[Proof of Lemma \ref{Le:SumInt}]
	Define the partial sum $S_{\mathbf{l}}$ as follows:
	\begin{align*}
	S_{\bl}=\sum_{\be\leq\mathbf{k}\leq\mathbf{l}}\zeta_{\bk},\quad\text{with}\quad S_{\bl}\equiv0\quad\text{if}\quad\prod_{j=1}^dl_j=0.
	\end{align*}
	Each random variable $\zeta_{\bk}$ can be expressed in terms of increments of the corresponding partial sum function, 
	\begin{align}\label{PS}
	\zeta_{\bk}=\sum_{\balpha\in\{0,1\}^d}(-1)^{|\balpha|}S_{\bk-\balpha}.
	\end{align}
	For each fixed $\bh\in\mathcal{H}^d$, there exists a set $\mathcal{T}_{\bh}$ such that 
	$
	\max_{1\leq i\leq N}F(\bt_i,\bh_i)=\max_{\bh\in\mathcal{H}^d}\max_{\bt\in\mathcal{T}_{\bh}}F(\bt,\bh). $ Let further $\frac{\be}{\bh}\mathcal{T}_{\bh}$ denote the set $\{(\bt/\bh)\,|\,\bt\in\mathcal{T}_{\bh}\}$.	
	For $\bt\in\mathcal{T}_{\bh}$ and $\bz_{\bj}:=(j_1/(nh_1),\ldots,j_d/(nh_d))$
	\begin{align*}
	&\sum_{\be\leq\bj\leq\mathbf{n}-\be}\int_{[\bz_{\bj},\bz_{\bj+\be})}\Xi(\bt-\bz_{\bj})\,\mathrm{d}W_{\bz}=\sum_{\be\leq\bj\leq\mathbf{n}-\be}\Xi\Bigl(\frac{\bt}{\bh}-\bz_{\bj}\Bigr)\sum_{\balpha\in\{0,1\}^d}(-1)^{|\balpha|}W(\bz_{\bj-\balpha})
	\\
	& \stackrel{\mathcal{D}}{=} \Big(\sqrt{\bh_{i}^{\be}}\Big)^{-1}n^{-\frac{d}{2}}\sum_{\be\leq\bj\leq\mathbf{n}-\be}\Xi\Bigl(\frac{\bt}{\bh}-\bz_{\bj}\Bigr)\sum_{\balpha\in\{0,1\}^d}(-1)^{|\balpha|}W(\bj-\balpha)\\
	&\stackrel{\mathcal{D}}{=} \Big(\sqrt{\bh_{i}^{\be}}\Big)^{-1}n^{-\frac{d}{2}}\sum_{\be\leq\bj\leq\mathbf{n}-\be}\Xi\Bigl(\frac{\bt}{\bh}-\bz_{\bj}\Bigr)\sum_{\balpha\in\{0,1\}^d}(-1)^{|\balpha|}S_{\bj-\balpha}
	=
	\Big(\sqrt{\bh_{i}^{\be}}\Big)^{-1}n^{-\frac{d}{2}}\sum_{\be\leq\bj\leq\mathbf{n}-\be}\Xi\Bigl(\frac{\bt-\bx_{\bj}}{\bh}\Bigr)\zeta_{\mathbf{j}},
	\end{align*} 		
	where the last equality follows from \eqref{PS}. 	
	For fixed $\bh\in\mathcal{H}^d$ and parameter $\bt\in\frac{\be}{\bh}\mathcal{T}_{\bh}$ consider the processes
	\begin{align*}
	Y_{\bt}^{(1)}=Y_{\bt}^{(1)}(\bh)=\int\Xi(\bt-\bz)\,\mathrm{d}W_{\bz}-\sum_{\be\leq\bj\leq\mathbf{n}-\be}\int_{[\bz_{\bj},\bz_{\bj+\be})}\Xi(\bt-\bz)\,\mathrm{d}W_{\bz}
	\end{align*}
	and
	\begin{align*}
	Y_{\bt}^{(2)}&=Y_{\bt}^{(2)}(\bh)=\int\Xi(\bt-\bz)\,\mathrm{d}W_{\bz}-\int\Xi\Big(\bt-\Big(\bz+\frac{1}{n\bh}\Big)\Big)\,\mathrm{d}W_{\bz}\\
	&=\int\Xi(\bt-\bz)\,\mathrm{d}W_{\bz}-\int\sum_{\be\leq\bj\leq\mathbf{n}-\be}I_{[\bz_{\bj},\bz_{\bj+\be})}(\bz)\Xi\Big(\bt-\Big(\bz+\frac{1}{n\bh}\Big)\Big)\,\mathrm{d}W_{\bz}.
	\end{align*}
	We first show that
	\begin{align}\label{Y1bounded}
	\mathbb{E}\bigg[\max_{\bt\in\frac{\be}{\bh}\mathcal{T}_{\bh}}Y^{(2)}_{\bt}\bigg]=O\biggl(\frac{\sqrt{\log(n)}}{(nh_{\min})^{\gamma}}\biggr),
	\end{align}
	where $\gamma$ is the degree of average smoothness (see \eqref{S1}).
	By Assumption \eqref{S1}, for $\bt\in\frac{\be}{\bh}\mathcal{T}_{\bh},$ we immediately obtain
	\begin{align*}
	d_2(\bt,\bt'):=\mathbb{E}\|Y_{\bt}^{(2)}-Y_{\bt'}^{(2)}\|_2^2\leq C\|\bt-\bt'\|_2^{2\gamma}.
	\end{align*}
	Hence, for all $\varepsilon>0$ we find
	\begin{align*}
	\mathcal{N}(\varepsilon,\tfrac{\be}{\bh}\mathcal{T}_{\bh},d_2)\leq \frac{C}{h_{i,1}\cdot\ldots\cdot h_{i,d}}\Bigl(\frac{1}{\varepsilon}\Bigr)^{d/\gamma},
	\end{align*}
	where $\mathcal{N}(\varepsilon,\frac{\be}{\bh}\mathcal{T}_{\bh},d_2)$ denotes the covering number of $\frac{\be}{\bh}\mathcal{T}_{\bh}$ with respect to the (pseudo-)distance $d_2$. Furthermore, also by Assumption \eqref{S1}, we find
	$
	\mathrm{Var}[Y_{\bt}^{(2)}]\leq L\Bigl(\frac{1}{nh_{\min}}\Bigr)^{2\gamma}.
	$ For $\lambda>\sqrt{L}\Bigl(\frac{1}{nh_{\min}}\Bigr)^{\gamma}(1+\sqrt{d/\gamma}),$ an application of Theorem 4.1.2 in \citet{adltay2007} yields
	\begin{align}\label{tailbound}
	\mathbb{P}\biggl(\max_{\bt\in\tfrac{\be}{\bh}\mathcal{T}_{\bh}}Y_{\bt}^{(2)}>\lambda\biggr)\leq\frac{C}{h_{i,1}\cdot\ldots\cdot h_{i,d}}\biggl(\frac{\lambda(nh_{\min})^{2\gamma}}{L\sqrt{d/\gamma}}\biggr)^{d/\gamma}\overline{\psi}\biggl(\frac{\lambda(nh_{\min})^{\gamma}}{L}\biggr),
	\end{align}
	where  $\overline{\psi}(x)=\frac{1}{\sqrt{2\pi}}\int_{x}^{\infty}\exp\big(-\frac{1}{2}z^2\big)\,\mathrm{d}z$ denotes the tail function of the standard normal distribution. We further obtain
	\begin{align*}
	\mathbb{E}\bigg[\max_{\bt\in\frac{\be}{\bh}\mathcal{T}_{\bh}}Y^{(2)}_{\bt}\bigg]&\leq \mathbb{E}\bigg[\max_{\bt\in\frac{\be}{\bh}\mathcal{T}_{\bh}}|Y^{(2)}_{\bt}|\bigg]=\int_{0}^{\infty}\mathbb{P}\bigg(\max_{\bt\in\frac{\be}{\bh}\mathcal{T}_{\bh}}|Y^{(2)}_{\bt}|>\lambda\bigg) \, d\lambda\leq2\int_{0}^{\infty}\mathbb{P}\bigg(\max_{\bt\in\frac{\be}{\bh}\mathcal{T}_{\bh}}Y^{(2)}_{\bt}>\lambda\bigg) \, d\lambda\\
	&=2\int_{0}^{\frac{\mathcal{C}\sqrt{\log(n)}}{(nh_{\min})^{\gamma}}}\mathbb{P}\bigg(\max_{\bt\in\frac{\be}{\bh}\mathcal{T}_{\bh}}Y^{(2)}_{\bt}>\lambda\bigg) \, d\lambda+2\int_{\frac{\mathcal{C}\sqrt{\log(n)}}{(nh_{\min})^{\gamma}}}^{\infty}\mathbb{P}\bigg(\max_{\bt\in\frac{\be}{\bh}\mathcal{T}_{\bh}}Y^{(2)}_{\bt}>\lambda\bigg) \, d\lambda\\
	&\leq \tfrac{2\mathcal{C}\sqrt{\log(n)}}{(nh_{\min})^{\gamma}} +2\int_{\frac{\mathcal{C}\sqrt{\log(n)}}{(nh_{\min})^{\gamma}}}^{\infty}\mathbb{P}\bigg(\max_{\bt\in\frac{\be}{\bh}\mathcal{T}_{\bh}}Y^{(2)}_{\bt}>\lambda\bigg) \, d\lambda.
	\end{align*}
	Assertion \eqref{Y1bounded} now follows from the tail bound \eqref{tailbound}, integration by parts and a proper adjustment of the constant $\mathcal{C}.$ 
	We now show that for fixed $\bh\in\mathcal{H}^d,$ $Y_{\bt}^{(1)}$ and $Y_{\bt}^{(2)}$ are close to each other.
	\begin{align}\label{ExpY1Y2}
	\biggl|\mathbb{E}\Big[\Big|Y_\bt^{(1)}-Y_{\bt'}^{(1)}\Big|^2\Big]-\mathbb{E}\Big[\Big|Y_\bt^{(2)}-Y_{\bt'}^{(2)}\Big|^2\Big]\Big|\leq C\bigg(\bigg(\frac{1}{nh_{\min}}\bigg)^{2\gamma}+\bigg(\frac{1}{\log(n)\log\log(n)}\bigg)^2\bigg).
	\end{align}
	An application of Theorem 2.2.5 in \citet{adltay2007} yields
	\begin{align*}
	\bigg|\mathbb{E}\bigg[\max_{\bt\in\frac{\be}{\bh}\mathcal{T}_{\bh}}Y^{(1)}_{\bt}\bigg]-\mathbb{E}\bigg[\max_{\bt\in\frac{\be}{\bh}\mathcal{T}_{\bh}}Y^{(2)}_{\bt}\bigg]\bigg|\leq C\sqrt{\log(N)\biggl(\bigg(\frac{1}{nh_{\min}}\bigg)^{2\gamma}+\frac{1}{\log(n)^2\log\log(n)^2}\biggr)}.
	\end{align*}
	Recall that, by Assumption \eqref{R1}, we have $\gamma\in[1/2,1]$ and $h_{\min}\geq\frac{C_{\min}\log(n)^3}{n\log\log(n)^2}$. Since, by Assumption \eqref{growth}, we also have $\log(N)=O(\log(n))$ and we obtain
	\begin{align*}
	\bigg|\mathbb{E}\bigg[\max_{\bt\in\frac{\be}{\bh}\mathcal{T}_{\bh}}Y^{(1)}_{\bt}\bigg]-\mathbb{E}\bigg[\max_{\bt\in\frac{\be}{\bh}\mathcal{T}_{\bh}}Y^{(2)}_{\bt}\bigg]\bigg|\leq C_{1,2}\frac{1}{\sqrt{\log(n)}\log\log(n)},
	\end{align*}   		
	for some positive constant $C_{1,2}.$ Let $\lambda>2C_{1,2}+1$. Then
	\begin{align}\label{C12}
	&\mathbb{P}\biggl(\biggl|\max_{\bt\in\frac{\be}{\bh}\mathcal{T}_{\bh}}Y^{(1)}_{\bt}-\max_{\bt\in\frac{\be}{\bh}\mathcal{T}_{\bh}}Y^{(2)}_{\bt}\biggr|>\frac{\lambda}{\sqrt{\log(n)}\log\log(n)}\biggr)\\
	&~\leq\mathbb{P}\biggl(\biggl|\max_{\bt\in\frac{\be}{\bh}\mathcal{T}_{\bh}}Y^{(1)}_{\bt}-\mathbb{E}\Bigl[\max_{\bt\in\frac{\be}{\bh}\mathcal{T}_{\bh}}Y^{(1)}_{\bt}\Bigr]\biggr|>\frac{\lambda}{2\sqrt{\log(n)}\log\log(n)}\biggr)\\
	&~~+\mathbb{P}\biggl(\biggl|\max_{\bt\in\frac{\be}{\bh}\mathcal{T}_{\bh}}Y^{(2)}_{\bt}-\mathbb{E}\Bigl[\max_{\bt\in\frac{\be}{\bh}\mathcal{T}_{\bh}}Y^{(2)}_{\bt}\Bigr]\biggr|>\frac{\lambda-2C_{1,2}}{2\sqrt{\log(n)}\log\log(n)}\biggr).
	\end{align}
	For $j=1,2$, Var$(Y_{\bt}^{(j)})\leq 2L(nh_{\min})^{-2\gamma}\leq 2L\log\log(n)^2/(\log(n)^3C_{\min}),$ an application of  Borell's inequality to each of the two terms in \eqref{C12} yields 
	\begin{align*}
	&\mathbb{P}\biggl(\biggl|\max_{\bt\in\frac{\be}{\bh}\mathcal{T}_{\bh}}Y^{(1)}_{\bt}-\max_{\bt\in\frac{\be}{\bh}\mathcal{T}_{\bh}}Y^{(2)}_{\bt}\biggr|>\frac{\lambda}{\sqrt{\log(n)}\log\log(n)}\biggr)\leq4\exp\biggl(-\frac{\log(n)^2C_{\min}}{16L\log\log(n)^4}\biggr)=o(n^{-\kappa}).
	\end{align*}
	We can use the latter result to show that $Y_{\bt}^{(1)}(\bh)$ and $Y_{\bt}^{(2)}(\bh)$ are close to each other, uniformly with respect to $\bh\in \mathcal{H}^d$.
	\begin{align*}
	&\mathbb{P}\biggl(\biggl|\max_{\bh\in\mathcal{H}^d}\max_{\bt\in\frac{\be}{\bh}\mathcal{T}_{\bh}}Y^{(1)}_{\bt}(\bh)-\max_{\bh\in\mathcal{H}^d}\max_{\bt\in\frac{\be}{\bh}\mathcal{T}_{\bh}}Y^{(2)}_{\bt}(\bh)\biggr|>\frac{\lambda}{\sqrt{\log(n)}\log\log(n)}\biggr)\\
	&~\leq \mathbb{P}\biggl(\max_{\bh\in\mathcal{H}^d}\biggl|\max_{\bt\in\frac{\be}{\bh}\mathcal{T}_{\bh}}Y^{(1)}_{\bt}(\bh)-\max_{\bt\in\frac{\be}{\bh}\mathcal{T}_{\bh}}Y^{(2)}_{\bt}(\bh)\biggr|>\frac{\lambda}{\sqrt{\log(n)}\log\log(n)}\biggr)\\
	&~\leq\sum_{\bh\in\mathcal{H}^d}\mathbb{P}\biggl(\biggl|\max_{\bt\in\frac{\be}{\bh}\mathcal{T}_{\bh}}Y^{(1)}_{\bt}(\bh)-\max_{\bt\in\frac{\be}{\bh}\mathcal{T}_{\bh}}Y^{(2)}_{\bt}(\bh)\biggr|>\frac{\lambda}{\sqrt{\log(n)}\log\log(n)}\biggr)=o(1)\quad\text{as}\quad n\to\infty.
	\end{align*}
	In the same way, an application of Borell's inequality yields a tail bound for  $\max_{\bt\in\frac{\be}{\bh}\mathcal{T}_{\bh}}Y^{(2)}_{\bt}(\bh)$  	which shows that $\max_{\bh\in\mathcal{H}^d}\max_{\bt\in\frac{\be}{\bh}\mathcal{T}_{\bh}}Y^{(2)}_{\bt}(\bh)=O_{\mathbb{P}}\bigl((\log(n)\log\log(n))^{-\frac{1}{2}}\bigr).$ The assertion of the Lemma now follows.	
\end{proof}

\subsection{Gaussian Coupling}

\begin{Lemma}\label{Le:coupling} Suppose that Assumptions \ref{Ass:Dict} and \ref{Ass:Noise} hold.
	Then,  
	\begin{align*}
	\lim_{n\to\infty} \mathbb{P}_{0}\Bigl(\mathcal{S}(Y)\leq q_{1-\alpha}\Bigr)\geq 1-\alpha,
	\end{align*} 
	where $q_{1-\alpha}$ is such that $\mathbb P_{0} \left[\mathcal S \left(Y\right) > q_{1-\alpha}\right] \leq \alpha$.	
\end{Lemma}
\begin{proof}[Proof of Lemma \ref{Le:coupling}]  ~\\ 
	\textbf{Step I: Gaussian coupling}~\\Define 
	\begin{align}\label{X}
	\Xi_i:=\frac{\Phi_i}{n^{\frac{d}{2}}\|\sigma \Phi_i\|_2},\;\;
	X_{\bj}:=\bigl(\xi_{\bj}(\sigma\Xi_i)(\bx_{\bj})\bigr)_{i=1}^{N}\;\;
	S_{i,n}(\xi)=\sum_{\bj\in I_n^d}X_{\bj,i},\;\;
	\text{and}\;\; X:=\max_{1\leq i\leq N}S_{i,n}(\xi).
	\end{align}
	Consider a field of normally distributed random variables $\{\zeta_{\bj}\sim\mathcal{N}(0,\sigma^2(\bx_{\bj}))\,|\,\bj\in I_n^d\}$ and,
	as in \eqref{X}, define the quantities $Z_{\bj}:=\bigl(\zeta_{\bj}\sigma(\bx_{\bj})\Xi_i(\bx_{\bj})\bigr)_{i=1}^{N}$ and $Z:=\max_{1\leq i\leq N}S_{i,n}(\zeta).$
	By Corollary 4.1 in \citet{CheCheKat2014} it follows that 
	\begin{align}\label{Chernoresult}
	\mathbb{P}\bigl(|X-\widetilde{Z}|>16\Delta\bigr)
	\leq \frac{\bigl\{B_1+\Delta^{-1}(B_2+B_4)(d\vee \kappa)\log(n)\bigr\}(d\vee \kappa)\log(n)}{\Delta^{2}}+\frac{\Delta\log(n)}{n^d},
	\end{align}
	for some $\widetilde Z\stackrel{\mathcal{D}}{=}Z$ and all $\Delta>0$, where $B_1,$ $B_2$ and $B_4$ are defined as follows:
	\begin{align*}
	B_1=\mathbb{E}\biggl[\max_{1\leq i,l\leq N}\biggl|\sum_{\bj\in I_n^d} X_{\bj,i}X_{\bj,l}-\mathbb{E}[X_{\bj,k}X_{\bj,l}]\biggr|\biggr],\quad B_2=\mathbb{E}\biggl[\max_{1\leq i\leq N}\sum_{\bj\in I_n^d}|X_{\bj,i}|^3\biggr]
	\end{align*}  
	and 
	\begin{align*}
	B_4=\sum_{\bj\in I_n^d}\mathbb{E}\bigg[\max_{1\leq i\leq N} |X_{\bj,i}|^3 I\bigg\{\max_{1\leq i\leq N}|X_{\bj,i}|\geq \frac{\Delta}{(d\vee \kappa)\log(n)}\bigg\}\bigg].
	\end{align*}  
	Define
	\begin{align*}
	U_{\bj;l,i}:=X_{\bj,i}X_{\bj,l}-\mathbb{E}[X_{\bj,i}X_{\bj,l}] = (\xi_{\bj}^2-\sigma^2(\bx_{\bj}))\frac{\Phi_l(\bx_{\bj})\Phi_i(\bx_{\bj})}{n^d\|\sigma\Phi_l\|_2\|\sigma\Phi_i\|_2},
	\end{align*}
	\begin{align*}
	B_{i,l}:=\frac{\max_{\bj\in I_n^d}|\Phi_l(\bx_{\bj})\Phi_i(\bx_{\bj})|}{n^d\|\sigma \Phi_i\|_2\|\sigma\Phi_l\|_2}\leq\frac{\|\Phi\|_{\infty}^2}{c_{\sigma}^{2}(nh_{\min})^d}
	\quad\text{and}\quad
	v_{\bj;i,l}:=M\frac{\Phi_l^2(\bx_{\bj})\Phi_i^2(\bx_{\bj})}{n^{2d}\|\sigma \Phi_i\|_2^2\|\sigma\Phi_l\|_2^2}.
	\end{align*} 
	By the structure of the dictionary \eqref{Wavelet} and boundedness of $\Phi$ by Assumption \ref{Ass:Dict} d), we immediately obtain $B_{i,l}\leq\frac{\|\Phi\|_{\infty}^2}{c_{\sigma}^{2}(nh_{\min})^d},$ where $c_{\sigma}^{2}>0$  is a uniform lower bound on $\sigma^2$ which exists by Assumption \eqref{Moments2}.
	By the moment condition  \eqref{Moments} we obtain
	$
	\mathbb{E}|U_{\bj;i,l}|^J\leq  \frac{1}{2}J!\,v_{\bj;i,l}\cdot B_{i,l}^{J-2},\; J\geq 2.$
	Hence, we can bound $B_1$ using the Bernstein inequality \citep[cf.][Lemma 2.2.11]{VanWel1996} as follows. Using
	$
	\sum_{\bj\in I_n^d}v_{\bj;l,i}\leq2\frac{\|\Phi\|_{\infty}^4}{c_{\sigma}^{4}(nh_{\min})^d}
	$
	and
	\begin{align}\label{Bernstein}
	\mathbb{P}\biggl(\Bigl|\sum_{\bj\in I_n^d}U_{\bj;l,i}\Bigr|>x\biggr)\leq2\exp\biggl(-\tfrac{x^2}{2\sum_{\bj\in I_n^d}v_{\bj;l,i}+2xB_{i,l}}\biggr)\leq2\exp\biggl(-\tfrac{x^2(nh_{\min})^d}{4c_{\sigma}^{-4}\|\Phi\|_{\infty}^4+2x4c_{\sigma}^{-2}\|\Phi\|_{\infty}^2}\biggr),
	\end{align}
	we find, for an arbitrary constant $\mathcal{C}$,
	\begin{align*}
	B_1&\leq \mathcal{C}\sqrt{\log(n)}\biggl(\sum_{\bj\in I_n^d}v_{\bj;i,l}\biggr)^{\frac{1}{2}}+ N^2\max_{i,l}\mathbb{E}\biggl[\Big|\sum_{\bj\in I_n^d}U_{\bj;l,i}I\bigg\{\Big|\sum_{\bj\in I_n^d}U_{\bj;l,i}\Big|>\mathcal{C}\sqrt{\log(n)}\biggl(\sum_{\bj\in I_n^d}v_{\bj;i,l}\biggr)^{\frac{1}{2}}\bigg\}\Big|\biggr]\\
	&\leq O\biggl(\Bigl(\frac{1}{nh_{\min}}\Bigr)^{\frac{d}{2}}\sqrt{\log(n)}\biggr)+N^2\max_{i,l}\int_{\mathcal{C}\sqrt{\log(n)}\bigl(\sum_{\bj\in I_n^d}v_{\bj;i,l}\bigr)^{\frac{1}{2}}}^{\infty} \mathbb{P}\biggl(\biggl|\sum_{\bj\in I_n^d}U_{\bj;l,i}\biggr|>x\biggr)\,dx\\
	&+N^2\mathcal{C}\max_{i,l}\sqrt{\log(n)}\biggl(\sum_{\bj\in I_n^d}v_{\bj;i,l}\biggr)^{\frac{1}{2}}\mathbb{P}\biggl(\biggl|\sum_{\bj\in I_n^d}U_{\bj;l,i}\biggr|>\mathcal{C}\sqrt{\log(n)}\biggl(\sum_{\bj\in I_n^d}v_{\bj;i,l}\biggr)^{\frac{1}{2}}\biggr).
	\end{align*}
	For a sufficiently large choice of the constant $\mathcal{C}$, by \eqref{Bernstein}   we  find 
	\begin{align}\label{B1}
	B_1=O\biggl(\Bigl(\frac{1}{nh_{\min}}\Bigr)^{\frac{d}{2}}\sqrt{\log(n)}\biggr).
	\end{align}
	Recall that by assumption, the support of  $\Phi$ is contained in $[0,1]^d$ and $\Phi$ is uniformly bounded. Then, the moment assumptions \eqref{Moments} and \eqref{Moments2} imply
	\begin{align}\label{B2}
	B_2=\mathbb{E}\biggl[\max_{1\leq i\leq N}\sum_{\bj\in n[\bt_i-\bh_i,\bt_i]}|X_{\bj,i}|^3\biggr]\leq\frac{C}{(nh_{\min})^{\frac{d}{2}}}.
	\end{align} 
	For the estimation of the term $B_4$ we make use of the moment assumption \eqref{Moments} once more and find
	\begin{align*}
	\mathbb{E}[\exp(\theta|\xi_{\bj}^2-\sigma^2(\bx_{\bj})|)]\leq\exp\biggl(\theta\mathbb{E}|\xi_{\bj}^2-\sigma^2(\bx_{\bj})|+\frac{\theta^2\mathbb{E}|\xi_{\bj}^2-\sigma^2(\bx_{\bj})|^2}{2(1-\theta)}\biggr)\quad\text{for}\quad \theta\in[0,1).
	\end{align*}  
	Hence, for $\theta=0.5,$
	\begin{align}\label{MGF}
	\mathbb{E}[\exp(\theta|\xi_{\bj}^2-\sigma^2(\bx_{\bj})|)]\leq e^{\frac{3}{4}M},
	\end{align} 
	where the constant $M$ is a uniform upper bound on the fourth moments of the $\xi_{\bj}$ which exists due to  \eqref{Moments2}.
	We further obtain
	\begin{align*}
	B_4&\leq\frac{C}{(nh_{\min})^{\frac{3d}{2}}}\biggl(\sum_{\bj\in I_n^d}\mathbb{E}\Bigl[I\Bigl\{|\xi_{\bj}^2-\sigma^2(\bx_{\bj})|+\sigma^2(\bx_{\bj})>\Bigl(\frac{\Delta(nh_{\min})^{d/2}}{4\kappa\log(n)}\Bigr)^2\Bigr\}\Bigr]\biggr)^{\frac{1}{2}}\\
	&\leq\frac{Cn^d}{(nh_{\min})^{\frac{3d}{2}}}\biggl(\exp\Bigl(-\frac{c_{\sigma^2}}{16}(nh_{\min})^2\Delta^2/\log(n)^2\Bigr)\biggr)^{\frac{1}{2}}.
	\end{align*} 
	\textbf{Step II: The level is maintained}~\\
	Define
	\begin{align*}
	P_{0,n}:=\mathbb{P}_{0}\Bigl(\mathcal{S}(Y)\leq q_{1-\alpha}\Bigr)=\mathbb{P}_{0}\bigl(\langle Y,\Phi_i\rangle_n\leq q_{i,1-\alpha}\;\forall \;1\leq i\leq N\bigr)
	\end{align*}
	We now show that 
	$\lim_{n\to\infty} P_{0,n}\geq 1-\alpha$.
	We have that
	\begin{align*}
	P_{0,n}&=\mathbb{P}_{0}\biggl(\langle Y,\Phi_i\rangle_n\leq\sigma_i \Bigl(\tfrac{q_{1-\alpha}}{\omega_i}+\omega_i\Bigr)\;\forall \;1\leq i\leq N\biggr)\\
	&\geq\mathbb{P}_{0}\biggl(\langle Y,\Phi_i\rangle_n\leq \Bigl(\tfrac{q_{1-\alpha}}{\omega_i}+\omega_i\Bigr)\bigl(1-|r_n|\bigr)\;\forall \;1\leq i\leq N\biggr)\\
	&=\mathbb{P}_{0}\biggl(\tfrac{1}{n^{\frac{d}{2}}\|\Phi_{\bh_i}\|_2\sqrt{\bh_{i}^{\be}}}\sum_{\bj\in I_n^d}\xi_{\bj}\Phi_i(\bx_{\bj})\leq \Bigl(\tfrac{q_{1-\alpha}}{\omega_i}+\omega_i\Bigr)\bigl(1-|r_n|\bigr)\;\forall \;1\leq i\leq N\biggr),
	\end{align*}
	for $r_n=o(1/\log(n)),$ since $n^d\|\Phi_i\sigma\|_2^2=\mathrm{Var}(\langle Y,\Phi_i\rangle_n)+o(1/\log(n)),$ uniformly in $i.$\\
	\textbf{Step II.1: Reduction of the set of scales}~\\
	We now separate the variables $\bt_i$ and $\bh_{i}$ and proceed as in Step I in the proof of Theorem \ref{Thm:Gumbel} and consider only scales on a dyadic grid $\mathcal{H}_{\mathrm{dyad}}^d$.
	We find
	\begin{align*}
	P_{0,n}&\geq\mathbb{P}_{0}\biggl(\tfrac{1}{n^{\frac{d}{2}}\|\Phi_{\bh}\|_2\sqrt{\bh^{\be}}}\sum_{\bj\in I_n^d}\xi_{\bj}\Phi_{\bh}\Bigl(\frac{\bt-\bx_{\bj}}{\bh}\Bigr)\leq \Bigl(\tfrac{q_{1-\alpha}}{\omega_{\bh}}+\omega_{\bh}\Bigr)\bigl(1-|r_n|\bigr)\;\forall (\bt,\bh)\in\mathcal{T}_{\bh}\times\mathcal{H}_{\mathrm{dyad}}^d\biggr)\\
	&=\mathbb{P}_{0}\biggl(\tfrac{1}{n^{\frac{d}{2}}\|\Phi_{\bh_{\bp}}\|\sqrt{\bh_{\bp}^{\be}}}\sum_{\bj\in I_n^d}\xi_{\bj}\Phi_{\bh_{\bp}}\Bigl( \frac{\bt-\bx_{\bj}}{\bh_{\bp}}\Bigr)\leq \Bigl(\tfrac{q_{1-\alpha}}{\omega_{\bh_{\bp}}}+\omega_{\bh_{\bp}}\Bigr)\bigl(1-|r_n|\bigr)\;\forall (\bt,\bp)\in\mathcal{T}_{\bh}\times\mathcal{P}^d\biggr),
	\end{align*}
	where $\bh_{\bp}=(h_{p_1},\ldots, h_{p_d}).$
	Let further $\frac{\be}{\bh}\mathcal{T}_{\bh}$ denote the set $\{(\bt/\bh)\,|\,\bt\in\mathcal{T}_{\bh}\}$. Then
	\begin{align*}
	P_{n,0}\geq\mathbb{P}_{0}\biggl(\tfrac{1}{n^{\frac{d}{2}}\|\Phi_{\bh_{\bp}}\|\sqrt{\bh_{\bp}^{\be}}}\sum_{\bj\in I_n^d}\xi_{\bj}\Phi_{\bh_{\bp}}\Bigl(\bt- \frac{\bx_{\bj}}{\bh_{\bp}}\Bigr)\leq \Bigl(\tfrac{q_{1-\alpha}}{\omega_{\bh_{\bp}}}+\omega_{\bh_{\bp}}\Bigr)\bigl(1-|r_n|\bigr)\;\forall (\bt,p,q)\in\frac{\be}{\bh}\mathcal{T}_{\bh}\times\mathcal{P}^d\biggr),
	\end{align*}
	where, as usual, $\frac{\bx_{\bj}}{\bh_{\bp}}$ denotes component-wise division. Replace each set
	\begin{align*}
	\tfrac{\be}{\bh}\mathcal{T}_{\bh}\quad\text{by}\quad \tfrac{\be}{\bh}\widetilde{\mathcal{T}}_{\bh}:=\bigcup_{\widetilde{\bh}\leq\bh}\tfrac{\be}{\bh}\mathcal{T}_{\bh},
	\end{align*}
	where the inequality $\widetilde{\bh}\leq\bh$ is meant  component-wise. It follows that
	\begin{align*}
	P_{0,n}\geq\mathbb{P}_{0}\biggl(\tfrac{1}{n^{\frac{d}{2}}\|\Phi_{\bh_{\bp}}\|\sqrt{\bh_{\bp}^{\be}}}\sum_{\bj\in I_n^d}\xi_{\bj}\Phi_{\bh_{\bp}}\Bigl(\bt- \frac{\bx_{\bj}}{\bh_{\bp}}\Bigr)\leq \Bigl(\tfrac{q_{1-\alpha}}{\omega_{\bh_{\bp}}}+\omega_{\bh_{\bp}}\Bigr)\bigl(1-|r_n|\bigr)\;\forall (\bt,p,q)\in\frac{\be}{\bh}\widetilde{\mathcal{T}}_{\bh}\times\mathcal{P}^d\biggr),
	\end{align*}
	and the total number of $(\bt,\bh)$ considered is still of polynomial order in $n$.\\
	\textbf{Step II.2: Partition of the parameter set}~\\
	We partition the parameter set for $\bt$ in the same way as in step I.1 of the proof of Theorem \ref{Thm:Gumbel}. The blocks $B_{\bp}$ are constructed such that the suprema over different blocks are independent since supp$\Phi\subset[0,1]^d.$ This yields
	\begin{align*}
	\mathbf{q}_{\bk}:=\Bigl(\tfrac{q_{1-\alpha}}{\omega_{\bh_{\bk}}}+\omega_{\bh_{\bk}}\Bigr)=\min_{\bk\leq\bp}\Bigl(\tfrac{q_{1-\alpha}}{\omega_{\bh_{\bp}}}+\omega_{\bh_{\bp}}\Bigr).
	\end{align*}
	\begin{align*}
	P_{0,n}&\geq \prod_{\bk\in\mathcal{P}^d}\mathbb{P}\biggr(\max_{\bk\leq\bp}\max_{\bt\in B_{\bk}\cap\frac{\be}{\bh}\widetilde{\mathcal{T}}_{\bh}}\tfrac{1}{n^{\frac{d}{2}}\|\Phi_{\bh_{\bp}}\|\sqrt{\bh_{\bp}^{\be} }}\sum_{\bj\in I_n^d}\xi_{\bj}\Phi_{\bh_{\bp}}\Bigl( \bt-\frac{\bx_{\bj}}{\bh_{\bp}}\Bigr)\leq \mathbf{q}_{\bk}\bigl(1-|r_n|\bigr)\biggr)\\
	&\geq\prod_{\bk\in\mathcal{P}^d}\mathbb{P}\biggr(\max_{\bt\in B_{\bk}\cap\frac{\be}{\bh}\widetilde{\mathcal{T}}_{\bh}}\max_{\bp\in\mathcal{P}^d}\tfrac{1}{n^{\frac{d}{2}}\|\Phi_{\bh_{\bp}}\|\sqrt{\bh_{\bp}^{\be}}}\sum_{\bj\in I_n^d}\xi_{\bj}\Phi_{\bh_{\bp}}\Bigl( \bt-\frac{\bx_{\bj}}{\bh_{\bp}}\Bigr)\leq \mathbf{q}_{\bk}\bigl(1-|r_n|\bigr)\biggr).
	\end{align*}
	By the results of Step I, for each fixed multi-index $\bk$, we can now replace the variables $\xi_{\bj}$ by Gaussian ones. Each set of Gaussian random variables depends on $\bk$, but the corresponding distributions do not. By \eqref{Chernoresult}, it follows that
	\begin{align*}
	P_{0,n}&\geq\prod_{\bk\in\mathcal{P}^d}\biggl\{\mathbb{P}\biggl(\max_{\bt\in B_{\bk}\cap\frac{\be}{\bh}\widetilde{\mathcal{T}}_{\bh}}\max_{\bp\in\mathcal{P}^d}\tfrac{1}{n^{\frac{d}{2}}\|\Phi\|\sqrt{\bh_{\bp}^{\be}}}\sum_{\bj\in I_n^d}\zeta_{\bj}\Phi\Bigl( \bt-\frac{\bx_{\bj}}{\bh_{\bp}}\Bigr)>\mathbf{q}_{\bk}\bigl(1-|r_n|\bigr)-\Delta_n\biggr)\\
	&\times\biggl(1-C|\mathcal{P}|^d\Delta_n^{-2}\bigl\{B_1+\Delta_n^{-1}(B_2+B_4)\log(n)\bigr\}\log(n)+\frac{|\mathcal{P}|^d\Delta_n\log(n)}{n^d}\biggr)\biggr\}.
	\end{align*}
	By assumption, $h_{\min}\gtrsim n^{-1}\log(n)^{\frac{15}{d}\vee3}\log\log(n)^2$. Hence, choosing $\Delta_n=1/(\sqrt{\log(n)\log\log(n)})$, the results from Step I give, uniformly in $\bk$
	\begin{align*}
	|\mathcal{P}|^d\Delta_n^{-2}\bigl\{B_1+\Delta_n^{-1}(B_2+B_4)\log(n)\bigr\}\log(n)+\frac{|\mathcal{P}|^d\Delta_n\log(n)}{n^d}=O\biggl(\frac{1}{\log(n)^2\sqrt{\log(\log(n))}}\biggr).
	\end{align*}
	This yields
	\begin{align*}
	P_{0,n}&\geq\biggl\{\prod_{\bk\in\mathcal{P}^d}\mathbb{P}\biggl(\max_{\bt\in B_{\bk}\cap\frac{\be}{\bh}\widetilde{\mathcal{T}}_{\bh}}\max_{\bp\in\mathcal{P}^d}\tfrac{1}{n^{\frac{d}{2}}\|\Phi_{\bh_{\bp}}\|\sqrt{\bh_{\bp}^{\be}}}\sum_{\bj\in I_n^d}\zeta_{\bj}\Phi\Bigl( \bt-\frac{\bx_{\bj}}{\bh_{\bp}}\Bigr)\leq\mathbf{q}_{\bk}\bigl(1-|r_n|\bigr)-\Delta_n\biggr)\biggr\}\\
	&~\times\biggl(1-\frac{1}{\sqrt{\log\log(n)}|\mathcal{P}|^d}\biggr)^{|\mathcal{P}|^d}.
	\end{align*}
	Using Lemma \ref{Le:SumInt}, we replace the sum by the corresponding Wiener integral
	\begin{align*}
	P_{0,n}&\geq(1-o(1))\prod_{\bk\in\mathcal{P}^d}\mathbb{P}\biggl(\max_{\bt\in B_{\bk}}\int\Xi(\bt-\bz)\,\mathrm{d}W_{\bz}\leq\mathbf{q}_{\bk}\bigl(1-|r_n|\bigr)-o(1/\sqrt{\log(n)})\biggr),
	\end{align*}	
	where $\Xi_{\bp}=\Phi_{\bp}/\|\Phi_{\bp}\|.$ Finally,
	\begin{align*}
	P_{0,n}&\geq(1-o(1))\mathbb{P}\biggl(\max_{\bk\in\mathcal{P}^d}\omega_{\bk}\biggl(\max_{\bt\in B_{\bk}}\int\Xi(\bt-\bz)\,\mathrm{d}W_{\bz}-\omega_{\bk}\biggr)\leq q_{1-\alpha}-o(1)\biggr)\\
	&\geq(1-o(1))\mathbb{P}\biggl(\max_{\bk\in\mathcal{P}^d}\omega_{\bk}\biggl(\max_{\bt\in [\be,\be/\bh_{\bk}]}\int\Xi_{\bp}(\bt-\bz)\,\mathrm{d}W_{\bz}-\omega_{\bk}\biggr)\leq q_{1-\alpha}-o(1)\biggr),
	\end{align*}	
	which implies
	$
	\lim_{n\to\infty}P_{0,n}\geq1-\alpha.
	$
\end{proof}

\subsection{Unbounded Support}
\begin{Lemma}\label{Lemma:Unbounded}
	If $\Xi$ is of unbounded support but decays sufficiently fast, i.\,e., 
	\begin{align}\label{decay}
	\int(1+\|\bz\|^2)^{\frac{1}{2}}\Xi^2(\bz)\,d\bz<\infty,\tag{Dec}
	\end{align} 
	the results of  Theorem   \ref{GWSMain} still hold.
\end{Lemma}	
\begin{proof}[Proof of Lemma \ref{Lemma:Unbounded}]
	Let $\chi_n$ be a sequence of smooth  functions with supp$\chi_n=[-\log(n)^3,\log(n)^3]^d$. Define
	\begin{align*}
	\widetilde{\Xi}:=\chi_n\cdot\Xi\quad\text{and}\quad{\Xi}_{\Delta}:=\Xi-\widetilde{\Xi}
	\end{align*}
	and the corresponding process
	\begin{align*}
	\quad Z_{\Delta,\bt}:=\int\Xi_{\Delta}(\bt-\bz)\,\mathrm{d}W_{\bz}.
	\end{align*}
	Using \eqref{decay}, we obtain
	\begin{align*}
	\mathrm{Var}(Z_{\Delta,\bt})\leq\sum_{j=1}^d\int_{\{|t_j-z_j|>\log(n)^3\}}\Xi_{\Delta}(\bt-\bz)^2\,d\bz\leq\frac{C}{\log(n)^3}.
	\end{align*}
	With the same arguments as for \eqref{tailbound} in the proof of Lemma \ref{Le:SumInt} we obtain
	\begin{align*}
	\mathbb{P}\biggl(\sup_{\bt\in[\be,\be/\mathbf{h}]}Z_{\Delta,\bt}>\lambda/(\log(n)^2\log(\log(n)))\biggr)\leq C n^{-\lambda}.
	\end{align*}
	Hence, $Z_{\Delta,\bt}$ is asymptotically negligible. For the term $Z_{\Delta,\bt}$ the same arguments as in the proof of Theorem \ref{GWSMain} apply. We use the same slicing technique as in Step I.1. Then we partition the parameter set in the same way as in Step I.2 of the proof of Theorem \ref{GWSMain} separating the main blocks by stripes of widths and lengths $2\log(n)^3.$ Since $|1/h_{i}-1/h_{i+1}|\geq C/h_{\max}\gtrsim n^{\delta}$, those are still small in comparison. A careful inspection of the proof of Theorem \ref{GWSMain} shows that the same arguments apply in this case as well.
\end{proof}
\subsection{Proofs of the main results}\label{Sec:Proofs of the main reults}
\begin{proof}[Proof of Theorem \ref{Th:1}]
	The proof of this theorem is a combination of the subsequent applications of the auxiliary results of the previous subsection. The Gaussian approximation follows from Lemma \ref{Le:coupling}. The continuous approximation is valid due to Lemma \ref{Le:SumInt}. An application of Lemma \ref{Le:Loc} yields a distribution free approximation to which, finally, Theorem \ref{GWSMain} is applied to $\Xi=\Phi/\|\Phi\|_2.$
\end{proof}
\begin{proof}[Proof of Theorem \ref{Thm:Power}] Theorem \ref{Thm:Power} is a generalization of Theorems 4 and 6 in \citet{ExactScan} and some of the arguments of this proof are similar. 
	~\\
	\vspace{-0.5cm}
	\begin{itemize}
		\item[(a)] For $f\in \mathscr{S}_{\{\bt_\star\}}(\bh_\star, \mu_n)$ there exists exactly one $\bt_\star\in[\bh_\star,\be]$ such that $\mathrm{supp}(f)=[\bt_\star-\bh_\star,\bt_\star].$ In the worst case scenario, the signal is spread equally within its support because it is less likely to be detected with local tests other than $\mathcal{S}(Y,i_\star),$ i.e.
		\begin{align*}
		f_\star=\mu_nn^{-\frac{d}{2}}\frac{\|\Phi_{i}\|_2}{\|\varphi_{i}\|_1}I_{[\bt_\star-\bh_\star,\bt_\star]}.
		\end{align*}
		Recall that, throughout the proof of assertion (a), $\bh_i\equiv\bh_\star.$ Define 
		$\overline{\mathcal{N}}:=\{i\in I_N\,|\,|\bt_\star-\bt_i|\leq3\bh_\star\}.$ The set $\overline{\mathcal{N}}$ is chosen such that  for all  $i\in I_N\backslash\overline{\mathcal{N}}$,  $[\bt_\star-\bh_\star,\bt_\star]\cap[\bt_i-\bh_\star,\bt_i]=\emptyset$, which implies that $\langle f_\star,\varphi_{i}\rangle=0$.
		Define further
		\begin{align}\label{eq:DefV}
		V_{i}:=\frac{n^{\frac{d}{2}}}{\sigma(\bt_i)\|\Phi_{i}\|_2}\langle f_\star,\varphi_{i}\rangle+Z_{\bt_i,\bh_i}.
		\end{align}
		Let $q^{G}_{1-\alpha}$ denote the $(1-\alpha)$-quantile of the Gumbel limit distribution. Then
		\begin{align*}
		\mathbb{P}_{f_\star}\Big(\max_{1\leq i\leq N}\mathcal{S}(Y,i)>q_{1-\alpha}\Big)=\mathbb{P}_{f_\star}\Big(\max_{i\in I_N}\omega_{\bh_{\star}}(V_{i}-\omega_{\bh_{\star}})>q^{G}_{1-\alpha}\Big)+o(1)
		\end{align*}
		by continuity of the limit distribution. Furthermore,
		\begin{align*}
		&\mathbb{P}_{f_\star}\Big(\max_{i\in I_N}\omega_{\bh_{\star}}(V_{i}-\omega_{\bh_{\star}})>q^{G}_{1-\alpha}\Big)
		=\mathbb{P}_{f_\star}\Big(\max_{i\in \overline{\mathcal{N}}}\omega_{\bh_{\star}}(Z_{\bt_i,\bh_\star}-\omega_{\bh_{\star}})>q^{G}_{1-\alpha}\Big)\\
		&~~+\mathbb{P}_{f_\star}\Big(\max_{i\in \overline{\mathcal{N}}}\omega_{\bh_{\star}}(V_{i}-\omega_{\bh_{\star}})>q^{G}_{1-\alpha}\,\wedge\,\max_{i\in I_N\backslash\overline{\mathcal{N}}}\omega_{\bh_{\star}}(Z_{\bt_i,\bh_\star}-\omega_{\bh_{\star}})\leq q^{G}_{1-\alpha}\Big)\\
		&\geq \mathbb{P}_{f_\star}\Big(\max_{i\in I_N\backslash\overline{\mathcal{N}}}\omega_{\bh_{\star}}(Z_{\bt_i,\bh_\star}-\omega_{\bh_{\star}})>q^{G}_{1-\alpha}\Big)
		\\
		&+\mathbb{P}_{f_\star}\Big(\omega_{\bh_{\star}}(V_{i_{\star}}-\omega_{\bh_{\star}})>q^{G}_{1-\alpha}\Big)\cdot\mathbb{P}_{f_\star}\Big(\max_{i\in I_N\backslash\overline{\mathcal{N}}}\omega_{\bh_{\star}}(Z_{\bt_i,\bh_\star}-\omega_{\bh_{\star}})\leq q^{G}_{1-\alpha}\Big),
		\end{align*}
		by independence of $Z_{\bt_i,\bh_\star},i\in I_N\backslash\overline{\mathcal{N}}$ and $Z_{\bt_\star,\bh_\star}.$ 
		Let $\tau_{\overline{\mathcal{N}}}=\{\bt_i\,|\,i\in\overline{\mathcal{N}}\}.$
		Since 
		\begin{align}\label{eq:Nneg}
		\sup_{\bt\in\tau_{\overline{\mathcal{N}}}}Z_{\bt,\bh_\star}\stackrel{\mathcal{D}}{=}	\sup_{\bt\in\frac{1}{\bh_\star}\tau_{\overline{\mathcal{N}}}}Z_{\bt}\leq\sup_{\bt\in[\mathbf{0},\mathbf{3}]}Z_{\bt}=O_{\mathbb{P}}(1),
		\end{align}
		we have that
		\begin{align*}
		&\mathbb{P}_{f_\star}\Big(\max_{i\in I_N}\omega_{\bh_{\star}}(V_{i}-\omega_{\bh_{\star}})>q^{G}_{1-\alpha}\Big)
		\geq \mathbb{P}_{f_\star}\Big(\max_{i\in I_N}\omega_{\bh_{\star}}(Z_{\bt_i,\bh_\star}-\omega_{\bh_{\star}})>q^{G}_{1-\alpha}\Big)
		\\
		&+\mathbb{P}_{f_\star}\Big(\omega_{\bh_{\star}}(V_{i_\star}-\omega_{\bh_{\star}})>q^{G}_{1-\alpha}\Big)\cdot\mathbb{P}_{f_\star}\Big(\max_{i\in I_N}\omega_{\bh_{\star}}(Z_{\bt_i,\bh_\star}-\omega_{\bh_{\star}})\leq q^{G}_{1-\alpha}\Big)+o(1)\\
		&=\alpha+\mathbb{P}_{f_\star}\Big(\omega_{\bh_{\star}}(V_{i_\star}-\omega_{\bh_{\star}})>q^{G}_{1-\alpha}\Big)\cdot(1-\alpha)+o(1)
		\end{align*}
		by Theorem \ref{Thm:Gumbel}. 
		Since
		\begin{align*}
		\mathbb{P}_{f_\star}\Big(\omega_{\bh_{\star}}(V_{i_\star}-\omega_{\bh_{\star}})>q^{G}_{1-\alpha}\Big)
		=\mathbb{P}\bigg(\mathcal{N}(0,1)+\frac{\mu_n}{\sigma(\bt_\star)}>\frac{q^{G}_{1-\alpha}}{\omega_{\bh_\star}}+\omega_{\bh_\star}\bigg)
		\end{align*}
		and we conclude
		\begin{align*}
		&\mathbb{P}_{f_\star}\Big(\max_{i\in I_N}\omega_{\bh_{\star}}(V_{i}-\omega_{\bh_{\star}})>q^{G}_{1-\alpha}\Big)
		\geq\alpha+\tail\biggl(\omega_{\bh_\star}-\frac{\mu_n}{\sigma(\bt_\star)}\biggr)\cdot(1-\alpha)+o(1).
		\end{align*}
		Next we show the opposite direction, i.e.
		\begin{align*}
		\mathbb{P}_{f_\star}\Big(\max_{i\in I_N}\omega_{\bh_{\star}}(V_{i}-\omega_{\bh_{\star}})>q^{G}_{1-\alpha}\Big)
		\leq\alpha+\tail\biggl(\omega_{\bh_\star}-\frac{\mu_n}{\sigma(\bt_\star)}\biggr)\cdot(1-\alpha)+o(1).
		\end{align*}
		To this end, define
		\begin{align*}
		\mathcal{N}^{o}:=\{i\,|\,\langle\varphi_i,I_{[\bt_\star-\bh_{\star},\bt_\star]}\rangle>\|\varphi_i\|_1(1-\varepsilon_n)\},
		\end{align*}
		where $(\varepsilon_n)_{n\in\mathbb{N}}$ is a sequence of positive numbers such that $\varepsilon_n\to0$ but $\mu_n\varepsilon_n\to\infty.$
		Then, similar as in the previous step, we estimate
		\begin{align*}
		&\mathbb{P}_{f_\star}\Big(\max_{i\in I_N}\omega_{\bh_{\star}}(V_{i}-\omega_{\bh_{\star}})>q^{G}_{1-\alpha}\Big)\\
		&\leq \mathbb{P}_{f_\star}\Big(\max_{i\in I_N\backslash\overline{\mathcal{N}}}\omega_{\bh_{\star}}(Z_{\bt_i,\bh_\star}-\omega_{\bh_{\star}})>q^{G}_{1-\alpha}\Big)+\mathbb{P}_{f_\star}\Big(\max_{i\in \overline{\mathcal{N}}\backslash\mathcal{N}^o}\omega_{\bh_{\star}}(V_{i}-\omega_{\bh_{\star}})>q^{G}_{1-\alpha}\Big)
		\\
		&+\mathbb{P}_{f_\star}\Big(\max_{i\in \mathcal{N}^o}\omega_{\bh_{\star}}(V_{i}-\omega_{\bh_{\star}})>q^{G}_{1-\alpha},\max_{i\in I_N\backslash\overline{\mathcal{N}}}
		\omega_{\bh_{\star}}(Z_{\bt_i,\bh_\star}-\omega_{\bh_{\star}})\leq q^{G}_{1-\alpha}\Big)+o(1).
		\end{align*}
		As before, this gives
		\begin{align}\label{eq:PowerUpper}
		&\mathbb{P}_{f_\star}\Big(\max_{i\in I_N}\omega_{\bh_{\star}}(V_{i}-\omega_{\bh_{\star}})>q^{G}_{1-\alpha}\Big)\leq \alpha+ \mathbb{P}_{f_\star}\Big(\max_{i\in \overline{\mathcal{N}}\backslash\mathcal{N}^o}\omega_{\bh_{\star}}(V_{i}-\omega_{\bh_{\star}})>q^{G}_{1-\alpha}\Big)\notag\\
		&+(1-\alpha)\cdot \mathbb{P}_{f_\star}\Big(\max_{i\in \mathcal{N}^o}\omega_{\bh_{\star}}(V_{i}-\omega_{\bh_{\star}})>q^{G}_{1-\alpha}\Big) +o(1),
		\end{align}
		as $\overline{\mathcal{N}}$ is negligible and $(Z_{\bt_i,\bh_\star},i\in\mathcal{N}^o)$ and $(Z_{\bt_i,\bh_\star},i\in I_N\backslash\overline{\mathcal{N}})$ are independent. 
		Since $\overline{\mathcal{N}}\backslash\mathcal{N}^o\subset\overline{\mathcal{N}}$ we further obtain $\max_{i\in \overline{\mathcal{N}}\backslash\mathcal{N}^o}|Z_{\bt_i,\bh_\star}|=O_{\mathbb{P}}(1)$ by \eqref{eq:Nneg}. For all $i\in \overline{\mathcal{N}}\backslash\mathcal{N}^o$ we have $\langle\varphi_i,I_{[\bt_\star-\bh_{\star},\bt_\star]}\rangle\leq\|\varphi_i\|_1(1-\varepsilon_n)$. Hence,
		\begin{align*}
		&\mathbb{P}_{f_\star}\Big(\max_{i\in \overline{\mathcal{N}}\backslash\mathcal{N}^o}\omega_{\bh_{\star}}(V_{i}-\omega_{\bh_{\star}})>q^{G}_{1-\alpha}\Big)\leq\mathbb{P}_{f_\star}\Big(\omega_{\bh_\star}\Big(\frac{\mu_n}{\sigma(\bt_\star)}(1-\varepsilon_n)+O_{\mathbb{P}}(1)-\omega_{\bh_\star}\Big)>q^{G}_{1-\alpha}\Big)\\
		&=\mathbb{P}_{f_\star}\Big(\omega_{\bh_\star}\Big(\frac{\mu_n}{\sigma(\bt_\star)}-\omega_{\bh_\star}\Big)+\omega_{\bh_\star}\Big(-\frac{\mu_n}{\sigma(\bt_\star)}\varepsilon_n+O_{\mathbb{P}}(1)\Big)>q^{G}_{1-\alpha}\Big)=o(1),
		\end{align*}
		if $\frac{\mu_n}{\sigma(\bt_\star)}-\omega_{\bh_\star}\to C\in[-\infty,\infty).$ 
		If $\frac{\mu_n}{\sigma(\bt_\star)}-\omega_{\bh_\star}\to\infty$, the asymptotic expansion in \eqref{eq:ExpPower} holds in any case, since the asymptotic power is equal to 1. It remains to show that
		\begin{align*}
		\mathbb{P}_{f_\star}\Big(\max_{i\in \mathcal{N}^o}\omega_{\bh_{\star}}(V_{i}-\omega_{\bh_{\star}})>q^{G}_{1-\alpha}\Big)=
		\tail\bigg(\sqrt{2\log\left(\tfrac{1}{\bh_\star^{\be}}\right)}-\frac{\mu_n}{\sigma(\bt_{\star})}\bigg)+o(1).
		\end{align*}
		Obviously, $\langle\varphi_i,I_{[\bt_i-\bh_\star,\bt_i]}\rangle=\|\varphi_i\|_1$ and $\langle\varphi_i,I_{[\bt-\bh_\star,\bt]}\rangle$ is strictly decreasing in $\|\bt-\bt_i\|^{\frac{1}{2}\wedge\gamma}_2,$ such tat for any $i\in\mathcal{N}^o$, we have $\|\bt_i-\bt_{\star}\|_2^{\frac{1}{2}\wedge\gamma}\leq C\bh_\star^{\be}\varepsilon_n$ and 
		for all $i\in\mathcal{N}^o$ we have
		\begin{align*}
		\rho((\bt_\star,\bh_\star);(\bt_i,\bh_\star))=O\biggl(\sum_{j=1}^d\biggl|\frac{t_{i,j}-t_{\star,j}}{h_{\star,j}}\biggr|^{\gamma}\biggr)
		=o(1),
		\end{align*}
		where $\rho$ is defined in \eqref{eq:rho}. As in \eqref{eq:Dudley} in the proof of Theorem \ref{Thm:Gumbel} we conclude
		\begin{align*}
		\max_{i\in\mathcal{N}^o}\omega_{\bh_\star}|Z_{\bt,\bh_\star}-Z_{\bt_\star,\bh_\star}|=o_{\mathbb{P}}(1).
		\end{align*}	
		Since $\varphi$ is at least as smooth as $\Phi$, we obtain
		\begin{align*}
		\biggl\|\varphi\biggl(\frac{\bt}{\bh_\star}-\cdot\biggr)-\varphi\biggl(\frac{\bt_\star}{\bh_\star}-\cdot\biggr)\biggr\|_2=O\biggl(\mu_n\biggl\|\frac{\bt}{\bh_\star}-\frac{\bt_\star}{\bh_\star}\biggr\|_2^{\gamma}\biggr)=o(1),
		\end{align*}
		for all $i\in\mathcal{N}^o.$
		By assumption, $\sigma\in C^1[0,1]^d$ and is bounded below, which gives
		\begin{align*}
		\frac{n^{\frac{d}{2}}\langle f,\varphi_{i}\rangle}{\sigma(\bt_i)\|\Phi_{i_\star}\|_2}=\frac{n^{\frac{d}{2}}\langle f,\varphi_{i_\star}\rangle}{\sigma(\bt_\star)\|\Phi_{i_\star}\|_2}(1+o(1)),
		\end{align*}
		uniformly in $i$.
		This yields
		\begin{align*}
		\mathbb{P}_{f_\star}\Big(\max_{i\in \mathcal{N}^o}\omega_{\bh_{\star}}(V_{i}-\omega_{\bh_{\star}})>q^{G}_{1-\alpha}\Big)=
		\mathbb{P}_{f_\star}\Big(\mathcal{N}(0,1)+o_{\mathbb{P}}(1)+\frac{\mu_n}{\sigma(\bt_\star)}(1+o(1))>\frac{q^{G}_{1-\alpha}}{\omega_{\bh_\star}}+\omega_{\bh_\star}\Big).
		\end{align*}		
		
		The first assertion of this theorem now follows from Slutzky's lemma.
		\item[(b)] In order to show the second assertion of this theorem we can proceed as in part (a) and we only need to show the first direction. Since now the scales vary, a suitable set $\overline{\mathcal{N}}$ might be considerably larger than in the first part. For any $i\in I_N,$ set
		\begin{align*}
		\widetilde{\mathcal{N}}:=\{i\in I_N\,|\,(\bt_i-\bh_i)\vee(\bt_\star-\bh_\star)\leq\bt_i\wedge\bt_\star\}\subset[\bt_\star-\bh_\star,\bt_\star+\bh_i].
		\end{align*}
		We now show that $ \widetilde{\mathcal{N}}$ is asymptotically negligible compared to $I_N,$ i.\,e.
		\begin{align*}
		\mathbb{P}\biggl(\max_{i\in\widetilde{\mathcal{N}}}\omega_{\bh_i}(Z_{\bt_i,\bh_i}-\omega_{\bh_i})>q^{G}_{1-\alpha}\biggr)=o(1).
		\end{align*}
		Split the set $I_N$ into
		\begin{align*}
		I_N=\{i\,|\,\bh_\star/\bh_i\leq\log(n)\}\cup\{i\,|\,\bh_\star/\bh_i>\log(n)\}=:I_{N,1}\cup I_{N,2}.
		\end{align*}
		Then
		\begin{align*}
		&\mathbb{P}\biggl(\max_{i\in\widetilde{\mathcal{N}}\cap I_{N,1}}\omega_{\bh_i}(Z_{\bt_i,\bh_i}-\omega_{\bh_i})>q^{G}_{1-\alpha}\biggr)\\
		&\leq\mathbb{P}\biggl(\max_{i\in\widetilde{\mathcal{N}}\cap I_{N,1}}\omega_{\bh_i}(O_{\mathbb{P}}(\sqrt{\log\log(n)})-\sqrt{2d\delta\log(n)})>q^{G}_{1-\alpha}\biggr)=o(1).
		\end{align*}
		For $I_{N,2}$ we proceed as in the proof of Theorem \ref{Thm:Gumbel} and obtain 
		\begin{align*}
		&\mathbb{P}\biggl(\max_{i\in\widetilde{\mathcal{N}}\cap I_{N,2}}\omega_{\bh}(Z_{\bt_i,\bh_i}-\omega_{\bh_i})>q^{G}_{1-\alpha}\biggr)\\
		&~=1-\prod_{\bp\in\mathcal{P}^d}\biggl\{1-\mathbb{P}\biggl(\sup_{\bt\in \bh_\star B_{\bp}}\int\Xi(\bt-\bz)\,\mathrm{d}W_{\bz}>\frac{q_{1-\alpha}^G}{\omega_{\bh_{\bp}}}+\omega_{\bh_{\bp}}\biggr)\biggr\}+o(1).
		\end{align*}
		By Borell's inequality we estimate
		\begin{align*}
		&\mathbb{P}\biggl(\sup_{\bt\in \bh_\star B_{\bp}}\int\Xi(\bt-\bz)\,\mathrm{d}W_{\bz}>\frac{q_{1-\alpha}^G}{\omega_{\bh_{\bp}}}+\omega_{\bh_{\bp}}\biggr)\leq\exp\biggl(-\frac{1}{2}\biggl(\frac{q_{1-\alpha}^G}{\omega_{\bh_{\bp}}}+\omega_{\bh_{\bp}}-\mathbb{E}\Bigl[\sup_{\bt\in \bh_\star B_{\bp}}Z_{\bt}\Bigr]\biggr)^2\biggr)\\
		&~\leq\exp\biggl(-\frac{1}{2}\Bigl(\sqrt{2\log(1/\bh_{\bp}^{\be})}-\sqrt{2\log(\bh_\star^{\be}/\bh_{\bp}^{\be})}\Bigr)^2\biggr)\big(1+o(1)\big).
		\end{align*}
		Hence, there exists some positive constant $C_\star>0,$ independent of $n$, such that
		\begin{align*}
		\mathbb{P}\biggl(\sup_{\bt\in \bh_\star B_{\bp}}\int\Xi(\bt-\bz)\,\mathrm{d}W_{\bz}>\frac{q_{1-\alpha}^G}{\omega_{\bh_{\bp}}}+\omega_{\bh_{\bp}}\biggr)\leq\exp\biggl(-C_\star\log(1/\bh_\star^{\be})\biggr)\big(1+o(1)\big)
		=O\bigl(\bh_\star^{\be}\bigr)^{C_\star}.
		\end{align*}
		Hence, there exists a positive constant $\delta_\star>0$ such that
		\begin{align*}
		&\mathbb{P}\biggl(\sup_{\bh\in\mathcal{H}_2}\sup_{\bt\in\mathcal{N}(\bh)}\omega_{\bh}(Z_{\bt,\bh}-\omega_{\bh})>q^{G}_{1-\alpha}\biggr)\\
		&~\sim1-\prod_{\bp\in\mathcal{P}^d}\biggl\{1-\mathbb{P}\biggl(\sup_{\bt\in \bh_\star B_{\bp}}\int\Xi(\bt-\bz)\,\mathrm{d}W_{\bz}>\frac{q_{1-\alpha}^G}{\omega_{\bh_{\bp}}}+\omega_{\bh_{\bp}}\biggr)\biggr\}
		~\leq1-\prod_{\bp\in\mathcal{P}^d}\Bigl\{1-n^{-\delta_\star}\Bigr\}.
		\end{align*}
		We further have
		\begin{align*}
		\log\biggl(\prod_{\bp\in\mathcal{P}^d}\Bigl\{1-n^{-\delta_\star}\Bigr\}\biggr)=\sum_{\bp\in\mathcal{P}^d}\log\Bigl\{1-n^{-\delta_\star}\Bigr\}\sim-\sum_{\bp\in\mathcal{P}^d}n^{-\delta_\star}\to0 \quad\text{as}\quad n\to\infty,
		\end{align*}
		since $|\mathcal{P}|=O(\log(n)).$ This implies
		\begin{align*}
		\prod_{\bp\in\mathcal{P}^d}\Bigl\{1-n^{-\delta_\star}\Bigr\}\to1 \quad\text{as}\quad n\to\infty,
		\end{align*}
		We finally obtain
		\begin{align*}
		\mathbb{P}\biggl(\max_{i\in\widetilde{\mathcal{N}}\cap I_{N,2}}\omega_{\bh_i}(Z_{\bt,\bh_i}-\omega_{\bh_i})>q^{G}_{1-\alpha}\biggr)=o(1),
		\end{align*}
		which concludes the proof of the second assertion of this theorem.
	\end{itemize} 
	
\end{proof}

\begin{proof}[Proof of Lemma \ref{Le:WRadon}]
	
	The elements $\Phi_i\in\mathcal{W}$ need to satisfy the following requirement
	\begin{align*}
	\langle f,\varphi_i\rangle_{L^2(\R^d)}=\langle Tf,\Phi_i\rangle_{L^2(\R\times\mathbb{S}^{d-1})}.
	\end{align*}
	On the one hand we find
	\begin{align*}
	\langle f,\varphi_i\rangle_{L^2(\R^d)}=h_i^{-d/2}\int f(\bx)\varphi((\bx-\bt_i)/h_i)\,d\bx=\frac{h_i^{d/2}}{(2\pi)^d}\int\mathcal{F}_df(\bxi)\mathcal{F}_d\varphi(h_i\bxi)\exp\left(\textup{i}\langle\bxi,\bt_i\rangle\right)\,\mathrm d\bxi.
	\end{align*}
	Introducing polar coordinates $(\bxi\mapsto s\pmb{\vartheta})$ yields
	\begin{align*}
	\langle f,\varphi_i\rangle_{L^2(\R^d)}=\frac{h_i^{d/2}}{(2\pi)^d}\int_{\mathbb{S}^{d-1}}\int_{\R^+}\mathcal{F}_df(s\pmb{\vartheta})\mathcal{F}_d\varphi(h_is\pmb{\vartheta})\exp\left(\textup{i}\langle s\pmb{\vartheta},\bt_i\rangle\right)s^{d-1}\,\mathrm ds\,\mathrm d\pmb{\vartheta}.
	\end{align*}
	On the other hand
	\begin{align*}
	\langle Tf,\Phi_i\rangle_{L^2(\R\times\mathbb{S}^{d-1})}&=2\int_{\mathbb{S}^{d-1}}\int_{\R^+}Tf(u,\pmb{\vartheta})\Phi_i(u,\pmb{\vartheta})\,\mathrm du\,\mathrm d\pmb{\vartheta}\\
	&=\frac{2}{2\pi}\int_{\mathbb{S}^{d-1}}\int_{\R^+}\mathcal{F}_1(Tf(\cdot,\pmb{\vartheta}))(s)\mathcal{F}_1(\Phi_i(\cdot,\pmb{\vartheta}))(s)\,\mathrm ds\,\mathrm d\pmb{\vartheta}
	\end{align*}
	by Plancherel's theorem. Since $(\mathcal{F}_1Tf(\cdot,\pmb{\vartheta}))(s)=\mathcal{F}_df(s\pmb{\vartheta})$ by Theorem 1.1 in \cite{n86} we further deduce
	\begin{align*}
	\langle Tf,\Phi_i\rangle_{L^2(\R\times\mathbb{S}^{d-1})}=\frac{2}{2\pi}\int_{\mathbb{S}^{d-1}}\int_{\R^+}\mathcal{F}_df(s\pmb{\vartheta})\mathcal{F}_1(\Phi_i(\cdot,\pmb{\vartheta}))(s)\,\mathrm ds\,\mathrm d\pmb{\vartheta}.
	\end{align*}
	This yields the condition
	\begin{align*}
	\mathcal{F}_1(\Phi_i(\cdot,\pmb{\vartheta}))(s)=\frac{h_i^{(2-d)/2}}{2(2\pi)^{d-1}}(\mathcal{F}_d\varphi)(h_is\pmb{\vartheta})\exp\left(\textup{i}\langle s\pmb{\vartheta},\bt_i\rangle\right)|h_is|^{d-1},
	\end{align*}
	which implies
	\begin{align*}
	\Phi_i(u,\pmb{\vartheta})=\frac{h_i^{-d/2}}{2(2\pi)^{d}}\mathcal{F}_1\left(\left(\mathcal{F}_d\varphi\right)\left(\cdot\pmb{\vartheta}\right)|\cdot|^{d-1}\right)\biggl(\frac{u-\langle\pmb{\vartheta},\bt_i\rangle}{h_i}\biggr).
	\end{align*}
	Note that due to the rotational invariance of $\varphi$, the function 
	\[
	\Phi \left(x\right) := \frac{1}{2(2\pi)^{d}}\mathcal{F}_1\left(\left(\mathcal{F}_d\varphi\right)\left(\cdot\pmb{\vartheta}\right)|\cdot|^{d-1}\right)\left(x\right), \qquad x\in \mathbb R
	\]
	is in fact independent of $\pmb{\vartheta}$.
\end{proof}

\begin{proof}[Proof of Lemma \ref{Le:AHCbRadon}]
	
	Since
	\begin{align*}
	\mathcal{F}_1(g(x-\cdot))(v)=\int g(x-z)e^{izv}\,\mathrm d z=e^{ixv}\int g(y)e^{-iyv}\,\mathrm d z=2\pi e^{ixv}\mathcal{F}_1^{-1}g(v),
	\end{align*}
	we obtain with Plancherel's theorem
	\begin{align*}
	&\int_{\mathbb{S}^{d-1}}\int_{\R}\left|\Xi\left(\langle\bt,\pmb{\vartheta}\rangle-u\right)-\Xi\left(\langle\bs,\pmb{\vartheta}\rangle-u\right)\right|^2\,\mathrm du\,\mathrm d\pmb{\vartheta}\\
	=& 2\pi\int_{\mathbb{S}^{d-1}}\int_{\R}\Bigl|\exp\left(\textup{i}\langle u\pmb{\vartheta},\bt\rangle\right)-\exp\left(\textup{i}\langle u\pmb{\vartheta},\mathbf{s}\rangle\right)\Bigr|^2 \big|\big(\mathcal{F}_1^{-1}\Xi\big)(u)\big|^2 \,\mathrm du\,\mathrm d\pmb{\vartheta}\notag\\
	&=C_{\varphi,d}\int_{\mathbb{S}^{d-1}}\int_{\R^+}\Bigl|\exp\left(\textup{i}\langle u\pmb{\vartheta},\bt\rangle\right)-\exp\left(\textup{i}\langle u\pmb{\vartheta},\mathbf{s}\rangle\right)\Bigr|^2\Bigl|(\mathcal{F}_d)\varphi(u\pmb{\vartheta})|u\pmb{\vartheta}|^{d-1}\Bigr|^2\,\mathrm du\,\mathrm d\pmb{\vartheta}.
	\end{align*}
	We now go back to Euclidean coordinates and obtain
	\begin{align*}
	&\int_{\mathbb{S}^{d-1}}\int_{\R}\left|\Xi\left(\langle\bt,\pmb{\vartheta}\rangle-u\right)-\Xi\left(\langle\bs,\pmb{\vartheta}\rangle-u\right)\right|^2\,\mathrm du\,\mathrm d\pmb{\vartheta}\\
	&=C_{\varphi,d}\int_{\R^d}\Bigl|\exp\left(\textup{i}\langle\pmb{\omega},\bt\rangle\right)-\exp\left(\textup{i}\langle\pmb{\omega},\mathbf{s}\rangle\right)\Bigr|^2\bigl|(\mathcal{F}_d\varphi)(\pmb{\omega})\bigr|^2\|\pmb{\omega}\|_2^{d-1}\,\mathrm d\pmb{\omega}\\
	&=C_{\varphi,d}\sum_{j,k=1}^d(t_k-s_k)(t_j-s_j)\int_{\mathbb{R}^d}\omega_j\omega_k\|\pmb{\omega}\|_2^{d-1}\big|(\mathcal{F}_d\varphi)(\pmb{\omega})\bigr|^2\,\mathrm d\pmb{\omega}+o(\|\bs-\bt\|_2^2)
	\end{align*} 
	using Taylor's Theorem. This implies \eqref{S2} with
	\begin{align*}
	D_{\Xi}^{-2}:= 	\mathrm{diag}\biggl(C_{\varphi,d}\int_{\mathbb{R}^d}\omega_1^2\|\pmb{\omega}\|_2^{d-1}\big|(\mathcal{F}_d\varphi)(\pmb{\omega})\bigr|^2\,\mathrm d\pmb{\omega}\biggr),
	\end{align*}
	by rotational symmetry, if $(\mathcal{F}_d\varphi)$ decays sufficiently fast.
\end{proof}
\begin{proof}[Proof of Theorem \ref{Thm:Radon}]
In this example of the Radon transform, the Gaussian approximation has a slightly different structure as the integral is with respect to white noise on a non-Euclidean space. We now briefly recall the necessary definitions, following \citet{adltay2007}, Chapter 1.4.3.
Consider the $\sigma$-finite measure $\nu$ on $\mathcal{B}\bigl(\R\times\mathbb{S}^{d-1}\bigr)$ defined by
\[A\mapsto \nu(A) = \int_{\mathbb{S}^{d-1}} \int_\R 1_{A}(s,u) \, du \, d\pmb{\vartheta},\]
where $d\pmb{\vartheta}$ is the common surface-measure on $\mathbb{S}^{d-1}$, as discussed in the main text.
Define 
\[\bigl(\mathcal{B}(\R\times\mathbb{S}^{d-1})\bigr)_{\nu}:=\bigl\{A\in\mathcal{B}\bigl(\R\times\mathbb{S}^{d-1}\bigr)\,\bigl|\,\nu(A)<\infty\bigr\}\]
and let $W$  be Gaussian noise on $\mathcal{Z}$ based on $\nu$, i.e.~a Gaussian random set function such that 
for $A,B\in\bigl(\mathcal{B}(\R\times\mathbb{S}^{d-1})\bigr)_{\nu},\;A\cap B=\emptyset$ we have $W(A\cup B)=W(A)+W(B)$ a.s., 
\[ W(A) \sim \mathcal{N}\big(0, \nu(A) \big),\]
and $W(A)$ and $W(B)$ are independent. The integral with respect to the (random, signed) measure $W$ is now defined as $L^2$-limit of integrals of elementary functions. Analogously to the Wiener sheet on a Euclidean space, a point-indexed version can be defined via of the set-indexed white noise $W$ can be defined via parametrization with polar coordinates. The resulting parametrized Wiener integral is of the same structure as the integral with respect to $(W_\bz)_{\bz\in[0,1]^d}$ since $\mathrm{d}\pmb{\vartheta}$ does not depend on the variable $u$. The results of the auxiliary Lemmas \ref{Le:SumInt} and \ref{Le:coupling} therefore transfer to this setting.\\
Let now $\Xi_i:=\Phi_{i}/\|\Phi_{i}\|_{L^2(\R\times\mathbb{S}^{d-1})}$ and $\Phi$ as in \eqref{DefPhiRadon}.	Since $\varphi$ has bounded support, its Fourier transform, $\mathcal{F}_d\varphi$ is smooth. Hence, the smoothness of $\mathcal{F}_d\varphi(\cdot\pmb{\vartheta})|\cdot|^{d-1}$ is determined by the smoothness of $|\cdot|^{d-1}.$ Hence,  for all $d\in \mathbb{N},$ $\mathcal{F}_d\varphi(\cdot\pmb{\vartheta})|\cdot|^{d-1}$ is at least $d-1$ times weakly differentiable with square integrable weak derivative, i.\,e. $\mathcal{F}_d\varphi(\cdot\pmb{\vartheta})|\cdot|^{d-1}\in H^{d-1}(\R)$ implying that $\mathcal{F}_1(\mathcal{F}_d\varphi(\cdot\pmb{\vartheta})|\cdot|^{d-1})(\xi)|\xi|^{(d-1)}\in L^2(\R).$ Hence, $\Phi$ decays fast enough such that the results of the limit theorem still hold by Lemma \ref{Lemma:Unbounded}. Let
	\begin{align*}
	\Delta_{\bt,h}:=\int\Xi_i\biggl(\frac{u-\langle\pmb{\vartheta,\bt}\rangle}{h}\biggr)\,\mathrm{d}W_{u,\pmb{\vartheta}}-\int\bigl(\Xi_i\cdot I_{[-1,1]}\bigr)\biggl(\frac{u-\langle\pmb{\vartheta,\bt}\rangle}{h}\biggr)\,\mathrm{d}W_{u,\pmb{\vartheta}}.
	\end{align*}
	Since \eqref{decay} holds, 
	by a change of variables, there exists a constant $\tau>0$ such that
	\begin{align*}
	\mathrm{Var}(\Delta_{\bt,h})\leq C\int_{\{|u|>\rho/h\}}\Xi_{i}(u)^2\,du\leq\frac{C}{n^{\tau}}.
	\end{align*}
	With the same arguments as used to prove \eqref{tailbound} in the proof of Lemma \ref{Le:SumInt} we obtain
	\begin{align*}
	\mathbb{P}\biggl(\sup_{h\in\mathcal{H}}\sup_{\bt\in[\mathbf{0},\be-\pmb{\rho}]}\Delta_{\bt,h}>\lambda/(\log(n)^2\log(\log(n)))\biggr)\leq C n^{-\lambda}.
	\end{align*}
	Moreover, from Lemma \ref{Le:AHCbRadon} we know that condition \eqref{S2} holds with
	\begin{align*}
	(D_{\Xi}D_{\Xi}^*)^{-1}= 
	\mathrm{diag}\biggl(C_{\varphi,d}\int_{\mathbb{R}^d}\omega_1^2\|\pmb{\omega}\|_2^{d-1}\big|(\mathcal{F}_d\varphi)(\pmb{\omega})\bigr|^2\,\mathrm d\pmb{\omega}\biggr),
	\end{align*}
	where $C_{\varphi,d}=4\pi\|\mathcal{F}_1\big((\mathcal{F}_d\varphi)(\cdot\pmb{\vartheta})|\cdot|^{d-1}\big)(u-\langle \bt_i,\pmb{\vartheta}\rangle)\|_{L^2(\R\times\mathbb{S}^{d-1})}^{-1}$.
	The assertion of this Theorem now follows.
\end{proof}

\begin{proof}[Proof of Theorem \ref{Thm:Conv2}]

	Let $\Xi_i:=\Phi_{\bh_i}/\|\Phi_{\bh_i}\|_2$
	\begin{align*}
	\Delta_{\bt,\bh}:=\frac{1}{\sqrt{\bh^{\be}}}\int\Xi_i\biggl(\frac{\bt-\bz}{\bh}\biggr)\,\mathrm{d}W_{\bz}-\frac{1}{\sqrt{\bh^{\be}}}\int\bigl(\Xi_i\cdot I_{[\mathbf{0},\be]^d}\bigr)\biggl(\frac{\bt-\bz}{\bh}\biggr)\,\mathrm{d}W_{\bz}.
	\end{align*}
	By \eqref{D1} and  $\varphi\in H^{2a+\gamma\vee 1/2}(\mathbb{R}^d)$ it is straightforward to see that \eqref{decay} holds, hence,
	by a change of variables, there exists a constant $\tau>0$ such that
	\begin{align*}
	\mathrm{Var}(\Delta_{\bt,\bh})\leq\sum_{j=1}^d\int_{\{|x_j|>\rho/h_j\}}\Xi_{i}(\bx)^2\,d\bx\leq\frac{C}{n^{\tau}},
	\end{align*}
	where the constant $C$ is independent of $i$.
	With the same arguments as used to prove \eqref{tailbound} in the proof of Lemma \ref{Le:SumInt} we obtain
	\begin{align*}
	\mathbb{P}\biggl(\sup_{\bt\in[\mathbf{h}+\pmb{\rho},\be-\pmb{\rho}]}\Delta_{\bt,\bh}>\lambda/(\log(n)^2\log(\log(n)))\biggr)\leq C n^{-\lambda}.
	\end{align*}
	We now show that condition \eqref{S1} is satisfied. By Plancherel's Theorem and \eqref{DictConvolution} we obtain
	\begin{align*}
	\int\big|\Phi_{\bh}(\bt-\bz)-\Phi_{\bh}(\bs-\bz)\big|^2\,\mathrm{d}\bz=
	\frac{1}{(2\pi)^d}\int \big|\exp(i\langle\bt,\bxi\rangle)-\exp(i\langle\bs,\bxi\rangle\big|^2\biggl|\frac{\mathcal{F}_d\varphi(-\bxi)}{\mathcal{F}_dk(\bxi/\bh)}\biggr|^2	\,\mathrm{d}\bxi.
	\end{align*}
	This yields			
	\begin{align*}
	&(2\pi)^d\int\big|\Phi_{\bh}(\bt-\bz)-\Phi_{\bh}(\bs-\bz)\big|^2\,\mathrm{d}\bz
	\leq2^{2-2\gamma}\int \big|\langle\bt-\bs,\bxi\rangle\big|^{2\gamma}\biggl|\frac{\mathcal{F}_d\varphi(\bxi)}{\mathcal{F}_dk(\bxi/\bh)}\biggr|^2	\,\mathrm{d}\bxi.
	\end{align*}
	Observe that
	\begin{align*}
	\|\Phi_{\bh}\|_2&=\frac{1}{(2\pi)^d}\int\biggl|\frac{\mathcal{F}_d\varphi(-\bxi)}{\mathcal{F}_d\psf(\bxi/\bh)}\biggr|^2\,\mathrm{d}\bxi\geq\frac{1}{\min\{h_j|j\in I_d\}^{4a}(2\pi)^d\overline{C}}\min_{j\in I_d}\int|\mathcal{F}_d\varphi(\bxi)|^2|\xi_j|^{4a}\,\mathrm{d}\bxi,
	\end{align*}
	which implies
	\begin{align*}
	\|\Phi_{\bh}\|_2\geq\frac{C_{\Phi}}{\min\{h_j|j\in I_d\}^{4a}}>0.
	\end{align*}
	Thus
	\begin{align*}
	&\frac{(2\pi)^d}{\|\Phi_{\bh}\|_2}\int\big|\Phi_{\bh}(\bt-\bz)-\Phi_{\bh}(\bs-\bz)\big|^2\,\mathrm{d}\bz
	\leq\frac{2^{2-2\gamma}}{\|\Phi_{\bh}\|_2}\int \big|\langle\bt-\bs,\bxi\rangle\big|^{2\gamma}\biggl|\frac{\mathcal{F}_d\varphi(\bxi)}{\mathcal{F}_dk(\bxi/\bh_i)}\biggr|^2	\,\mathrm{d}\bxi\\
	&\leq\|\bt-\bs\|_2^{2\gamma}\frac{1}{\overline{c}C_{\Phi}}\int\,\|\bxi\|_2^{4a+2\gamma}|\phi(\bxi)|^2\mathrm{d}\bxi,
	\end{align*}
	and \eqref{S1} is satisfied since $\varphi\in H^{2a+\gamma\vee1/2}(\R^d)$. Now, an application of Theorem \ref{GWSMain} and Lemma \ref{Lemma:Unbounded} concludes the proof of claim (a).
	In order to show claim (b), by Taylor's theorem we deduce 	
	\begin{align*}
	&(2\pi)^d\int\big|\Phi_{\bh}(\bt-\bz)-\Phi_{\bh}(\bs-\bz)\big|^2\,\mathrm{d}\bz
	=\int \big|\langle\bt-\bs,\bxi\rangle\big|^2\biggl|\frac{\mathcal{F}_d\varphi(\bxi)}{\mathcal{F}_dk(\bxi/\bh_i)}\biggr|^2	\,\mathrm{d}\bxi+o(\|\bs-\bt\|_2^2)\\
	&=\int \biggl|\frac{\sum_{j=1}^d\xi_j(t_j-s_j)\mathcal{F}_d\partial_j\varphi(\bxi)}{\mathcal{F}_dk(\bxi/\bh_i)}\biggr|^2	\,\mathrm{d}\bxi+o(\|\bs-\bt\|_2^2)
	=(\bt-\bs)^TD_{\Xi,i}^{-2}(\bt-\bs)+o(\|\bs-\bt\|_2^2),
	\end{align*}
	where
	\begin{align*}
	D_{\Xi,i}^{-2}=\frac{1}{\|\Phi_{\bh_i}\|_2}\biggl(\int \biggl|\frac{\mathcal{F}_d\partial_k\varphi(\bxi)\mathcal{F}_d\partial_j\varphi(\bxi)}{\mathcal{F}_dk(\bxi/\bh_i)}\biggr|^2	\,\mathrm{d}\bxi\biggr)_{j,k=1}^d.
	\end{align*}		
	Hence, also the stronger condition \eqref{S2} is satisfied. Notice that by \eqref{D1}, the determinant of the matrix $D_{\Xi,i}^{-2}$ is uniformly bounded from below and above.		
	We now show the second claim of this Theorem. With the same arguments as is the proof of Theorem \ref{Thm:Gumbel} (and Lemma \ref{Lemma:Unbounded}) 
	for scales in the dyadic grid $\mathcal{H}_{\mathrm{dyad}}^d$ (and for  approximations of the $\Xi_i$ of bounded support growing logarithmically in $n$) because the bounds on the canonical metric $\rho$ in \eqref{eq:dist} remain the same up to a change in the constants by \eqref{D1}. Therefore, we can proceed precisely as in Theorem \ref{Thm:Gumbel} and find
	\begin{align*}
	\mathbb{P}(\mathcal{S}(W)\leq\lambda)\sim\exp\bigg(-e^\lambda\frac{H_2}{K\sqrt{2\pi}}\sum_{\bp\in\mathcal{P}^d}\sqrt{\mathrm{det}\big(D_{\Xi,\bp}^{-2}\big)}\log\bigg(\frac{K}{\bh_{\bp}^{\be}}\bigg)^{-d}\bigg).
	\end{align*}
	Assertion (b) of this theorem now follows from the uniform boundedness of $\mathrm{det}\big(D_{\Xi,\bp}^{-2}\big)$. 
	Claim (c) s an immediate consequence from the previous calculations and an application of Lemma \ref{Lemma:Unbounded}.
\end{proof}

\begin{proof}[Proof of Lemma \ref{Lemma:Power1}]
	We have that
	\begin{align*}
	&\mathbb{P}\bigl( \langle Y,\Phi_i\rangle_n>q_{i,1-\alpha}\;\;\forall\; i\in\mathcal{I}_{2,\alpha}\bigr)\\
	&~~=\mathbb{P}\bigl( \langle \xi\,,\,\Phi_i\rangle_n+\langle f\,,\,\phi_i\rangle+o(1/\log(n)^2)>q_{i,1-\alpha}\;\;\forall\; i\in\mathcal{I}_{2,\alpha}\bigr)\\
	&~~\geq\mathbb{P}\bigl( \langle \xi\,,\,\Phi_i\rangle_n+ o(1/\log(n)^2)>-q_{i,1-\alpha}\;\;\forall\; i\in\mathcal{I}_{2,\alpha}\bigr)
	\end{align*}
	Thus, the claim of the lemma now follows by an application of Theorem \ref{Th:1}.
\end{proof}
\begin{proof}[Proof of Theorem \ref{Thm:Conv1}]
	By \eqref{eq:tildePhi}, we see that $\mathrm{supp}\Phi_{\bh}\subset[0,1]^2$. It follows from the proof of Theorem \ref{Thm:Conv2} that \eqref{S1} holds, i.e., assertion (a) follows.\\
	Since under the assumptions of $(b)$ the assumptions of Theorem \ref{Thm:Conv2} (b) are satisfied, assertion (b) is an immediate consequence of the latter.\\
	Furthermore, by \eqref{Def:XiDec}, \eqref{eq:XiDec1}, and \eqref{eq:XiDec2} it follows that \eqref{eq:Phi_h_convergence} holds. Therefore, (c) follows from Theorem \ref{Th:1} (c) and Corollary \ref{Cor} (b).

\end{proof}
Finally, the claim of Lemma \ref{Lemma:Power3} follows by the same arguments as Lemma \ref{Lemma:Power1}, which concludes this section. 

\appendix

\section{The full width at half maximum of a convolution kernel}\label{appA}
In a deconvolution problem with convolution kernel $\psf$, the so-called full width at half maximum (FWHM) is a common standard measure for the spread of a convolution kernel in optics, see e.\,g. \citep{HW:94,p06}. There is a common understanding that objects which are closer to each other than a distance of approximately the FWHM cannot be identified as separate objects. To get a visual idea of the FWHM in our exemplary situations, we depict a convolution kernel $\psi$ from the family $\left\{\psf_{a,b} ~\big|~ a \in \mathbb N, b >0\right\}$ defined in Fourier space via \eqref{Def:PSF} with parameters $a = 2$ and $b = 0.0243$ in Figure \ref{fig:fwhm} and indicate its FWHM as well. 

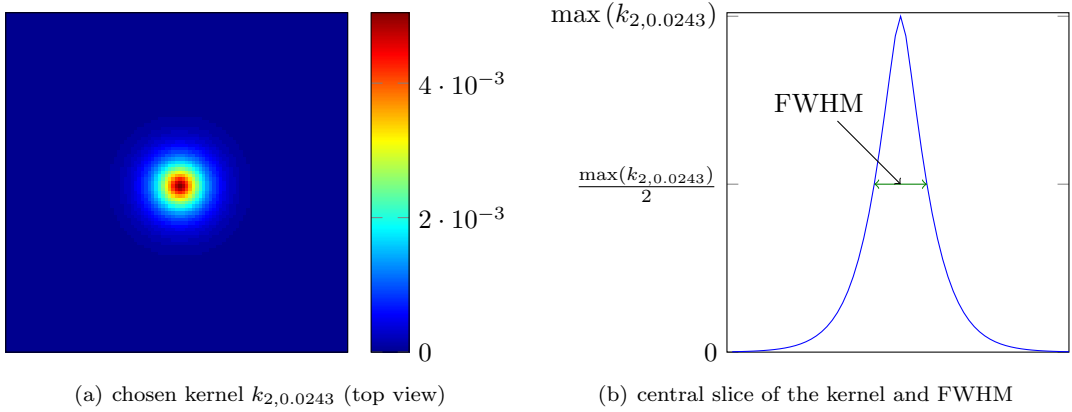
\begin{figure}[!htb]
\setlength\fheight{4.5cm} \setlength\fwidth{4.5cm}
\centering
\subfigure[chosen kernel $\psf_{2,0.0243}$ (top view)]{
\label{subfig:sim_kernel}
% This file was created by matlab2tikz v0.4.7 running on MATLAB 7.11.
% Copyright (c) 2008--2014, Nico Schlömer <nico.schloemer@gmail.com>
% All rights reserved.
% Minimal pgfplots version: 1.3
% 
% The latest updates can be retrieved from
%   http://www.mathworks.com/matlabcentral/fileexchange/22022-matlab2tikz
% where you can also make suggestions and rate matlab2tikz.
% 
\begin{tikzpicture}

\begin{axis}[%
width=\fwidth,
height=\fheight,
axis on top,
scale only axis,
xmin=0.5,
xmax=101.5,
y dir=reverse,
xtick = \empty,
ytick = \empty,
ymin=0.5,
ymax=101.5,
colormap={mymap}{[1pt] rgb(0pt)=(0,0,0.5625); rgb(7pt)=(0,0,1); rgb(23pt)=(0,1,1); rgb(39pt)=(1,1,0); rgb(55pt)=(1,0,0); rgb(63pt)=(0.5,0,0)},
colorbar,
colorbar style={
scaled ticks = false
},
point meta min=0,
point meta max=0.00504717057403539
]
\addplot [forget plot] graphics [xmin=0.5,xmax=101.5,ymin=0.5,ymax=101.5] {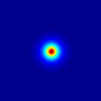};
\end{axis}
\end{tikzpicture}% 
}
\subfigure[central slice of the kernel and FWHM]{
\label{subfig:sim_kernel_FWHM}
% This file was created by matlab2tikz v0.4.7 running on MATLAB 7.11.
% Copyright (c) 2008--2014, Nico Schlömer <nico.schloemer@gmail.com>
% All rights reserved.
% Minimal pgfplots version: 1.3
% 
% The latest updates can be retrieved from
%   http://www.mathworks.com/matlabcentral/fileexchange/22022-matlab2tikz
% where you can also make suggestions and rate matlab2tikz.
% 
\begin{tikzpicture}
\begin{axis}[%
width=\fwidth,
height=\fheight,
scale only axis,
xticklabels = \empty,
scaled ticks = false,
ytick = {0, 0.00252358528701764, 0.00504717057403539},
yticklabels = {0, $\frac{\max\left(\psf_{2,0.0243}\right)}{2}$, $\max\left(\psf_{2,0.0243}\right)$},
xtick = \empty,
xmin=0,
xmax=65,
ymin=0,
ymax=0.0051
]
\addplot [color=blue,solid,forget plot]
  table[row sep=crcr]{%
1	5.83810149975011e-06\\
2	7.42747096554553e-06\\
3	9.39665848648507e-06\\
4	1.19342377424082e-05\\
5	1.5094113995671e-05\\
6	1.91399311668625e-05\\
7	2.41927503956756e-05\\
8	3.0629674305562e-05\\
9	3.86800705596594e-05\\
10	4.88926207178054e-05\\
11	6.16679197939795e-05\\
12	7.78132396202358e-05\\
13	9.79946917085141e-05\\
14	0.000123406216440661\\
15	0.000155116925440343\\
16	0.00019489114241867\\
17	0.000244393405351569\\
18	0.000306211169386346\\
19	0.000382855684144864\\
20	0.000478060210886637\\
21	0.000595467243724543\\
22	0.000740305689240437\\
23	0.000917554367682834\\
24	0.00113414589517699\\
25	0.00139616673263642\\
26	0.00171180010004986\\
27	0.00208653503684177\\
28	0.00252715156943235\\
29	0.00303210436007108\\
30	0.00359688969540957\\
31	0.00418885853697467\\
32	0.00475068555217558\\
33	0.00504717057403539\\
34	0.00475068555217558\\
35	0.00418885853697467\\
36	0.00359688969540957\\
37	0.00303210436007108\\
38	0.00252715156943235\\
39	0.00208653503684177\\
40	0.00171180010004986\\
41	0.00139616673263642\\
42	0.00113414589517699\\
43	0.000917554367682834\\
44	0.000740305689240437\\
45	0.000595467243724543\\
46	0.000478060210886637\\
47	0.000382855684144864\\
48	0.000306211169386346\\
49	0.000244393405351569\\
50	0.00019489114241867\\
51	0.000155116925440343\\
52	0.000123406216440661\\
53	9.79946917085142e-05\\
54	7.78132396202357e-05\\
55	6.16679197939795e-05\\
56	4.88926207178055e-05\\
57	3.86800705596594e-05\\
58	3.06296743055619e-05\\
59	2.41927503956756e-05\\
60	1.91399311668626e-05\\
61	1.50941139956709e-05\\
62	1.1934237742408e-05\\
63	9.39665848648491e-06\\
64	7.42747096554537e-06\\
65	5.83810149975011e-06\\
};

\coordinate (A) at (axis cs:28,0.00252358528701764);
\coordinate (B) at (axis cs:38,0.00252358528701764);

\draw [<->,color=black!50!green,solid] (A) -- (B);

\coordinate (C) at (axis cs:33,0.00252358528701764);
\node (D) at (axis cs:17.5,0.00375) {FWHM};
\draw [->,solid] (D) -- (C);

\end{axis}
\end{tikzpicture}% 
}
\caption{The function $\psf_{2,0.0243}$ and its full width at half maximum (FWHM).}
\label{fig:fwhm}
\end{figure}

\section{A mathematical model for STED microscopy}\label{appB}
STED (stimulated emission depletion) super-resolution microscopy \citep{HW:94,KH:99,H:07} allows to image samples marked by fluorescent dyes on a sub-diffraction spatial resolution. 

Just as in confocal microscopy (cf. \citep{p06} for an overview or \citep{HW:16} for the mathematical treatment) the specimen is illuminated with a diffraction-limited spot for excitation. The specimen is also irradiated with a ring-like beam distribution for inhibition. This distribution prevents fluorophores from emitting fluorescence by stimulating photon emission at a longer wavelength than the ordinary one. This red-shifted emission light can be removed by a filter and hence is not seen through the microscope. Molecules that emit photons through this channel cannot emit any photons at the usual wavelength. Consequently, the light is collected from a significantly smaller region than in standard confocal microscopy, which enhances the resolution (cf. \citep{aem15} for a more detailed description from a statistical perspective).

With this technique the specimen is imaged along a grid, where for each grid point several excitation pulses (say $t$) are applied and measured. The corresponding measurements are well described by a binomial model
\[
Y_{\bj} \stackrel{\text{independent}}{\sim} \text{Bin}\left(t, \left(\psf \ast f\right) \left(\bx_{\bj}\right)\right), \qquad \bj \in \left\{1,...,n\right\}^2.
\]
Here $\text{Bin}\left(k,p\right)$ denotes the Binomial distribution with parameters $k \in \mathbb N$ and $p \in \left[0,1\right]$, $n^2$ is the number of pixels in the grid $\bigl\{\bx_{\bj} ~\big|~ \bj \in \left\{1, ..., n\right\}^2\bigr\}$ and $f \left(\bx\right)$ is the probability that a photon emitted at grid point $\bx$ is recorded at the detector in a single excitation pulse. 

The convolution kernel or point spread function (psf) can be computed by means of scalar diffraction theory \citep{bw99} as the absolute square of the Fraunhofer diffraction pattern. In case of a circular aperture (which is the case in our experimental setup) using the paraxial approximation it simplifies to the Airy pattern
\begin{equation}\label{eq:airy}
\bx \mapsto \left\vert 2A \left(\frac{2 \pi r}{\lambda} \frac{b}{f_e} \left\Vert \bx \right\Vert_2\right)\right\vert^2
\end{equation}
where $r$ is the refractive index of the image space (here $r \approx 1$ for air), $\lambda$ is the wavelength of the incoming light, $b$ is the aperture radius of the exit lens and $f_e$ its focal length. Here $A\left(\xi\right)= J_1 \left(\xi\right) / \xi$ with the Bessel function $J_1$ of first kind. Instead of using the complicated convolution kernel in \eqref{eq:airy}, we suggest to approximate the psf by a function $\psf$ from the two parameter family \eqref{Def:PSF} by matching FWHM and kurtosis of the kernel. These values are available from measurements of the experimental psf, which yield $a = 2$ and $b = 0.016$, cf. Figure \ref{fig:exp_kernels} for details. The plots show that these two parameters provide a remarkably good matching of the kernel functions. Consequently, we believe that
\[
Y_{\bj} \stackrel{\text{independent}}{\sim} \text{Bin}\left(t, \left(\psf_{2,0.016} \ast f\right) \left(\bx_{\bj}\right)\right), \qquad \bj \in \left\{1,...,n\right\}^2.
\]
is a highly accurate model for our experimental data.

\begin{figure}[!htb]
\setlength\fheight{5cm} \setlength\fwidth{5cm}
\centering
\begin{tabular}{ll}
% This file was created by matlab2tikz v0.4.7 running on MATLAB 7.11.
% Copyright (c) 2008--2014, Nico Schlömer <nico.schloemer@gmail.com>
% All rights reserved.
% Minimal pgfplots version: 1.3
% 
% The latest updates can be retrieved from
%   http://www.mathworks.com/matlabcentral/fileexchange/22022-matlab2tikz
% where you can also make suggestions and rate matlab2tikz.
% 
\begin{tikzpicture}

\begin{axis}[%
width=\fwidth,
height=\fheight,
axis on top,
scale only axis,
xmin=0.5,
xmax=80.5,
y dir=reverse,
xtick = \empty,
ytick = \empty,
ymin=0.5,
ymax=80.5,
colormap={mymap}{[1pt] rgb(0pt)=(0,0,0.5625); rgb(7pt)=(0,0,1); rgb(23pt)=(0,1,1); rgb(39pt)=(1,1,0); rgb(55pt)=(1,0,0); rgb(63pt)=(0.5,0,0)},
colorbar,
colorbar style = {
scaled ticks = false
},
point meta min=0,
point meta max=0.00641469284880538
]
\addplot [forget plot] graphics [xmin=0.5,xmax=80.5,ymin=0.5,ymax=80.5] {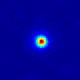};

\coordinate (A) at (370,150);
\coordinate (B) at (460,150);

\draw [color=white,solid,line width=2.0pt] (A) -- (B) node[midway,above] {\tiny 100nm};
\end{axis}
\end{tikzpicture}% 
&
% This file was created by matlab2tikz v0.4.7 running on MATLAB 7.11.
% Copyright (c) 2008--2014, Nico Schlömer <nico.schloemer@gmail.com>
% All rights reserved.
% Minimal pgfplots version: 1.3
% 
% The latest updates can be retrieved from
%   http://www.mathworks.com/matlabcentral/fileexchange/22022-matlab2tikz
% where you can also make suggestions and rate matlab2tikz.
% 
\begin{tikzpicture}

\begin{axis}[%
width=\fwidth,
height=\fheight,
axis on top,
scale only axis,
xmin=0.5,
xmax=80.5,
y dir=reverse,
xtick = \empty,
ytick = \empty,
ymin=0.5,
ymax=80.5,
colormap={mymap}{[1pt] rgb(0pt)=(0,0,0.5625); rgb(7pt)=(0,0,1); rgb(23pt)=(0,1,1); rgb(39pt)=(1,1,0); rgb(55pt)=(1,0,0); rgb(63pt)=(0.5,0,0)},
colorbar,
colorbar style = {
scaled ticks = false
},
point meta min=0,
point meta max=0.00641469284880538
]
\addplot [forget plot] graphics [xmin=0.5,xmax=80.5,ymin=0.5,ymax=80.5] {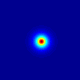};

\coordinate (A) at (370,150);
\coordinate (B) at (460,150);

\draw [color=white,solid,line width=2.0pt] (A) -- (B) node[midway,above] {\tiny 100nm};

\end{axis}
\end{tikzpicture}%
\\
\hspace*{.5cm}(a) experimental kernel (top view)
&
\hspace*{.5cm}(b) approximating kernel $\psf_{2,0.016}$ (top view)
\\[0.3cm]
\input{exp_kernel_slice.tikz} 
&
\input{approx_kernel_slice.tikz} 
\\
\hspace*{.5cm}(c) experimental kernel (slices)
&
\hspace*{.5cm}(d) approximating kernel $\psf_{2,0.016}$ (slices)
\end{tabular}
\caption{Experimentally measured kernel and function $\psf_{2,0.016}$. For the experimental kernel, we computed FWHM $= 75.9501$ nm and kurtosis $= 3.102$, and for $\psf_{a,b}$ we have FWHM $= 77.5881$ nm and kurtosis $= 3$.}
\label{fig:exp_kernels}
\end{figure}

\section*{Acknowledgements}

The authors gratefully  acknowledge   financial  support  by  the  German  Research Foundation DFG through subproject A07 of CRC 755. Funding through the VW foundation is also gratefully acknowledged. Futhermore we thank Haisen Ta and Jan Keller who are with the lab of Stefan Hell at the Department of NanoBiophotonics, Max Planck Institute for Biophysical Chemistry for providing the experimental data and expertise. We also thank two anonymous referees and the editors for a number of insightful questions and constructive comments which helped us to improve the quality of the paper substantially.

\bibliography{MISCAT_ref}

\end{document}